\documentclass[10pt, oneside]{article}

\IfFileExists{./math_headers.sty}
  {\input ./math_headers.sty}
  {\IfFileExists{../math_headers.sty}
    {\input ../math_headers.sty}
    {\typeout{Header file not found}
    }
  }

\title{Geometric Langlands Twists of $N=4$ Gauge Theory \\ from Derived Algebraic Geometry}
\author{Chris Elliott \and Philsang Yoo}
\date{}

\DeclareMathOperator{\EOM}{EOM}

\DeclareMathOperator{\PT}{\mathbb{PT}}
\DeclareMathOperator{\PN}{\mathbb{PN}}
\newcommand{\del}{\partial}
\newcommand{\delby}[1]{\frac{\del}{\del #1}}

\begin{document}
\maketitle

\begin{abstract}
We develop techniques for describing the derived moduli spaces of solutions to the equations of motion in twists of supersymmetric gauge theories as derived algebraic stacks. We introduce a holomorphic twist of $N=4$ supersymmetric gauge theory and compute the derived moduli space. We then compute the moduli spaces for the Kapustin--Witten topological twists as its further twists. The resulting spaces for the A- and B-twist are closely related to the de Rham stack of the moduli space of algebraic bundles and the de Rham moduli space of flat bundles, respectively. In particular, we find the unexpected result that the moduli spaces following a topological twist need not be entirely topological, but can continue to capture subtle algebraic structures of interest for the geometric Langlands program. 
\end{abstract}

\tableofcontents

\section{Introduction}
In this work we will attempt to make precise, using derived algebraic geometry and the classical BV formalism, the relationship between certain topological twists of $N=4$ gauge theories and the moduli stacks that occur in the geometric Langlands program as pioneered in \cite{KW}.  In this paper we'll construct the moduli spaces of solutions to the equations of motion in the Kapustin--Witten twists of $N=4$ gauge theory as shifted symplectic derived stacks, and note the appearance of interesting representation theoretic moduli spaces, with appropriate algebraic structures.

\subsection{Statement of Geometric Langlands} \label{GL_section}

Historically, the original motivation for the geometric Langlands conjecture comes from number theory: from trying to find the right analogue of the Langlands reciprocity conjecture in the realm of complex geometry. Because the objects of interest behave better in a geometric setting, one can prove stronger results in a cleaner way and hope to eventually transport some ideas from geometry to number theory. Ng\^o's proof \cite{Ngo} of the fundamental lemma using the geometry of the Hitchin system is an example of a striking success of this program (explained in an expository article of Nadler \cite{NadlerFL}).

In this subsection, we will briefly recall the heuristic categorical statement of geometric Langlands conjecture due to Beilinson and Drinfeld. For an introduction following the historical path with more detail, one can refer to Frenkel's series of lectures on the Langlands program \cite{FrenkelLectures}.

Let $G$ be a complex reductive algebraic group.  The \emph{Langlands dual} group $G^\vee$ to $G$ is the unique reductive algebraic group with dual root data: the roots and coroots are interchanged, as are the weights and coweights.  Let $\Sigma$ be a smooth projective algebraic curve over the complex numbers. The \emph{geometric Langlands conjecture} alleges an equivalence of two categories related to the geometric representation theory of $G$ and $G^\vee$ respectively on $\Sigma$.

\begin{conjecture}[``The best hope'']
There is an equivalence of dg-categories $\mr{D}(\bun_G(\Sigma)) \simeq \mr{QC}(\Flat_{G^\vee}(\Sigma))$ intertwining the natural symmetries on both sides.
\end{conjecture}

For the category on the left-hand side (the \emph{geometric}, \emph{automorphic} or `A'-side), we take the moduli stack $ \bun_G(\Sigma)$ of algebraic principal $G$-bundles on $\Sigma$, and consider the dg-category of \emph{D-modules} on $\bun_G(\Sigma)$, which we'll denote by $\mr D(\bun_G(\Sigma))$.  This category has been well studied in work of Drinfeld and Gaitsgory \cite{DrinfeldGaitsgory}.

For the category on the other side (the \emph{spectral}, \emph{Galois} or `B'-side), we consider the moduli stack $\Flat_{G^\vee}(\Sigma)$ of \emph{$G^\vee$-flat connections} on $\Sigma$.  We should be careful about what exactly we mean by this category -- we'll mean the stack parameterizing algebraic $G^\vee$-bundles with flat connection (as opposed to the character stack, say, which in particular does not depend on the complex structure of $\Sigma$).  The ``best hope'' version of the geometric Langlands conjecture uses the dg-category $\QC(\Flat_{G^\vee}(\Sigma))$ of quasi-coherent sheaves on $\Flat_{G^\vee}(\Sigma)$.

\begin{remarks}
\begin{enumerate}
\item The conjecture also includes a compatibility of the equivalence with natural symmetries on the two sides.  We're omitting the details of these symmetries for now, but they are crucial aspects of the correspondence that one must address.
\item As written above, the conjecture is known to be false as soon as $G$ is not a torus: there are objects on the A-side which are ``too large'' to correspond to anything on the B-side.  This phenomenon is already visible on the curve $\bb{CP}^1$ \cite{Lafforgue}. Arinkin--Gaitsgory \cite{ArinkinGaitsgory} formulated a form of the conjecture which is intended to correct this incompatibility by suitably enlarging the category on the B-side.
\end{enumerate}
\end{remarks}

We address the second remark in a follow-up to this paper \cite{EY2}.

\subsection{Kapustin--Witten and Geometric Langlands} \label{KW_section}
The observation that the magnetic dual group in the S-duality of Goddard, Nuyts, and Olive \cite{GNO} and Montonen and Olive \cite{MontonenOlive} recovers the notion of the Langlands dual group prompts the natural question ``is there any further relationship between S-duality and Langlands duality?''  Kapustin and Witten argue that the answer is yes: the geometric Langlands equivalence of categories is recovered as an equivalence of categories of branes along a Riemann surface in certain twists of $N=4$ gauge theories.  This paper and its sequel will attempt to mathematically understand part of this claim, so we begin by recalling the arguments of Kapustin and Witten.

We'll review some of Kapustin--Witten's main ideas here in a rather heuristic way that aims to reveal its relationship with the geometric Langlands correspondence. The first idea is to construct a family of 4-dimensional topological field theories parametrized by $\Psi \in \CC \PP^1$ as topological twists of $N=4$ supersymmetric gauge theory with gauge group $G$. S-duality interchanges the theory with gauge group $G$ and parameter $\Psi$ with the theory with gauge group $G^\vee$ and parameter $-\frac{1}{\Psi}$ \footnote{Strictly speaking this is only correct when $G$ is simply laced.  For general groups S-duality is expected to exchange parameters $\Psi$ and $-\frac 1{n_\gg \Psi}$ where $n_\gg$ is the lacing number: the ratio of the lengths of a longest and shortest root.}, identifying two manifestly different theories.  The relevant parameters for the geometric Langlands conjecture are $\Psi=0$ and $\Psi=\infty$.

The second idea is to consider the compactification of the theories along a compact Riemann surface $\Sigma$, and identify them as a family of topological sigma-models with target $\mathcal{M}_G(\Sigma)$ -- the Hitchin moduli space -- whose complex structures (and hence corresponding symplectic structures) are parametrized by $\Psi \in \bb{CP}^1$. Furthermore, at those special points $(G,0)$ and $(G^\vee,\infty)$, upon compactification, S-duality becomes mirror symmetry between the A-model with target $\mathcal{M}_G^K(\Sigma)$ in the complex structure $K$ and the B-model with target $\mathcal{M}_{G^\vee}^J(\Sigma)$ in the complex structure $J$. Since the A-model is known to only depend on the symplectic structure of the target manifold, one can identify $\mathcal{M}_G^K(\Sigma)$ with the moduli space of Higgs bundles, or $T^*\bun_G(\Sigma)$. On the other hand, one can identify $\mathcal{M}_{G^\vee}^J(\Sigma)$ with the moduli space $\Flat_{G^\vee}(\Sigma)$ of principal $G^\vee$ bundles with flat connection.

Already from these two ideas one can obtain a version of the geometric Langlands correspondence. Kapustin and Witten argued at the physics level of rigor that the relevant A-branes on $T^*\bun_G$ can be identified with D-modules on $\bun_G$, and hence a version of homological mirror symmetry would give $\mr{D}(\bun_G(\Sigma)) \simeq \QC(\Flat_{G^\vee}(\Sigma))$.  A mathematical theorem about the relationship between the Fukaya category and D-modules for real analytic manifolds is provided by Nadler--Zaslow \cite{NadlerZaslow} and Nadler \cite{NadlerMicrolocal}. This argument yields a statement that seems exactly of the form of the best hope conjecture.

Although Kapustin and Witten's argument is both beautiful and influential, it has a few mathematical defects. First of all, the geometric Langlands conjecture is an algebraic statement, whereas all the above discussion is at best analytic, for example involving A-branes and complex flat connections.  An additional argument is therefore needed to recover the actual categories studied in the geometric Langlands program.  We study the classical twisted $N=4$ theories in a more rigorous way to identify things in an algebraic category, providing a stronger algebraic version of the statement of Kapustin and Witten.  What's more, Kapustin and Witten's argument does not (immediately) provide a way to remedy the deficiencies of the best hope conjecture. In future work we will argue that a careful study of the theory incorporating a choice of quantum vacuum naturally leads to Arinkin and Gaitsgory's modified version of the conjecture.

We should say a little more about the significance of the determination of the algebraic structure on our moduli spaces.  There will, in general, be several analytically equivalent possible versions of the moduli space of solutions to the equations of motion.  The choices that appear in the geometric Langlands conjecture involve a two (real) directions in which the theory is truly topological -- we might call these \emph{Betti} directions -- and two directions in which the theory depends on a complex algebraic structure -- \emph{de Rham} directions.  For example, on the B-side geometric Langlands discusses the moduli space of flat connections on a curve as opposed to the character stack.  In the physical story we'll discuss all four directions will most naturally be de Rham, and we'll have to describe a version of the story in which two of the de Rham directions are replaced by Betti directions (as yet without a strong physical motivation) in order to obtain the moduli spaces of interest for geometric Langlands. 

\subsection{Outline of the Paper}
We begin in Section \ref{main_theory_section} by setting up the formalism for twists of supersymmetric field theories that we'll use in the rest of the paper.  We describe the $N=4$ supersymmetry algebra in four dimensions and its square-zero supercharges: the \emph{holomorphic} supercharges for which half of the translations are exact, and the \emph{topological} supercharges for which all the translations are exact.  In particular, we'll describe the A and B topological supercharges whose corresponding twists are discussed by Kapustin and Witten.  The A supercharge is approximated by a $\CC^\times$ family built from \emph{holomorphic-topological} supercharges for which three translation directions are exact.  After performing a holomorphic twist all of these supercharges admit descriptions as vector fields on a superspace of form $\CC^{2|3}$, which we'll describe, allowing us to generalize the twisted theories to classical field theories on curved manifolds.  The background on supersymmetry algebras which we refer to is reviewed in Appendix \ref{SUSY_appendix}.

We proceed by defining classical field theories, both locally and globally, in the language of derived algebraic geometry: in particular it will make sense to talk about theories that depend on an algebraic structure on spacetime. We discuss what it means to \emph{twist} a classical field theory by an action of the supergroup $\CC^\times \ltimes \Pi \CC$: examples of such twisting data arise naturally from square-zero supercharges in a supersymmetric field theory.  Twists of non-perturbative field theories are defined as one-parameter deformations that are compatible with the perturbative twists described by Costello \cite{CostelloSH} when we restrict to the tangent complex.  There are natural constructions of twists using results of Gaitsgory and Rozenblyum that identify derived stacks with formal maps from a base derived stack $\mc X$ with Lie algebroids on $\mc X$. 

In Section \ref{constructing_theories_section} we review the main constructions of $N=4$ supersymmetric gauge theories.  We begin by introducing the language of compactification and (informally) dimensional reduction for classical field theories.  The first construction is sketched at a lower level of rigor: dimensional reduction from $N=1$ super Yang--Mills theory on $\RR^{10}$.  More rigorous is the construction by compactification from holomorphic Chern--Simons theory on $N=4$ twistor space, although there are still subtleties stemming from the non-holomorphicity of the relevant twistor map.  We review some background from twistor theory, and then prove that the linearised BV complex in holomorphic Chern--Simons yields the linearised BV complex of $N=4$ anti-self-dual super Yang--Mills theory under compactification.

Our main results appear in Section \ref{classical_states_section}, where we compute the holomorphic, B- and A-twists of $N=4$ super Yang--Mills theory as an assignment of derived stacks, beginning from the twistor space perspective. Here we compute the holomorphic twist first, because it is the minimal twist which admits an algebraic description, and realize the B- and A-twists as further twists. We find the following:

\begin{theorem}
The moduli space of solutions to the equations of motion in the holomorphic twist of $N=4$ super Yang--Mills on a smooth proper complex algebraic surface $X$ is equivalent to 
\[\EOM_{ \mr{hol} }(X) \iso  T^*_{\mr{form}}[-1] \higgs^{\mr{fer}}_G(X) \]
as a $(-1)$-shifted symplectic derived stack.
\end{theorem}

\begin{theorem}
\begin{itemize}
\item The moduli space of solutions to the equations of motion in the A-twist of $N=4$ super Yang--Mills on a smooth proper complex algebraic surface $X$ is equivalent to 
\[\EOM_A(X) \iso  \higgs_G(X)_{\mr {dR} }\]
as a $(-1)$-shifted symplectic derived stack.
\item The moduli space of solutions to the equations of motion in the B-twist of $N=4$ super Yang--Mills on a smooth proper complex algebraic surface $X$ is equivalent to 
\[\EOM_B(X) \iso T^*_{\mr{form}}[-1] \Flat_G(X)  \]
as a $(-1)$-shifted symplectic derived stack.
\end{itemize}
\end{theorem}

We further proceed to identify the moduli spaces of solutions for a surface of the form $X=C \times \Sigma$ for smooth proper curves $C,\Sigma$ which exhibits the relevance to the geometric Langlands correspondence. On the other hand, even from this description from a 4-dimensional perspective, one can make the following observation. Since it is realized as a deformation of the Higgs moduli space, it is clear that the B-twist sees the algebraic structure of the de Rham moduli space $\Flat_G(\Sigma)$ of flat $G$-bundles, as opposed to the Betti moduli space of locally constant sheaves. Also, on the A-side, we have the appearance of the de Rham stack, which explains the appearance of $D$-modules: in particular, one doesn't have to deal with a Fukaya-type category to see the categorical geometric Langlands correspondence from this perspective.

These two aspects are the results that Kapustin and Witten don't find from their analysis but which play a crucial role for us to recover the refined conjecture of Arinkin and Gaitsgory in our subsequent work \cite{EY2}. That physics can capture subtle algebraic dependence on a spacetime after a twist was not at all expected: first of all, supersymmetric gauge theory is defined purely in the realm of the smooth category.  Moreover, a topological twist was expected to yield a topological field theory in the sense of Atiyah and Segal. Indeed, this was one of the crucial reasons why many experts in the geometric Langlands program were skeptical of how much new insight the work of Kapustin and Witten could bring in to the original program. Our research makes the first explicit bridge between these two different ways of thinking, by setting up a framework based on the recent rigorous development of derived algebraic geometry (as in, for instance the work of To\"en--Vezzosi \cite{HAGI, HAGII}, Lurie \cite{HigherAlgebra}, Pantev--To\"en--Vaqui\'e--Vezzosi \cite{PTVV} and Gaitsgory--Rozenblyum \cite{GR}). 

\subsection{Conventions}
Throughout this paper we'll work with $(\infty,1)$-categories, where between two objects one has a topological space -- or a simplicial set -- of morphisms. We won't use any model-dependent arguments, but to be concrete one may consider the formulation in terms of quasi-categories, which is most extensively developed by Lurie \cite{LurieHTT}. Henceforth, we will usually just say category when we mean an $(\infty,1)$-category, use the word functor to mean a functor of $(\infty,1)$-categories, and a limit for a limit in $(\infty,1)$-categories, and so on, unless otherwise specified. As is usual in the subject, there are a lot of technicalities which must be stated in order to make subtle arguments, most of which we will omit when possible for simplicity.

Also throughout the paper we'll work over the complex number field $\CC$, although most of the formal arguments would proceed under more relaxed hypotheses.

Our main background language is derived algebraic geometry for which we don't offer an extensive exposition. This is justified partially because our arguments are mainly formal, not using any deep result of algebro-geometric content, and also because there are a few great references, for instance due to Gaitsgory \cite{GaitsgoryQC} \cite{GaitsgoryStacks} and To\"en \cite{ToenOverview} \cite{ToenSurvey}. For the reader's convenience, in the appendix we provide a summary of aspects of formal derived algebraic geometry that we take advantage of throughout.

\vspace{-8pt}
\begin{itemize}
\item By a (super) \emph{cdga} $R$ we'll always mean a (super) commutative differential graded algebra over $\CC$. We denote the category of such by $\mr{cdga}$. We also consider the functor (of ordinary categories) $(-)^\natural \colon \mr{cdga} \rightarrow \mr{cdga}$ by $R \mapsto R^\natural$, where $R^\natural$ is the underlying graded commutative algebra obtained after forgetting the differential. We use cohomological grading with respect to which we introduce the full subcategory $\mr{cdga}^{\leq 0} \subset \mr{cdga} $ of cdgas whose cohomology is concentrated in non-positive degrees. We denote the opposite category to $\cdga^{\le 0}$ by $\mr{dAff}$, the category of affine derived schemes, considering an object $R \in \mr{cdga}^{\leq 0}$ as the ring of functions on the space $\spec R$. In particular, a classical affine scheme is an affine derived scheme.

\item By a \emph{derived scheme}, we mean a ringed space $(X,\OO_X)$ where $\OO_X$ is a sheaf valued in $\mr{cdga}^{\leq 0}$ such that $(X, H^0(\OO_X) )$ is a classical scheme and $H^i(\OO_X)$ is a quasi-coherent sheaf over the scheme $(X,H^0(\OO_X))$. By definition, a scheme or an affine derived scheme forms a derived scheme in an obvious manner and a derived scheme yields a classical scheme as its classical truncation $X^{\mr{cl}} := (X,H^0(\OO_X))$. Note that an affine derived scheme could have been defined to be a derived scheme whose classical truncation is an affine scheme. We call the category of derived schemes $\mr{dSch}$. 

\item A \emph{prestack} $\mc X$ is a functor
 \[\mc X \colon \cdga^{\le 0} \to \sset,\]
 where $\sset$ is the category of simplicial sets. A \emph{derived stack} is a prestack satisfying a descent condition with respect to the \'etale topology and we denote the category of derived stacks by $\mr{dSt}$. In particular, any simplicial set provides a constant derived stack, and any derived scheme defines a derived stack by its functor of points.  That is, if $X$ is a derived scheme we define the corresponding derived stack whose $R$-points are the simplicial set whose $i$-simplices are $\hom_{\mr{dSch}}(\spec (R \otimes \Omega^\bullet_{\mr{alg}}(\Delta^i)),X)$, where $\Omega^\bullet_{\mr{alg}}(\Delta^i)$ is the ring of algebraic de Rham forms on the standard $i$-simplex $\Delta^i$.  The reduced part $\mc X^{\mr{red}}$ of a prestack $\mc X$ is the functor $\cring^{\mr{red}} \to \sset$ from reduced commutative rings obtained by the restriction along the functor $\cring^{\mr{red}} \to \mr{cdga}^{\leq 0}$.
 
\item A derived stack is a \emph{derived $0$-Artin stack} if it is an affine derived scheme. A derived stack is a \emph{derived $n$-Artin stack} if it is realized as a colimit over a smooth groupoid of derived $(n-1)$-Artin stacks. A derived stack is called a \emph{derived Artin stack} if it is a derived $n$-Artin stack for some $n$. For arguments involving shifted symplectic structures we'll need to restrict attention to derived Artin stacks which are locally of finite presentation.  This ensures that the cotangent complex is perfect, hence dualizable.  

\item For any two derived stacks $\mc X $, $\mc Y$, one can define the \emph{mapping stack} $\underline{\mr{Map}} (\mc X,\mc Y) \colon \mr{dAff}^{\mr{op}} \rightarrow \sset $ by $U \mapsto \mr{Map}_{\mr{dSt}}(\mc X \times U, Y)$. As an example of a mapping stack, one defines the $k$-shifted tangent space $T[k] \mc X$ of $\mc X$ to be $T[k]\mc X := \underline{\mr{Map}} (\spec \mathbb{C}[\eps] , \mc X )$, where $\eps$ is a parameter of cohomological degree $-k$ with $\eps^2=0$.  As another example, we define the loop space $\LL X: = \underline{\mr{Map}}(S^1_B , X)$, where the \emph{Betti circle} $S^1_B$ is the simplicial set $S^1$ understood as a derived stack. 

\item For a derived stack $\mc X$, one defines its category $\QC(\mathcal{X})$ of \emph{quasi-coherent sheaves} as the limit 
\[\QC(\mathcal{X}) : = \displaystyle\lim_{U \in ( \mr{dAff}_{/\mathcal{X}})^{\mr{op}} } \QC(U)\] 
over the opposite category ($\mr{dAff}_{/\mathcal{X}})^{\mr{op}}$ of the category of affine derived schemes over $\mathcal{X}$, where $\QC(\spec R)$ is defined to be the category $R\text{-mod}$ of dg modules over $R$. Similarly, one defines the category $\perf(\mc X )$ of perfect complexes using finitely generated dg-modules, and the category $\coh(\mc X)$ of coherent sheaves using bounded complexes with coherent cohomology.  Finally, one defines the category $\indcoh(\mc X)$ of ind-coherent sheaves on $\mc X$ as the ind-completion of the category $\coh(\mc X)$. 

\item Every derived Artin stack $\mc X$ admits a \emph{cotangent complex} $\mathbb{L}_{\mc X} \in \QC(\mc X )$ \cite{HAGII}[2.2.3.3]. Since $\mc X$ is assumed to be of locally finite presentation, $\bb L_{\mc X}$ is a perfect complex and hence dualizable, allowing one to define the \emph{tangent complex} $\mathbb{T}_{\mc X} : = \mathbb{L}_{\mc X}^*$.  We can recover this tangent complex from the previously defined notion of the tangent space $T[k]\mc X$ \cite{HAGII}[1.4.1.9].  The shifted tangent complex $\mathbb{T}_{\mc X}[k]$ is obtained as the limit of the objects $T[k]\mc X \times_{\mc X} U$ over all $U \in \mr{dAff}_{/\mc X}$, each of which is affine and finitely generated over $U$ so lies in $\perf(U)$, and therefore the limit defines an object in $\perf (\mc X)$. One can then define the $k$-shifted cotangent stack as the relative spectrum $T^*[k]\mc X :=  \underline{\mr{Spec}}_{\mc X }(\mr{Sym}(\mathbb{T}_{\mc X}[-k] ) )$.

\item For a prestack $\mc X$, we define its \emph{de Rham prestack} $\mc X_{\mr{dR}}$ to be the functor $R \mapsto \mc X(R^{\mr{red}} )$. For a map $\mc X \to \mc Y$ of prestacks, we introduce the \emph{formal completion} $\mc Y_{\mc X}^\wedge$ of $\mc Y$ along $\mc X$ defined by $\mc Y_{\mc X} ^\wedge := \mc X_{\mr{dR}} \times_{\mc Y_{\mr{dR}} } \mc Y$. Note that one recovers the usual notion when $\mc X \rightarrow \mc Y$ is a closed immersion of ordinary schemes, justifying the name. If $\mc Y = \mr{pt}$, then one obtains the de Rham prestack $\mc X_{\mr{dR}}$. If $\mc Y = T^*[k] \mc X$ is the $k$-shifted cotangent stack, then we set $T^*_{\mr{form}}[k] \mc X : = (T^*[k] \mc X)_{\mc X }^\wedge$ for the formal neighborhood of $\mc X$ inside $T^* [k] \mc X$. 

\item A \emph{inf-scheme} \cite[Chapter 2.3]{GR} is a prestack $\mc X$ whose reduced part $\mc X^{\mr{red}}$ is a reduced scheme, and which \emph{admits deformation theory} in the sense of \cite[Chapter 1.7]{GR} (in particular derived Artin stacks locally of finite presentation admit deformation theory).  A morphism $\mc X \to \mc Y$ of prestacks is \emph{inf-schematic} if the base change $\mc X \times_{\mc Y} \spec R$ by any affine derived scheme is an inf-scheme. For instance, any map of prestacks $\mc X \to \mc Y$ induces an inf-schematic map $\mc X \to \mc Y_{\mc X}^\wedge$.
\end{itemize}

\subsection*{Acknowledgements}
We would like to thank Kevin Costello for many productive discussions throughout this work, and for teaching us so much about classical and quantum field theory.  We would also like to thank David Ben-Zvi, Damien Calaque, Saul Glasman, Ryan Grady, Michael Gr\"ochenig, Owen Gwilliam, Aron Heleodoro, Theo Johnson-Freyd, Kobi Kremnitzer, David Nadler, Tony Pantev, Mihnea Popa, Nick Rozenblyum, Pavel Safronov, Claudia Scheimbauer, Brian Williams, Edward Witten, and Peng Zhou for many helpful discussions, and especially Thel Seraphim for a great number of ideas and insightful comments. We are also grateful to Richard Hughes and Anton Zeitlin for their comments and corrections to an earlier version.

This research was supported in part by Perimeter Institute for Theoretical Physics. Research at Perimeter Institute is supported by the Government of Canada through Industry Canada and by the Province of Ontario through the Ministry of Economic Development \& Innovation.

\section{Classical $N=4$ Theories and their Twists} \label{main_theory_section}
In this section we'll discuss the foundational constructions of supersymmetric gauge theories, and the general formalism of ``twisting'' supersymmetric theories.  For simplicity, from Section \ref{supercharge_section} onwards we'll stick to considering 4-dimensional theories in Riemannian signature, but many of the constructions we discuss (particularly those purely algebraic constructions involving supersymmetry algebras) have natural analogues in other dimensions.  For instance, the construction of $N=4$ supersymmetric gauge theories in four-dimensions by dimensional reduction fits into a natural family of constructions using the theory of normed division algebras.  This is beautifully explained by Anastasiou, Borsten et al \cite{ABDHN}.  Throughout this section we'll refer to Appendix \ref{SUSY_appendix} for general constructions with supersymmetry algebras.

\subsection{Holomorphic and Topological Twists} \label{twist_susy_algebra_section}
The idea of a \emph{twist} of a supersymmetry algebra, or of a supersymmetric field theory, originated in \cite{WittenTQFT} as a procedure for constructing topological ``sectors'' of general supersymmetric field theories, but one can make sense of twists in much greater generality. One can form a twist of a supersymmetry algebra $\mc A$ -- and a twist of a theory on which it acts -- from any supercharge $Q$ (i.e. fermionic element of the supersymmetry algebra) such that $[Q,Q] = 0$.  The definition of the twisted supersymmetry algebra is straightforward.  

Let $\mc A$ be the complexified supersymmetry algebra in dimension $n$ associated to a spinorial complex representation $\Sigma$ of $\spin(n)$, a non-degenerate pairing $\Gamma \colon \Sigma \otimes \Sigma \to V_\CC$ where $V_\CC$ is the $n$-dimensional vector representation, and a subalgebra $\gg_R$ of R-symmetries.  The example that we'll be most concerned with is the 4d supersymmetry algebra associated to a finite-dimensional complex vector space $W$, given by
\[\mc A^W = (\so(4;\CC) \ltimes \CC^4) \oplus \gg_R \oplus \Pi((S_+ \otimes W) \oplus (S_- \otimes W^*))\]
where $\gg_R = \sl(4;\CC)$, as described in Appendix \ref{SUSY_appendix}.

\begin{definition}
The \emph{twisted supersymmetry algebra} associated to a fermionic element $Q \in \mc A$ with $[Q,Q] = 0$ is the cohomology of $\mc A$ with respect to the differential $[Q,-]$.
\end{definition}

A more subtle notion is that of a twist of a supersymmetric field theory, which should be thought of as the derived $Q$-invariants of the original theory, admitting an action of the twisted supersymmetry algebra.  Such twisted theories inherit properties (invariance under certain natural symmetries) from properties of the supercharge $Q$.  We'll discuss two such properties: topological and holomorphic invariance. 

Perhaps the most important types of twist are \emph{topological twists}.  In the literature, these are defined as coming from supercharges $Q \in \Pi (\Sigma)$ which are $\mr{Spin}(n)$-invariant.  Of course, there are generally no such $Q$; for instance in 4 dimensions the odd part of the $N=1$ supersymmetry algebra decomposes as a sum of irreducible two-complex dimensional $\spin(4)$-representations.  However, it suffices to find $Q$ that is $\spin(n)$-invariant \emph{after} modifying the action of the complexified rotations $\so(n;\CC)$ on the space of supercharges.  Let's make this more precise by first giving a more natural definition, then showing why the above notion implies the more natural condition.
\begin{definition}
A supercharge $Q$ with $[Q,Q]=0$ is called \emph{topological} if the map
\[[Q,-] \colon \Sigma \to V_\CC\]
is surjective. 
\end{definition}

\begin{remark}
The above definition also makes sense for theories with an action of an uncomplexified supersymmetry algebra.  A real supercharge $Q$ is likewise called topological if the map $[Q,-]$ is surjective onto the space $V_\RR$ of real translations.
\end{remark}

We'll see shortly that this implies that all translations act trivially on $Q$-twisted theories for a topological supercharge $Q$. Now, let's recover the classical notion of a topological twist.  If $\phi \colon \so(n;\CC) \to \gg_R$ is a Lie algebra homomorphism, we can define a $\phi$-twisted action of $\so(n;\CC)$ on $\Sigma$.  Indeed, $\Sigma$ always takes the form $S \otimes W$ (in odd dimensions), or $S_+ \otimes W_1 \oplus S_- \otimes W_2$ (in even dimensions) where $W, W_1$ and $W_2$ are finite-dimensional vector spaces acted on by the R-symmetries.  With this in mind we define the twisted action of $X \in \so(n;\CC)$ by
\begin{align*}
&X(s \otimes w) = X(s) \otimes \phi(X)(w) \\
\text{or } &X(s_+ \otimes w_1 + s_- \otimes w_2) = X(s_+) \otimes \phi(X)(w_1) + X(s_-) \otimes \phi(X)^*(w_2) 
\end{align*}
depending on the dimension.

\begin{prop} \label{top_twist_condition}
Let $Q$ be a non-zero supercharge in $n$ dimensions such that $[Q,Q]=0$, and such that there exists a homomorphism $\phi \colon \so(n;\CC) \to \gg_R$ making $Q$ invariant under the $\phi$-twisted action of $\so(n;\CC)$.  Then $Q$ is topological.
\end{prop}

\begin{proof}
We can replace the supersymmetry algebra with the supersymmetry algebra \emph{twisted by $\phi$}, with brackets modified as follows:
\vspace{-10pt}
\begin{itemize}
 \item The rotations $\so(n;\CC)$ act on $\Sigma$ according to the $\phi$-twisted action.
 \item Rotations bracket with elements of $\gg_R$ as their image under the embedding $\phi$.
\end{itemize}
\vspace{-10pt}
The bracket of two odd elements is unchanged, so it suffices to check that $Q$ is topological in this twisted algebra.  In this algebra, since $Q$ spans an irreducible $\so(n;\CC)$ representation, the image of $[Q,-]$ in $V_\CC$ should be itself an irreducible subrepresentation, so either 0 or $V_\CC$ itself.  Since the pairing $\Gamma$ is non-degenerate, the map $[Q,-]$ is never 0 when $Q \ne 0$, so its image is all of $V_\CC$ as required.
\end{proof}

\begin{remark}
The converse to this proposition is false in general.  For a counterexample, we consider the case of the $N=1$ supersymmetry algebra in dimension $n=8$, where the positive helicity Weyl spinor representation is related to the vector representation by triality (i.e. by precomposing by an outer automorphism of $\so(8;\CC)$).  The R-symmetry group is just $\CC^\times$, so twisting homomorphisms are just characters, and we observe that there are no non-zero invariant vectors for the vector representation of $\so(8;\CC)$ twisted by a character, and similarly for the twisted Weyl spinor representation.  However, there \emph{are} topological supercharges in the positive Weyl spinor representation in dimension 8.  In dimension 8 any Weyl spinor $Q_+$ pairs with itself to 0 under the $\Gamma$-pairing, and if $Q_+$ is not \emph{pure} -- i.e. if its nullspace in $\CC^8$ under Clifford multiplication is not of dimension 4 -- then the map $\Gamma(Q_+, -) \colon S_{8-} \to \CC^8$ is surjective. 
\end{remark}

\begin{remark}
In dimension 4 -- the case we'll principally be interested in in this paper -- there is a classification of twisting homomorphisms $\phi$ that yield topological twists by this procedure \cite{Lozano}.  We'll investigate twists coming from the so-called ``Kapustin--Witten'' twisting homomorphism, which we'll define at the beginning of the next subsection.
\end{remark}

The notion of a topological twist suggests a natural definition for a \emph{holomorphic} twist.  We should ask the image of the map $[Q,-]$ from the odd to the even part of the supersymmetry algebra to contain \emph{exactly half} of all translations.  In order for this to make sense, suppose $n$ is even.
\begin{definition}
A supercharge $Q$ with $[Q,Q]=0$ is called \emph{holomorphic} if there exists a $\CC$-linear isomorphism between $V_\CC$ and $\CC^{n/2} \otimes_\RR \CC$ such that the image of $[Q,-]$ in $V_\CC$ spans the holomorphic subspace $\RR^{n/2} \otimes_\RR \CC$.  
\end{definition}
To put it another way, $Q$ is holomorphic if we can choose a splitting of the algebra of translations into holomorphic and anti-holomorphic directions such that the image of $[Q,-]$ is precisely the anti-holomorphic piece.  There's a natural procedure for constructing holomorphic twists analogous to the procedure for topological twists above, which is straightforward to describe in four dimensions.  The procedure depends on a choice of embedding $\SU(2) \to \SU(2)_+ \times \SU(2)_-$, or on the level of complexified Lie algebras $\so(3;\CC) \to \so(4;\CC)$.  This defines an action of $\SU(2)$ on $V_\CC$ by restricting the tensor product action on $S_+ \otimes S_-$, and thus a subspace of $V_\CC$ by taking invariant vectors.  We want this to give a real subspace (i.e. a half-dimensional subspace), so we must restrict attention to the inclusions $\iota_1$ and $\iota_2$ of the two factors.

\begin{prop} \label{holo_twist_condition}
Let $Q$ be a non-zero supercharge $Q$ with $[Q,Q]=0$, and suppose there exists a homomorphism $\phi \colon \so(4;\CC) \to \gg_R$ making $Q$ invariant under the $\phi$-twisted action of $\iota_i(\so(3;\CC))$, where $i = 1$ or 2.  Then $Q$ is either a holomorphic or a topological twist.
\end{prop}

\begin{proof}
This is very similar to the proof of Proposition \ref{top_twist_condition} above. Again we can replace the supersymmetry algebra by its $\phi$-twisted version, but now the image of $[Q,-]$ in the translations is a $\iota_i(\so(3;\CC))$-subrepresentation of $V_\CC$.  As a module for this algebra $V_\CC$ decomposes as the sum of two two-dimensional irreducible representations.  Thus the image of $[Q,-]$ is zero, half-dimensional, or full-dimensional.  As before, non-degeneracy of $\Gamma$ ensures that it's non-zero, so $Q$ is either holomorphic or topological.
\end{proof}

\subsubsection{Twists of the $N=4$ Supersymmetry Algebra} \label{supercharge_section}
For most of the rest of this paper, we'll specialize to the 4-dimensional setting and the case where $W = \CC^4$, i.e. to $N=4$ supersymmetry.  We'll take the R-symmetry algebra to be $\gg_R = \sl(4;\CC) \sub \gl(4;\CC)$; this is the R-symmetry algebra that'll act on supersymmetric gauge theories, since the theories we'll define will require fixing a choice of trivialization of $\det\CC^4$.  We'll consider several holomorphic and topological twists of an $N=4$ supersymmetric gauge theory, so let's discuss these twists at the level of the supersymmetry algebra
\[\mc A^{N=4} = \left( \so(4;\CC) \oplus \gg_R \oplus V_\CC \right) \oplus \Pi \left( S_+ \otimes W \oplus S_- \otimes W^* \right)\]
where $W = \CC^4$, and where $\gg_R$ acts on $W$ by its fundamental representation.

We'll first analyse a family of holomorphic twists of this supersymmetry algebra.  We'll fix a particular twisting homomorphism $\phi$, the \emph{Kapustin--Witten twist}, defined to be the composite
\[\phi_{\text{KW}} \colon \so(4;\CC) \iso \sl(2;\CC) \oplus \sl(2;\CC) \to \sl(4;\CC)\]
where the first map is the exceptional isomorphism in dimension 4, and the second map is the block diagonal embedding. We'll get a space of holomorphic supercharges for each factor of $\SU(2)_+ \times \SU(2)_-$, which we'll describe concretely.  Choose a $\CC$-basis for the space of supercharges by choosing bases for its constituent pieces as follows:
\begin{align*}
S_+ &= \langle \alpha_1 , \alpha_2 \rangle \\
S_- &= \langle \alpha_1^\vee, \alpha_2^\vee \rangle \\
W   &= \langle e_1, e_2, f_1, f_2 \rangle \\
W^* &= \langle e_1^*, e_2^*, f_1^*, f_2^* \rangle 
\end{align*}
where $\so(4;\CC)$ acts on $W$ via the $\phi_{KW}$-twist so that $\{e_i\}$ and $\{f_i\}$ are bases for the two semispin factors (i.e. the summands on which $\SU(2)_+$ and $\SU(2)_-$ act), and where the basis given for $W^*$ is the dual basis to the one for $W$.  Tensor products of basis elements yield a basis for $S_+ \otimes W \oplus S_- \otimes W^*$.  Consider the embedding $\iota_2 \colon \SU(2) \to \SU(2)_+ \times \SU(2)-$ by inclusion of the second factor.  The resulting invariant supercharges are those in $S_+ \otimes \langle e_1, e_2 \rangle$.  From now on we'll fix a reference holomorphic supercharge
\[Q_{\text{hol}} = \alpha_1 \otimes e_1.\]

Now, let's compute the $Q_{\text{hol}}$-cohomology of the $N=4$ supersymmetry algebra; that is, the cohomology of the cochain complex
\[\xymatrix{
\so(4;\CC) \oplus \gg_R \ar[r]^(.38){[Q_{\text{hol}},-]} &\Pi(S_+ \otimes W \oplus S_- \otimes W^*)\ar[r]^(.75){[Q_{\text{hol}},-]} &V_\CC
}.\]
Consider the terms sequentially.  
\vspace{-9pt}
\begin{itemize}
 \item In the translation term we expect to find a half-dimensional family of ``anti-holomorphic'' translations as the cokernel of $[Q_{\text{hol}},-]$.  Indeed, the image in the translations is the span of $\Gamma(\alpha_1, \alpha_1^\vee)$ and $\Gamma(\alpha_1, \alpha_2^\vee)$, which are linearly independent.  From now on we'll work in coordinates on $V_\CC$ defined by
 \[ \frac \dd{\dd \ol z_i} = \Gamma(\alpha_1, \alpha_i^\vee), \frac \dd{\dd z_i} = \Gamma(\alpha_2, \alpha_i^\vee).\]
 \item In the remaining bosonic term, the kernel of $[Q_{\text{hol}},-]$ is spanned by the subgroup of elements $(A,R) \in \so(4;\CC) \oplus \sl(4;\CC)$ such that $A(\alpha_1) = c\alpha_1$ and $R(e_1) = -ce_1$ for some $c \in \CC$. This subgroup is isomorphic to $\sl(2;\CC)_- \oplus \mf p$ where $\mf p$ is isomorphic to a parabolic subalgebra of $\sl(4;\CC)$ with Levi subalgebra $\sl(3;\CC)$.
 \item In the fermionic term, consider the two summands separately.  First look at $S_+ \otimes W$.  These elements are all $[Q_{\text{hol}},-]$-closed, and the exact elements are just the five-dimensional subspace generated by $S_+ \otimes \langle e_1 \rangle$ and $\langle \alpha_1 \rangle \otimes W$, leaving
 \[\langle \alpha_2 \otimes e_2, \alpha_2 \otimes f_1, \alpha_2 \otimes f_2\rangle\]
 as the cohomology.  Finally, look at $S_- \otimes W^*$.  There are no exact elements in this subspace, and the closed elements are given by
 \[S_- \otimes \langle e_2^*, f_1^*, f_2^* \rangle.\]
\end{itemize}
So overall, the twisted supersymmetry algebra has form
\[\left(\so(3;\CC) \oplus \mf p \oplus \left\langle \frac \dd{\dd z_1},\frac \dd{\dd z_2} \right\rangle \right) \oplus \Pi\bigg(\langle \alpha_2 \otimes e_2, \alpha_2 \otimes f_1, \alpha_2 \otimes f_2\rangle \oplus S_- \otimes \langle e_2^*, f_1^*, f_2^* \rangle \bigg)\]
where $\so(3;\CC)$ acts on $S_-$ by its spin representation, and $\sl(3;\CC) \sub \mf p$ acts on $\langle e_2, f_1, f_2 \rangle$ and its dual space by the fundamental and anti-fundamental representations respectively.

Now, the twists we'll really be concerned with will all be \emph{further twists} of such a holomorphic twist.  That is, they'll be determined by supercharges $Q = Q_{\text{hol}} + Q'$ where $Q'$ commutes with $Q_{\text{hol}}$ but is not obtained from $Q_{\text{hol}}$ by the action of some symmetry, so survives in the $Q_{\text{hol}}$ twist.  All such supercharges are holomorphic or stronger (i.e. at least half the translations are $Q$-exact); indeed, the image of $[Q,-]$ in $V_\CC$ contains the image of $[Q_{\text{hol}},-]$. 

\begin{remark} \label{remark_successive_twisting}
For our further twists, we have an isomorphism $H^\bullet(\mc A^{N=4}; Q_{\text{hol}} + Q') \iso H^\bullet(H^\bullet(\mc A^{N=4}; Q_{\text{hol}});Q')$.  This is clear for $Q'$ contained entirely in the $S_-$ summand of space of supersymmetries, this follows from the degeneration of the spectral sequence of the double complex for $\mc A^{N=4}$ where $S_+$ is placed in bidegree $(1,0)$ and $S_-$ is placed in bidegree $(0,1)$.  If instead $Q'$ is contained entirely in the $S_+$ summand, the complexes $(\mc A^{N=4}, Q_{\mr{hol}} + Q')$ and $(H^\bullet(\mc A^{N=4}, Q_{\mr{hol}}), Q')$ in degrees 0, 1 and 2 split as the sum of two two-step complexes.  The claim follows for further twists of form $Q' = \alpha_2 \otimes w$ where $w \in W$ by examining the cohomology of each of these two-step complexes.
\end{remark}

We'll investigate which such supercharges $Q$ are topological.  First let's describe those supercharges which are compatible with the twisting homomorphism $\phi_{KW}$ above independently of the holomorphic twist.  After turning on the twisting homomorphism we can identify $S_+ \otimes W \oplus S_- \otimes W^*$ with
\[S_+ \otimes (S_{+R} \oplus S_{-R}) \oplus S_- \otimes (S_{+R} \oplus S_{-R})\]
where $S_{+R}$ is spanned by $e_i$ in the first summand and $e_i^*$ in the second summand and likewise $S_{-R}$ is spanned by $f_i$ and $f_i^*$.  The compatible topological supercharges, i.e. those elements that are invariant for the twisted $\SO(4)$ action, have the form
\[\alpha(\alpha_1 \otimes e_1 - \alpha_2 \otimes e_2) + \beta(\alpha_1^\vee \otimes f_1^* - \alpha_2^\vee \otimes f_2^*)\]
for any pair $(\alpha, \beta)$ of complex numbers.  These supercharges are further twists of $Q_{\mr{hol}}$ exactly when $\alpha = 1$.

This $\CC$-family extends to a natural $\bb{CP}^1$-family of topological further twists of the holomorphic twist, with a special point at infinity which is no longer compatible with the full Kapustin--Witten twisting homomorphism.  Consider the $\bb{CP}^1$-family of supercharges
\[Q_{(\mu:\nu )} = Q_{\text{hol}} + (\mu(\alpha_1^\vee \otimes f_1^* - \alpha_2^\vee \otimes f_2^*) + \nu(\alpha_2 \otimes e_2) ), \text{ for } (\mu:\nu) \in \bb{CP}^1.\]
Here simultaneous rescaling of the two-parameter family over $(\mu,\nu) \in \CC^2$ does not change the affect the further twist obtained from an element of the $Q_{\mr{hol}}$-cohomology of the supersymmetry algebra.  These twists are all compatible with the twisting homomorphism $\phi_{KW}$ with the exception of the special point $(1:0)$.  One can verify that this twist is compatible not with a twisting homomorphism from the full group $\mr{Spin}(4)$ of rotations, but only with subgroups thereof, for example with twisting homomorphisms from the group $\mr{U}(2)$.  It is therefore only natural to describe the associated twisted theory on oriented 4-manifolds that admit a complex structure.   We call this family the \emph{Kapustin--Witten family of topological twists}.  We'll be most interested in the cases where $(\mu:\nu) = (0:1)$ and $(1:0)$.  We call these twists the \emph{A-twist} $Q_A$ and the \emph{B-twist} $Q_B$ respectively. 

We note that as a result of considering only further twists of $Q_{\mr{hol}}$ this $\bb{CP}^1$-family does not coincide (specifically at $\infty$) with the family of twists considered in \cite{KW}.  However we'll show that this family does indeed produce the expected B-model after twisting at the point $Q_B$. 

One can alternatively fit $Q_B$ into a larger family of twists that is closed under S-duality. First, we introduce a family of \emph{holomorphic-topological} supercharges $Q_{\text{hol}} - \lambda(\alpha_2^\vee \otimes f_2^*)$, so called because we think of them as being holomorphic in two real dimensions -- i.e. one complex dimension -- and topological in the remaining two.  In other words a \emph{three-dimensional} family of translations will be exact for the action of these supercharges. Holomorphic-topological twists of this form were originally studied by Kapustin \cite{KapustinHolo}.  We can then define a family of topological supercharges approximating $Q_A$ by
\[Q_\lambda = Q_{\text{hol}} - \lambda(\alpha_2^\vee \otimes f_2^*) + (\alpha_2 \otimes e_2) \]
for each $\lambda \in \CC$.  These are A-type deformations of the holomorphic-topological twists $Q_{\text{hol}} - \lambda(\alpha_2^\vee \otimes f_2^*)$ which converge to $Q_A$ as $\lambda \to 0$. Note that $Q_B$ can also be realized as a deformation of $Q_{\text{hol}} - \lambda(\alpha_2^\vee \otimes f_2^*)$. The task of investigating the family of \emph{quantum} field theories obtained by twisting by these supercharges remains to be addressed in future work. 
 
\subsubsection{Superspace Formalism} \label{superspace_section}

The above formalism will allow us to define the action of a supersymmetry algebra on certain theories over $\RR^4$, and to produce topologically and holomorphically twisted versions with desirable symmetry properties.  However, it'll be important for us to generalize these theories to theories defined on more general manifolds than $\RR^4$.  We'll do this by \emph{globalizing} the twisted supersymmetry algebras, i.e. realizing them as acting locally on the total spaces of certain \emph{super vector bundles} over our manifolds by infinitesimal symmetries.  To set up this so-called ``superspace formalism'' we'll need some language from supergeometry.  By a \emph{super-ring}, we'll just mean a $\ZZ/2\ZZ$-graded commutative ring.  We'll consider suitable ``superspaces'' whose local functions form such a super-ring.

\begin{definition}
A \textit{supermanifold} of dimension $n|m$ is a ringed space $(M,C^\infty_M)$ which is locally isomorphic to $(\RR^{n} , C^\infty(\RR^n;\CC) [\eps_1,\ldots,\eps_m])$, where the $\eps_i$ are odd variables.
\end{definition}

\begin{remark}
Note that we're defining a supermanifold to have a structure sheaf consisting of \emph{complex valued} functions.  Such an object is sometimes called a \emph{complex supersymmetric} (or \emph{cs}) supermanifold, for instance by Witten \cite{WittenSUSY}.
\end{remark}

A typical example of the kind of supermanifold we are going to consider is the total space of an odd vector bundle.  We can define this as follows.

\begin{example}
Let $M$ be a real manifold, and let $E$ be a complex vector bundle on $M$. Then we define a supermanifold $(\Pi E , C^\infty_{\Pi E})$ by setting $C^\infty_{\Pi E}(U) = C^\infty(U,\wedge^\bullet E^*)$ for each open set $U \subset M$. In particular, if $E = T_M$ is the tangent bundle of $M$, then the sheaf of smooth functions on $U \subset \Pi E$ is the space $\Omega^\bullet(U)$ of smooth differential forms on $U$.  Supermanifolds diffeomorphic to a supermanifold of this form are called \emph{split}. 
\end{example}

We can also define an algebraic analogue. 

\begin{definition}
A \textit{supervariety} of dimension $n|m$ is a ringed space $(X,\OO_X)$ which is locally isomorphic to $(\spec R , R[\eps_1,\cdots,\eps_m ] )$ for a reduced $\CC$-algebra $R$ of Krull dimension $n$. 
\end{definition}

Note that every smooth supervariety naturally yields a supermanifold.  Our vector bundle example still makes sense in an algebraic sense.

\begin{example}
Let $X$ be a smooth complex algebraic variety and $E$ be an algebraic vector bundle on $X$. Then we define a supervariety $(\Pi E , \OO_{\Pi E})$ by setting $\OO_{\Pi E}(U) = \OO( U ,\wedge^\bullet E^*)$. Supervarieties isomorphic to the ones of this form are called \emph{split} supervarieties.
\end{example}

\begin{remark}
There is a fundamental difference between the smooth and algebraic settings.  In the smooth setting, a theorem of Batchelor \cite{Batchelor} says that all supermanifolds are split.  In the complex algebraic setting this is very much not true, and there are many non-split supervarieties.  Luckily, all the examples we'll need to deal with in what follows will be split, so this subtlety will not play a role.
\end{remark}

\begin{example}
An example of a natural supervariety of this form is the complex \emph{super projective space} $\bb{CP}^{n|m}$, modelling the quotient of the supermanifold $\CC^{n+1|m} \bs \{0\}$ under the action of $\CC^\times$ by rescaling.  Concretely, $\bb{CP}^{n|m}$ is the total space of the odd algebraic vector bundle $\Pi(\OO(1) \otimes \CC^m)$ over $\bb{CP}^n$, as one can readily check by analysing the transition functions for the odd coordinates between affine charts.
\end{example}

If we want to do calculus on supermanifolds, we need an analogue of the \emph{canonical bundle} for a supermanifold.  
\begin{definition}
For a split supermanifold $\Pi E$ for $E \rightarrow M$, we define the \emph{Berezinian} to be the super vector bundle $\mr{Ber}_{\Pi E} = \det(T^*_M \oplus E^*)$ over $\Pi E$.  Similarly, for a split supervariety $\Pi E$ for $E \rightarrow X$, we define the Berezinian to be the algebraic super vector bundle $\mr{Ber}_{\Pi E} = \det(T^*_{X} \oplus E^*)$ over $\Pi E$, where $T_{X}$ denotes the algebraic tangent bundle of $X$.
\end{definition}

\begin{example}
Let $\Sigma$ be a smooth curve and $L$ be a line bundle over $\Sigma$. For the supervariety $\Pi L$ over $\Sigma$ with projection map $p \colon \Pi L \rightarrow \Sigma$, its Berezinian is the bundle $\mr{Ber}_{\Pi L} = p^*(K_\Sigma \otimes L^* )$ on $\Pi L$. 
\end{example}

\begin{definition} \label{CY_def}
A \emph{Calabi--Yau structure} on a supervariety $X$ is a trivialization of the Berezinian, i.e. a complex vector bundle isomorphism from $\mr{Ber}_X$ to the trivial bundle.
\end{definition}

Now let us globalize the Kapustin--Witten family of topological twists in the language of supergeometry.  To do this, we'll find an odd vector bundle $\Pi E$ over $\CC^2$ and an action of the $Q_{\text{hol}}$-cohomology of the supersymmetry algebra on $\Pi E$ extending the natural action of the bosonic symmetries $\so(3;\CC) \oplus \left\langle \frac \dd{\dd z_1},\frac \dd{\dd z_2} \right\rangle$.  Since the space of odd symmetries is 9-dimensional, a natural choice for $\Pi E$ is the superspace $\CC^{2|3} \to \CC^2$ (which has a 9-dimensional space of odd vector fields).  Choose coordinates $(z_1,z_2,\varepsilon,\varepsilon_1,\varepsilon_2)$ for this superspace, where the complexified rotations $\so(3;\CC)$ act on the bosonic coordinates by its spin representation, and the R-symmetries $\sl(3;\CC)$ act on the fermionic coordinates.  In these coordinates, we define the action of the supersymmetries by the following odd vector fields.
\begin{align*}
\alpha_2\otimes e_2 &= \delby{\varepsilon} \\
\alpha_2 \otimes f_i &= (-1)^{i+1}\delby{\varepsilon_i} \\
\alpha_j^\vee \otimes e_2^* &= \varepsilon \delby{z_j} \\
\text{and }\  \alpha_j^\vee\otimes f_i^* &= (-1)^{i+1}\varepsilon_i\delby{z_j}
\end{align*}
for $i,j \in \{1,2\}$. This does indeed define an action of the super Lie algebra, i.e. the vector fields satisfy the correct commutation relations.  In this notation, the topological supercharges act by the vector fields
\[Q_{(\mu:\nu)} = \left(\mu \left(\varepsilon_1\delby{z_1} + \varepsilon_2 \delby{z_2} \right) + \nu\delby{\varepsilon} \right).\]
Note that we abuse the notation $Q_{(\mu:\nu)}$ to mean the one in the previous subsection after taking $Q_{\mr{hol}}$-cohomology.

It remains to extend these local vector fields to global vector fields on a 4-manifold $X$.  We'll be able to do this if $X$ has the structure of a complex surface. Since $\SU(2)_-$ acts on $S_-$ as the fundamental representation, one can identify $\varepsilon_i=dz_i$ and hence simply write $\varepsilon_1\delby{z_1} + \varepsilon_2 \delby{z_2} = \del$. On the other hand, $\varepsilon$ belongs to the trivial representation, and hence should be a trivial odd line bundle. Namely, for a given complex surface $X$, the global superspace we end up with after the holomorphic twist is $Y=\Pi TX \times \CC^{0|1}$, where further twists are described by the algebraic vector fields $\lambda \del + \mu \delby{\varepsilon}$.

If, furthermore, $X$ splits as the product of two smooth algebraic curves $X = \Sigma_1 \times \Sigma_2$, we can globalize the action of the topological twists $Q_\lambda$.  In the coordinates above, these twists act locally by
\[Q_\lambda = \lambda \eps_2 \delby{z_2} + \delby{\eps}\]
which, by the argument above, describes the local action of the odd vector field $\lambda \del_2 + \delby{\eps}$ where $\del_2$ is the algebraic de Rham operator on $\Sigma_2$ only.

\subsection{Twisted Supersymmetric Field Theories} \label{classical_field_theory_section}
Now, let's discuss what we'll mean by a classical field theory, and what it means to twist such an object.  The definitions in this section will build on the perturbative definitions given by Costello in \cite{CostelloSH}, but extended to a global, non-perturbative setting.  In doing so we'll find that, indeed, topological and holomorphic twists give rise to topological and holomorphic field theories respectively, justifying their names (by holomorphic field theories, we mean those where observables depend only on a choice of complex structure on spacetime, not on a choice of metric.  In two dimensions this will coincide with the notion of a (chiral) conformal field theory).   The supercharge $Q$ with which we wish to twist generates a one-odd-dimensional abelian superalgebra $\CC Q$, and the twisted theory will be -- perturbatively -- defined as something very close to the derived $\CC Q$-invariants of the untwisted theory. 

Globally, we can define a twist with respect $Q$ as a family of derived stacks over $\bb A^1$ so that, on the relative tangent bundle to a section, we recover a perturbative twist of the fiber at 0 by $Q$.  In general there is no reason that such global twists should be unique, but in many examples we'll see that there exists a natural choice provided by theorems of Gaitsgory and Rozenblyum.

\subsubsection{Classical Field Theories}
Costello and Gwilliam  \cite{CostelloGwilliam1, CostelloGwilliam2} give a beautiful axiomatisation of the notion of a perturbative classical field theory amenable to quantization and explicit calculation.  The definition we'll give will be a global extension of this definition, but to perform any calculations (especially for quantization) we'll restrict to the world of perturbation theory, and to their language.  One should view our definition as encoding the \emph{moduli space of solutions to the equations of motion} in a theory, and Costello and Gwilliam's definition as describing the formal neighborhood of a point in this moduli space.  We'll begin by briefly recalling the definition of a perturbative classical field theory.

\begin{remark}
In this section, by ``vector spaces'' we'll mean cochain complexes of nuclear Fr\'echet spaces.  We'll use $E^\vee$ to denote the strong dual of a vector space, and $E \otimes F$ will denote the completed projective tensor product.  We'll write $\symc(E)$ for the completed symmetric algebra built using this tensor product. 

For a vector bundle $E$ on a space $X$, we'll use the calligraphic letter $\mc E$ for its sheaf of sections, and we'll denote by $\mc E_c$ the corresponding sheaf of \emph{compactly supported} sections.  We'll write $E^!$ for the twisted dual bundle $E^\vee \otimes \dens_X$ where $\dens_X$ is the bundle of \emph{densities}, so there's a natural pairing $E \otimes E^! \to \dens_X$ of vector bundles.
\end{remark}

\begin{definition}
An \emph{elliptic $L_\infty$ algebra} $E$ on a topological space $X$ is a local $L_\infty$ algebra (as in Appendix \ref{linfty_appendix}) over $X$ which is elliptic as a cochain complex. A \emph{perturbative classical field theory} is an elliptic $L_\infty$ algebra $E$ equipped with a non-degenerate, invariant, symmetric bilinear pairing
\[\langle-,- \rangle \colon E \otimes  E[3] \to \dens_X\]
where $\dens_X$ denotes the bundle of densities on $X$.  Here \emph{invariant} means that the induced pairing on the sheaf of compactly supported sections
\[\int_X \langle-,-\rangle \colon \mc E_c \otimes \mc E_c[3] \to \CC\]
is invariant.
\end{definition}

From a perturbative classical field theory in this sense, we can produce a more geometric object.  Indeed, the fundamental theorem of deformation theory (as described in Appendix \ref{linfty_appendix}) allows us to associate to a local $L_\infty$ algebra $E$ a sheaf of formal moduli problems $BE$, and this correspondence provides an equivalence of categories. If the $L_\infty$ algebra $E$ is equipped with a degree $k$ pairing then we say the formal moduli problem $BE$ inherits a \emph{presymplectic form} of degree $k+2$.  We use this to motivate a general definition in the language of derived algebraic geometry, using a theory of shifted symplectic structures that is applicable in great generality.

In their 2013 paper \cite{PTVV}, Pantev, To\"en, Vaqui\'e, and Vezzosi define the notion of a \emph{shifted symplectic structure} on a derived Artin stack.  We refer to their paper and the paper \cite{Calaque} of Calaque for details, but we should note that a $k$-symplectic structure on $\mc M$ induces a non-degenerate degree $k$ pairing on the tangent complex $\bb{T}_{\mc M}$, and thus a degree $k-2$ pairing on the shifted tangent complex $\bb{T}_{\mc M}[-1]$.  In the recent sequel \cite{CPTVV}, Calaque, Pantev, To\"en, Vaqui\'e, and Vezzosi generalize this notion to that of a \emph{shifted Poisson structure}, and prove that this recovers the notion of a shifted symplectic structure when a non-degeneracy condition is imposed (a different proof for Deligne--Mumford stacks only also appeared in an earlier preprint of Pridham \cite{Pridham}).

We'll begin by giving an \emph{ideal} definition of a non-perturbative classical field theory that we believe best captures the structure of local classical solutions to the equations of motion.

\begin{definition} \label{idealdefinition}
A \emph{classical field theory} on a smooth manifold $X$ is a sheaf $\mc M$ of $(-1)$-shifted Poisson derived stacks such that for each open set $U \subset X$, the shifted tangent complex $\bb{T}_{p}[-1] \mc M(U)$ for a closed point $p \in \mc M(U)$ is homotopy equivalent to a perturbative classical field theory when equipped with the degree $-3$ pairing induced from the shifted Poisson bracket.
\end{definition}

\begin{remark}
We assume that Costello's assumption of ellipticity is always satisfied in an algebraic setting, in view of the main example of de Rham forms $\Omega^\bullet_{\mr{alg} }(X)$ becoming elliptic in the analytic topology by the Dolbeault resolution. It is possible that one needs a more careful definition of ellipticity in an algebraic setting for a treatment of the quantization of algebraic perturbative theories, but this is beyond the scope of the present paper. 
\end{remark}

In practice, in this paper we'll need to use a modified, algebraic version of this definition.  There are several reasons for this.
\vspace{-10pt}
\begin{enumerate}
\item Since we hope to eventually describe the moduli spaces of interest in the geometric Langlands program as local solutions in a classical field theory, we'll need a model that depends on an \emph{algebraic structure} on the spacetime manifold.  As such we won't be able to make sense of classical solutions on a general analytic open set.  Instead we'll need to work with a topology whose open embeddings are algebraic maps.
\item The theories we'll construct will be built using mapping spaces out of spacetime.  In general, if a spacetime patch $U$ is not proper, these mapping spaces will be of infinite type, and so it will be technically difficult to describe shifted Poisson structures on them.  Rather than getting bogged down in these functional analysis issues we'll simply ask for a shifted symplectic structure on the \emph{global} sections (with the understanding that a more sophisticated analysis should also recover a global version of the local Poisson bracket used by Costello and Gwilliam).
\end{enumerate}

\begin{remark}
We expect that an alternative version of the theory should exist in the analytic topology, using a suitable notion of analytic derived stacks, for example the formalism introduced by Porta and Yu \cite{PortaYu} or a model based on the $C^\infty$ dg-manifolds of Carchedi and Roytenberg \cite{CarchediRoytenberg} or the d-manifolds of Joyce \cite{Joycedmanifolds}.
\end{remark}

\begin{definition} \label{classical_field_theory_def}
An \emph{algebraic classical field theory} on a smooth proper algebraic variety $X$ is an assignment of a derived stack $\mc M(U)$ to each Zariski open set $U \sub X$, with a $(-1)$-shifted symplectic structure on the space $\mc M(X)$ of global sections whose shifted tangent complex $\bb{T}_{\mc M(X)}[-1]$ is homotopy equivalent to the global sections of a perturbative classical field theory when equipped with the degree $-3$ pairing induced from the shifted symplectic pairing.
\end{definition}

\begin{remarks}
\begin{enumerate}
\item We've deliberately left the nature of the ``assignment'' in the definition imprecise, although we expect that the correct definition is a sheaf of derived stacks.  Constructing the restriction maps and finding a symplectic structure -- much like investigating the shifted Poisson structure on open sets -- will involve subtle functional analytic issues involving Verdier duality on infinite-dimensional stacks which is beyond the scope of the present work.  The main theorems of this paper involve a determination of the global sections of a classical field theory on a smooth proper variety, and are expected to need adjustment to extend to sheaves of derived stacks. We hope to discuss this issue elsewhere.
\item In what follows, we sometimes consider theories defined on not necessarily proper varieties, for instance $\CC^n$.  We will informally refer to assignments of derived stacks in this general setting also as algebraic classical field theories, even without an analysis of shifted Poisson structures. 
\item We could just as readily have made this definition using a finer topology, the \'etale topology for instance, but Zariski sheaves will be sufficiently general for the examples in the present paper.
\end{enumerate}
\end{remarks}

The intuition behind this definition is -- as we already stated -- to encode the idea of the derived moduli spaces of solutions to the equations of motion.  Globally, given a space of fields and an action functional we can produce a shifted symplectic derived stack by taking the derived critical locus of the action functional.  Locally there are subtleties due to the existence of a boundary (as discussed for instance by Deligne and Freed in their notes on classical field theories \cite{DeligneFreedCFT}): one can still determine the equations of motion but the space of derived solutions will at best have a shifted Poisson structure.

In what follows we'll single out a special family of algebraic classical field theories which is adapted for discussion of twists of supersymmetric Yang--Mills theories.  These will model theories whose classical fields include a 1-form field, which is constrained to describe an algebraic structure on a $G$-bundle on-shell, and where the rest of the fields are all determined by formal data.

\begin{definition} \label{algebraic_gauge_theory_def}
A \emph{formal algebraic gauge theory} on a smooth variety $X$ is an algebraic classical field theory $\mc M$ on $X$ with a map $\sigma \colon \bun_G(U) \to \mc M(U)$ for each Zariski open set $U \sub X$, such that $\sigma$ is inf-schematic and induces an equivalence $ \bun_G(U)^{\mr{red}} \to  \mc M(U)^{\mr{red}}$ of their reduced parts. If a formal algebraic gauge theory $\mc M$ additionally admits such a map $\pi \colon \mc M(U)  \to \bun_G(U)$ for each $U$ such that $\sigma$ is a section of $\pi$, then we call $\mc M$ \emph{fiberwise formal}.
\end{definition}

\begin{remark}
We'll see in our examples that there are natural twists of supersymmetric gauge theories that are not of this formal nature, for instance twists that form the total space of a (dg) vector bundle over $\bun_G$.  We'll motivate the appearance of such example by viewing them as natural extensions of formal algebraic gauge theories, but they do not intrinsically fit into the above definition.  We think of the definition as a tool that allows us to compute twists of supersymmetric gauge theories.
\end{remark}

\begin{example} \label{cotangent_theory_definition}
Given any sheaf $\mc M$ of derived stacks with elliptic tangent complex and where $\mc M(X)$ is finitely presented we obtain an algebraic classical field theory by taking the \emph{formal shifted cotangent space} $T^*_{\mr{form}}[-1]\mc M$.  At the perturbative level, if $E$ is an elliptic $L_\infty$ algebra this corresponds to taking the direct sum of $L_\infty$-algebras $E \oplus E^![-3]$, with invariant pairing induced from the evaluation pairing $E \otimes E^! \to \dens_X$.  If $\mc M$ admits a map $\sigma \colon \bun_G \to \mc M$ satisfying the hypotheses of Definition \ref{algebraic_gauge_theory_def}, then $T^*_{\mr{form}}[-1] \mc M$ defines a formal algebraic gauge theory, using the zero section map associated to the formal shifted cotangent space. Likewise, if $\mc M$ also admits a map $\pi \colon \mc M \to \bun_G$, so that $\sigma$ is a section as in Definition \ref{algebraic_gauge_theory_def} then the projection map makes $T^*_{\mr{form}}[-1] \mc M$ into a fiberwise formal algebraic gauge theory.
\end{example}

Having given a definition of a classical field theory, let's investigate what it means to \emph{twist} such objects.  We'll begin by explaining what it means to twist a perturbative classical field theory, then use this to give a non-perturbative definition of a twist of a formal algebraic gauge theory which will suffice for our examples. 

\subsubsection{Perturbative Twisting}
\begin{definition} 
A classical field theory $E$ on a space $X$ with an action of the super Poincar\'e algebra (such as $\RR^n$) is called \emph{supersymmetric} if it admits an action by the super Lie algebra $\mathfrak{so}(n,\CC)\ltimes \CC^4 \oplus \Pi((S_+ \otimes W) \oplus (S_- \otimes W^*))$ extending the natural action of the Poincar\'e algebra for some vector space $W$  (for a definition of a superalgebra action on a local $L_\infty$ algebra, see the appendix, Definition \ref{local_L_infty_definition}). 
\end{definition}

We'll be interested in supersymmetric field theories where the action extends to an action of the full supersymmetry algebra for some choice of R-symmetries.  In our examples for $N=4$, this will be the case with the subalgebra $\sl(4;\CC) \sub \gl(4;\CC)$ of (complexified) R-symmetries preserving a trivialization of the determinant bundle.

The data required to twist a classical field theory is the action of a certain supergroup.  Define a supergroup
\[H = \CC^\times \ltimes \Pi \CC\]
where $\CC^\times$ acts with weight 1.  This group arises as the group of automorphisms of the odd complex line.
\begin{definition}
\emph{Twisting data} for a classical field theory $\Phi$ on a space $X$ is a local action $(\alpha, Q)$ of $H$ on $\Phi(U)$ for all $U$.  That is, in the perturbative case $\Phi$ is a sheaf of $L_\infty$ algebras with $H$-module structure, and in the non-perturbative case $\Phi$ is a family of derived stacks with $H$-action.  In our notation, $\alpha$ is a $\CC^\times$ action, and $Q$ is an odd infinitesimal symmetry with $\alpha$-weight 1.  
\end{definition}

An important source of twisting data is a supersymmetry action.  Let $Q$ be a supercharge such that $[Q,Q]=0$, and let $\alpha$ be a $\CC^\times$ action such that $Q$ has weight one (we can always find such an action by choosing a suitable $\CC^\times$ in the group of R-symmetries, after choosing an exponentiation of the action of the R-symmetry algebra to an action of an R-symmetry \emph{group}.)  Since $[Q,Q]=0$, the supercharge $Q$ generates a subalgebra isomorphic to $\Pi \CC$ acting on any theory with the appropriate supersymmetry action, and along with $\alpha$ this defines an action of the supergroup $H$.

\begin{lemma} \label{supercochainequivalence}
There is an equivalence of categories
\[\{\text{super vector spaces with an } H \text{-action}\} \iso \{\text{super cochain complexes}\}.\]
\end{lemma}
Here the grading is given by the weight under the action of $\CC^\times$ and the differential is given by the action of $\Pi \CC$.  We use this fact to define a twisted theory for the data $(\alpha, Q)$.

\begin{definition} \label{perturbative_twist_def}
Let $E$ be a perturbative classical field theory with an action of the supergroup $H$.  The \emph{twisted theory} $E^Q$ (where $Q$ is a generator of $\Pi \CC$) is the theory obtained by introducing a new differential graded structure on $E$ in accordance with the previous lemma and taking the total complex with respect to this new grading and the cohomological grading.
\end{definition}

\begin{remark} \label{perturbative_deformation}
The twisted theory $E^Q$ fits into a family of classical field theories deforming $E$ -- i.e. a sheaf of perturbative field theories over the line $\bb A^1$ -- whose fiber at $\lambda$ is the theory obtained by applying the twisting construction with respect to the dilated twisting data $(\lambda Q, \alpha)$.
\end{remark}

\begin{remarks}
This definition needs some unpacking.  We should explain what we want to do intuitively, in particular the role of the action $\alpha$. 
\vspace{-15pt}
\begin{itemize}
 \item On the level of functions -- that is, observables -- our first idea is to take the $Q$-coinvariants.  By identifying observables with their orbits under $Q$ we force all $Q$-exact symmetries to act trivially, so if we choose a holomorphic or topological supercharge we impose strong symmetry conditions on the observables in the twisted theory.  The na\"ive thing to do to implement this procedure would be to take the derived invariants of our classical field theory with respect to the group $\Pi \CC$ generated by $Q$.  
  \item This is all well and good, but recall what a $\Pi \CC$-action actually means: the data of a family of classical field theories over the space $B (\Pi \CC)$ whose fiber over zero is $E$.  That is, a module over $\CC[[t]]$, where $t$ is a fermionic degree 1 parameter.  One really wants to restrict interest to a generic fiber of this family. 
 \item To do this we restrict to the odd formal \emph{punctured} disc, or equivalently invert the parameter $t$, then take invariants for an action $\alpha$ of $\CC^\times$ for which $t$ has weight 1, thus extracting a ``generic'' fiber instead of the special fiber at 0.  This is an instance of the Tate construction for the homotopy $\Pi \CC$ action $Q$.  It's important to restrict to the formal punctured disc, since not all these invariant fields extend across zero: if we just took $\CC^\times$ invariants in $\mc E[[t]]$ we'd obtain elements of $\mc E$ of the form $\phi t^k$ where $\phi$ had weight $-k$.  In particular we'd find ourselves throwing away everything of positive $\CC^\times$ weight in $\mc E$.  
 \item Now, this procedure is exactly the same as the definition we gave above.  Taking derived $Q$ invariants corresponds to taking the complex $\mc E[[t]]$ with differential $d_{\mc E} + tQ$.  Inverting $t$ and taking invariants under the action $\alpha$ is then the same as adding the $\alpha$ weight to the original grading, and adding the operator $Q$ to the original differential $d_{\mc E}$, just as in our definition.
\end{itemize}
\end{remarks}

\begin{prop} \label{perturbative_twist_still_a_theory}
The twisted theory $E^Q$ is still a classical field theory when equipped with a pairing inherited from $E$.
\end{prop}

\begin{proof}
First note that $E^Q$ is still an elliptic $L_\infty$ algebra.  The complex obtained as the $\Pi \CC$ invariants of the theory -- the complex $(\mc E[[t]], d_{\mc E} + tQ)$ -- is required to have the structure of an elliptic $L_\infty$ algebra by the definition of a group action on a field theory.  Inverting $t$ preserves this structure, as does taking $\CC^\times$-invariants, again because $\alpha$ is a local $L_\infty$-action.

It remains to construct an invariant pairing on $E^Q$ of the correct degree (we'll follow Costello \cite[13.1]{CostelloSH}).  The pairing on $\mc E$ induces a degree -3 pairing of form
\[\langle - , - \rangle_Q \colon E[[t]] \otimes E[[t]][3] \to \dens_X[[t]]\]
by $\langle e_1t^{k_1}, e_2t^{k_2} \rangle_Q = \langle e_1, e_2 \rangle t^{k_1 + k_2}$.  We only need to check that this is compatible with the differential $d_{\mc E} + tQ$, i.e. that exact terms on the left vanish under the pairing map, or more precisely that
\[\left(\langle d_{\mc E}f_1, f_2 \rangle + \langle f_1,d_{\mc E}f_2 \rangle \right)t^{k_1+k_2} + \left( \langle Qf_1,f_2 \rangle + \langle f_1, Qf_2 \rangle \right) t^{k_1+k_2+1} = 0.\]
The first term vanishes because of compatibility of $d_{\mc E}$ with the pairing, and the second term vanishes because $Q$ is a symmetry of the classical field theory.  This pairing yields an invariant $\dens_X((t))$ valued pairing after inverting $t$.  By construction these pairings are equivariant with respect to the action of $\CC^\times$ by rescaling $t$, so descends to a pairing
\[\langle - , - \rangle_Q \colon \left(E((t)) \otimes E((t))[3]\right)^{\CC^\times} \to \dens_X((t))^{\CC^\times} = \dens_X.\]
This pairing is still invariant, so gives $E^Q$ the structure of a classical field theory.
\end{proof}

\subsubsection{Global Twisting}
Now, let $\mc M$ be a non-perturbative algebraic classical field theory on $\CC^n$, and suppose $\mc M$ admits an action of a supersymmetry algebra extending the action of the translations.  As above, choose a supercharge $Q$ satisfying $[Q,Q]=0$, and an action $\alpha$ of $\CC^\times$ on $\mc M$ so that $Q$ has $\alpha$-weight one. 

\begin{definition}
A \emph{deformation} of a derived stack $\mc X$ is a derived stack $\pi \colon \mc X' \to \bb A^1$ flat over the affine line along with an immersion $\mc X \inj \mc X'_0$, and an equivalence $\mc X'|_{\bb G_m} \iso \mc X'_1 \times \bb G_m$, where $\mc X'_t$ is the fiber over the point $t$. 
\end{definition}

We'll begin with a prototypical example of a deformation, presented somewhat informally for motivation. We'll provide a more conceptual and general treatment of this example later in Example \ref{de_rham_stack_twist_example}.

\begin{example}
Consider a smooth proper variety $X$. We define a ringed space $X_{\mr{Dol}}$ by $X_{\mr{Dol}}:= (X, \OO_{T[1]X})$, where the structure sheaf $\OO_{T[1] X}$ is equivalent to $\mr{Sym}_{X} (\bb L_X[-1]) = \Omega^\bullet_{\mr{alg}, X}$. As one has a quasi-isomorphism $\Omega^p_{\mr{alg},X} \simeq (\mathcal{A}^{p,\bullet}_X , \overline{\del} )$ in the analytic topology, $X_{\mr{Dol}}$ is justifiably called the \emph{Dolbeault stack} of $X$. Similarly, one defines $X_{\lambda\text{-dR}}$ to be the ringed space $(X, (\Omega^\bullet_{\mr{alg}, X } , \lambda \del))$. Of course, $X_{\text{dR} } : = X_{1 \text{-dR}}$ is called the \emph{de Rham stack} of $X$ because one has $(\Omega^\bullet_{\mr{alg}, X } , \del)  \simeq  (\mathcal{A}^{\bullet,\bullet}_X , \del + \overline{\del} ) \simeq  (\mathcal{A}^{\bullet}_X , d)$ in the analytic topology. It will sometimes be convenient to write $X_{ 0 \text{-dR} }$ for $X_{\mr{Dol}}$.

There exists a ringed space $X_{\mr{Hod}}$ and a map $X_{\mr{Hod}} \to \bb A^1$ such that the fiber over $\lambda$ is $(X, (\Omega^\bullet_{\mr{alg}, X } ,  \lambda\del))$. That is, both squares in the following diagram are fiber product squares.
\[\xymatrix{
X_{\mr{Dol}} \ar[d] \ar[r] & X_{\mr{Hod}} \ar[d] & X_{\lambda\text{-dR}}\ar[d] \ar[l] \\
\{0\} \ar[r] & \bb A^1 & \{\lambda\} \ar[l]
}\]
In particular, $X_{\mr{Hod}}$ is a deformation of $X_{\mr{Dol}}$.

Now we would like to write down this information in a way that can be easily generalized to other situations. First, observe that as $X_{\mr{Dol}}$ and $X_{\lambda \text{-dR}}$ have the same closed points, they differ only by an infinitesimal thickening from the original space $X$. In order to write this more carefully, let us introduce the canonical map $\sigma_\lambda \colon  X \to X_{\lambda\text{-dR} }$. Then we would like to compare $\bb T_{\sigma_0(x) } X_{\mr{Dol}}$ and $\bb T_{\sigma_\lambda(x) } X_{\lambda\text{-dR}}$ for every $x \in X$. A way to compare them is to find a section $s \colon \bb A^1 \to X_{\mr{Hod}}$ so that both of them are realized as fibers of $s^* \bb T_{X_{\mr{Hod}} / \bb A^1}$. If that is the case, then one declares $X_{\mr{dR}}$ to be a \emph{twist} of $X_{\mr{Dol}}$.

On the other hand, in general, one might not have a map playing the role of $\sigma_\lambda$, even if we started with a map $\sigma_0$ which is an equivalence at the level of closed points. Then it would be reasonable to ask for compatibility for every point $x_1 \in X_{\text{dR} }$. Namely, for a closed point $x_1 \in X_{\text{dR}}$, we ask the existence of a section $s \colon \bb A^1\to X_{\mr{Hod}}$ such that
\vspace{-12pt}
\begin{enumerate}
\item $s(0) = \sigma_0(x)$ for some $x \in X$,
\item $s(\lambda) =  x_\lambda  $ for some $x_\lambda \in X_{\lambda \text{-dR} }$, and
\item  $s^* \bb T_{X_{\mr{Hod}} / \bb A^1}$ is a deformation of $\bb T_{\sigma_0(x) }{X_{\mr{Dol}}}$.
\end{enumerate}
\vspace{-12pt}
Even only with this weaker requirement, we think of $X_{\mr{dR}}$ as a twist of $X_{\mr{Dol}}$.
\end{example}

When we define a twist of a formal algebraic gauge theory, there are two additional small complications to be introduced.  Firstly, given twisting data $(\alpha, Q)$, before twisting by $Q$ we need to deal with modifying the gradings by the $\CC^\times$-weight under $\alpha$.

\begin{definition} 
A \emph{regrading} of a formal algebraic gauge theory $\mc M$ with respect to a $\CC^\times$ action $\alpha$ such that $\sigma \colon \bun_G \to \mc M$ is equivariant for the trivial action on $\bun_G$ is a formal algebraic gauge theory $\sigma_\alpha \colon \bun_G \to \mc M^\alpha$ such that the restricted tangent complex $\sigma_\alpha^* \bb{T}_{\mc M^\alpha}[-1]$ is equivalent to the restricted tangent complex of $\mc M$ with degrees modified by adding the $\alpha$-weight to the cohomological degree and the $\alpha$-weight mod 2 to the fermionic degree, as a sheaf of Lie algebras. 
\end{definition}

The second complication is that a perturbative classical field theory consists of more data than just a cochain complex, and our twist must preserve this additional information on the level of each tangent complex, in the sense discussed in the previous section on twists of perturbative field theories.

Bearing these two points in mind, by mimicking the motivating example with $X$ replaced by $\mr{Bun}_G$, we obtain the following definition.

\begin{definition} \label{twist_definition}
A classical non-perturbative field theory $\mc M^Q$ is a \emph{twist} of a formal algebraic gauge theory $\mc M$ with respect to twisting data $(\alpha,Q)$ if there is a deformation $\pi \colon \mc M' \to \bb A^1$ of the regrading $\mc M^\alpha$, whose generic fiber is equivalent to $\mc M^Q$, such that for every closed point $x_1 \in \mc M^Q$, there is a section $s \colon \bb A^1 \to \mc M'$ of the map $\pi$ such that
\vspace{-13pt}
\begin{enumerate}
\item $s(0) = \sigma_\alpha(x)$ for some $x \in \mr{Bun}_G$,
\item $s(\lambda) =  x_\lambda  $ for some $x_\lambda \in \mc M^{\lambda Q }$, and
\item  $s^* \bb T_{ \mc M ' / \bb A^1}$ is a perturbative twist of $\bb T_{\sigma(x) }\mc M$ with respect to the given twisting data as in Remark \ref{perturbative_deformation}.
\end{enumerate}
\end{definition}

\begin{remark}
One could define twists of more general classical field theories as long as they could be viewed as formal extensions of some fixed base stack (playing the role of $\bun_G$ in the above definition).  For example, one could replace $\bun_G$ by maps into a target other than $BG$ to describe twists of supersymmetric sigma models, or if $\mc M^Q = T^*[-1]\mc B$ was a cotangent theory one might use the base space $\mc B$.
\end{remark}

\begin{remark} \label{nonperturbative_twist_construction}
One ought to be able to produce twisted field theories explicitly from a functor-of-points perspective, along the lines of a construction explained by Grady and Gwilliam \cite{GG}.  Let $\mc L$ be an $L_\infty$ space (we refer the reader to Grady--Gwilliam or Costello \cite{CostelloSH} for details concerning the theory of $L_\infty$ spaces) over a scheme $M$ whose fibers are finitely generated and concentrated in non-negative degrees, and let $\mc L$ be equipped with a degree $-3$ invariant pairing on its fibers making it into a sheaf of perturbative classical field theories.  Then we can attempt to build a non-perturbative classical field theory out of $\mc L$ as follows.  Let $\mc L_{>0}$ be the truncation in positive degrees: a \emph{nilpotent} $L_\infty$ space, and let $L_0$ be the degree 0 piece: a sheaf of Lie algebras.  We can attempt to construct a sheaf $\mc M$ of derived stacks over $M$ by a Maurer--Cartan procedure.  To do so, choose an exponentiation of $L_0$ to a sheaf of algebraic groups $G$.  Define, for a cdga $R$ concentrated in non-positive degrees, the $R$-points of $\mc M$ over $U$ by
\[\mc M(U)(R) = \mr{MC}(\mc L_{>0}(U) \otimes R) / G(U)(R).\]
We can easily compute the shifted tangent complex at a point $p \in \mc M$, since
\begin{align*}
\bb{T}_p[-1]\mc M &= \bb{T}_p[-1](\mr{MC}(\mc L_{>0})/G) \\
&\iso (\mc L_0)_p \to (\mc L_{>0})_p \\
&= \mc L_p,
\end{align*}
so we recover the perturbative theory.  Grady and Gwilliam \cite{GG} prove that this construction satisfies a descent condition, albeit a weaker condition than the condition we've demanded for derived stacks.  We anticipate that applying this construction to the twist of a perturbative classical field theory will yield a non-perturbative twisted theory, compatibly with the examples we construct elsewhere in the paper.  
\end{remark}


We'll construct twists of the $N=4$ theories of interest to us in Section \ref{constructing_theories_section} below, but why should the twisted theory with respect to specified twisting data be well-defined?  Well, for many theories of the type we're considering it is possible to recover the full non-perturbative theory from a family of perturbative theories parametrized by $\bun_G$. This follows from a theorem of Gaitsgory and Rozenblyum \cite{GR}.  Even when this formal procedure fails we'll see that the Gaitsgory--Rozenblyum correspondence often provides a natural choice of twist.

The following definition, also due to Gaitsgory and Rozenblyum, models in derived algebraic geometry a family of formal moduli problems as described in Appendix \ref{linfty_appendix} over a base derived stack $\mc X$, coherently equipped with base points.

\begin{definition}
A \emph{pointed formal moduli problem} $\mc Y$ over a derived stack $\mc X$ is an inf-schematic morphism $\pi \colon \mc Y \to \mc X$ of prestacks with an inf-schematic section $\sigma \colon \mc X \to \mc Y$ such that the induced map $\pi^{\mr{red}} \colon \mc Y^{\mr{red}} \to \mc X^{\mr{red}}$ is an isomorphism.  We'll denote the category of pointed formal moduli problems over $\mc X$ by $\text{Ptd}(\text{FormMod}_{ / \mc X} )$.
\end{definition}

\begin{theorem} [Gaitsgory--Rozenblyum {\cite[5.1.6.4, 7.3.6.2]{GR}}] \label{GRtheorem}
For a derived stack $\mc X$ which is locally almost of finite type there is an equivalence
\[F \colon \text{Ptd}(\text{FormMod}_{/ \mc X} ) \to \text{LieAlg}(\text{IndCoh}(\mc X) ),\] 
where $\text{LieAlg}(\text{IndCoh}(\mc X) )$ is the category of Lie algebra objects in ind-coherent sheaves on $\mc X$.
\end{theorem}

We can now more succinctly say that a fiberwise formal algebraic gauge theory is an assignment to open sets in $X$ of pointed formal moduli problems over $\bun_G$, with the structure of an algebraic classical field theory on its total space. Theorem \ref{GRtheorem} therefore says that fiberwise formal algebraic gauge theories are completely determined by Lie algebra objects in sheaves over $\bun_G$.  We'll take advantage of this, and define the twist of a fiberwise formal algebraic gauge theory using this sheaf of Lie algebras.

It will be useful to unpack what exactly the functor in the theorem is. It is constructed as a composition of two equivalences \[\xymatrix{
\text{Ptd}(\text{FormMod} _{/ \mc X} ) \ar[r]^-{\Omega_{\mc X} }  & \text{Grp}(\text{FormMod}_{ / \mc X} ) \ar[r]^-{\text{Lie}} & \text{LieAlg}(\text{IndCoh}(\mc X) ),
} \] where $\text{Grp}(\text{FormMod} _{/ \mc X} ) $ stands for the category of group objects in $\mr{FormMod}_{/ \mc X}$. Here $\Omega_{\mc X}$ is the based loop space functor $\mc Y \mapsto \Omega_{\mc X} \mc Y = \mc X \times_{\mc Y} \mc X$ and $\mr{Lie}$ is the functor given by $\mc H \mapsto \mathbb{T}_{ \mc H / \mc X }|_{\mc X}$, so that the composition in terms of the underlying ind-coherent sheaf is simply $\mc Y \mapsto \mathbb{T}_{ \mc Y / \mc X }|_{\mc X} [-1]$. In other words, one can write $F = \sigma^* \mathbb{T}_{ / \mc X}[-1]$, the restricted relative shifted tangent complex.

Now, we'll discuss a construction of twists of fiberwise formal algebraic gauge theories.  In order to give as general a construction as possible we'll need to consider a stronger form of the Gaitsgory--Rozenblyum correspondence than Theorem \ref{GRtheorem}, also due to Gaitsgory--Rozenblyum. This is because a fiberwise formal algebraic gauge theory does not necessarily remain fiberwise formal when we twist: in general the twisting data will not preserve the fibers of the projection map $\pi$, so this structure is lost upon twisting. 

Consider the commutative diagram:
\[\xymatrix{
\text{Ptd}(\text{FormMod}_{/ \mc X} ) \ar[r]^-{\Omega_{\mc X} } \ar[d]_{\text{forget}} & \text{Grp}(\text{FormMod}_{ / \mc X} ) \ar[r]^-{\text{Lie}} \ar[d]^{\text{forget}}& \text{LieAlg}(\text{IndCoh}(\mc X) ) \ar[d],\\
\text{FormMod}_{\mc X/} \ar[r]^-{\Omega_{\mc X} } & \text{FormGrpoid}(\mc X) \ar[r]^-{\text{Lie}} & \text{LieAlgebroid}( \mc X).
}\]
Here $\text{FormMod}_{\mc X/}$ stands for the category of formal moduli problems under $\mc X$, so that a formal algebraic gauge theory is exactly an algebraic classical field theory -- given by a family of formal moduli problems -- under $\mc X = \bun_G$. The other categories are also defined in Gaitsgory--Rozenblyum, but for our purposes it will suffice to note that the abusive notations $\Omega_{\mc X}$ and $\mr{Lie}$ still realize equivalences and that the forgetful functor from $\text{Ptd}(\text{FormMod}_{/ \mc X} ) $ to $\text{FormMod}_{\mc X/}$ is given by the natural identification $\text{Ptd}(\text{FormMod}_{/ \mc X} ) = (\text{FormMod}_{\mc X/})_{/ \mc X}$.  We'll now state the necessary generalization of Theorem \ref{GRtheorem}.

\begin{theorem}[Gaitsgory--Rozenblyum {\cite[5.2.3.2, 8.2.1]{GR}}] \label{GRtheorem2}
The functor
\[\mr{Lie} \circ \Omega_{\mc X} \colon \text{FormMod}_{\mc X/} \to \text{LieAlgebroid}( \mc X)\]
is an equivalence for any derived stack $\mc X$ locally almost of finite type.
\end{theorem}

We don't define the general notion of Lie algebroids here, referring the reader instead to Gaitsgory--Rozenblyum for details. In the present paper essentially only two types of examples of Lie algebroids will appear, the initial object and the terminal object in the category $\text{LieAlgebroid}(\mc X)$, so we'll use a more concrete way to think about them in terms of an anchor map. Namely, we use the forgetful functor \[\mr{Anch} \colon \text{LieAlgebroid}(\mc X) \to \text{IndCoh}(\mc X)_{ / \bb{T}_{\mc X} }\] defined by sending the formal moduli problem $\mc X \to \mc Y$, which we identify with a Lie algebroid by Theorem \ref{GRtheorem2}, to $\bb T_{\mc X/ \mc Y} \to \bb{T}_{\mc X}$, where the map is induced from the identity $\bb{T}_{\mc X} \to \bb{T}_{\mc X}$. In particular we have $\mr{Anch}( \mc X \to \mc X ) = ( 0 \to\bb T_{\mc X} )$, which we call the \emph{zero Lie algebroid}, and $\mr{Anch}(\mc X \to \mc X_{\mr{ dR}} ) = (\mr{id} \colon  \bb T_{\mc X }\to \bb T_{\mc X} )$, which we call the \emph{tangent Lie algebroid}. 

At this point we'll introduce our main example: the de Rham prestack arising as a deformation of the formal 1-shifted tangent bundle.  Before we do so we'll introduce some relevant geometric objects originally constructed by Simpson \cite{Simpson2, Simpson3, Simpson4}. 

\begin{definition} \label{lambda_connection}
A \emph{$\lambda$-connection} on an algebraic $G$-bundle $P$ over a smooth complex variety $X$ is a map 
\[\del_\lambda \colon \Omega^{0}_{\mr{alg}}(X;\gg_P) \to \Omega^{1}_{\mr{alg}}(X;\gg_P)\] 
such that $\del_\lambda (f \cdot s ) = \lambda (\del f) s + f \del_\lambda s$ for $f \in \mc O_X$ and $s \in \Omega^{1}_{\mr{alg}}(X;\gg_P)$.  A $\lambda$-connection $\del_\lambda$ is called \emph{flat} if $\del_\lambda ^2=0$, where $\del_\lambda$ naturally extends to a map $\Omega_{\mr{alg}}^{i}(X;\gg_P) \to \Omega_{\mr{alg}}^{i+1}(X;\gg_P)$ for all $i$. 
\end{definition}

In particular, if $\lambda \ne 0$ and $\dd_\lambda$ is a flat $\lambda$-connection, then $\lambda^{-1} \dd_\lambda$ is an algebraic flat connection on an algebraic $G$-bundle.  If $\lambda=0$ then a flat $\lambda$-connection is a section $\phi$ of $\Omega_{\mr{alg}}^{1}(X;\gg_P)$ satisfying $[\phi,\phi]=0$: a Higgs field. 

\begin{example} \label{de_rham_stack_twist_example}
Let $\mc X$ be a derived Artin stack.  We can define a prestack $\mc X_{\mr{Hod}}$, the \emph{Hodge prestack} of $\mc X$, as a deformation of the formal 1-shifted tangent bundle $T_{\mr{form}}[1]\mc X$.  Such a deformation is -- by definition -- a flat morphism $\pi \colon \mc Y \rightarrow \mathbb{A}^1$ with $\mc Y_0 = T_{\mr{form}} [1] \mc X$ and $\mc Y |_{\mathbb{G}_m} \cong \mc Y_1 \times \mathbb{G}_m$. We first construct a formal moduli under $X \times \bb A^1$. Having $T_{\mr{form}}[1] \mc X$ as an object of $\text{FormMod}_{\mc X/}$ using Theorem \ref{GRtheorem2}, whose associated Lie algebroid is $0 \colon \mathbb{T}_{\mc X} \to \mathbb{T}_{\mc X}$, one can easily think of its deformation $\lambda Q$ parametrized by $\lambda \in \mathbb{A}^1$ with $Q= \mr{id} \colon  \mathbb{T}_{\mc X} \to \mathbb{T}_{\mc X}$ in the category of Lie algebroids: this gives rise to a formal moduli problem under $\mc X \times \bb A^1$. It remains to construct a map down to $\bb A^1$ for which we refer to Gaitsgory--Rozenblyum \cite[Chapter 9]{GR}, where this map is constructed as an example of a more general ``scaling'' construction, applied to the prestack $\mc X_{\mr{dR}}$.

We denote the fiber of $X_{\mr{Hod}}$ over a point $\lambda \in \CC$ by $X_{\lambda \text{-dR}}$.  The fiber over $\lambda = 1$ is the usual de Rham prestack $X_{\mr{dR}}$ -- since the formal moduli problem $\mc X_{\mr{dR}}$ under $\mc X$ corresponds to the tangent Lie algebroid $\mr{id} \colon \mathbb{T}_{\mc X} \to \mathbb{T}_{\mc X}$ -- and the fiber over $\lambda = 0$ is also called the \emph{Dolbeault stack}, and denoted $X_{\mr{Dol}}$.  We denote the mapping stack into $BG$ by
\[\underline{\mr{Map}}(X_{\lambda \text{-dR}}, BG) = \Flat^\lambda_G(X).\]
It represents flat $\lambda$-connections on $X$ when $X$ is a smooth variety.  When $\lambda=0$ we recover the moduli stack of Higgs bundles on $X$ for the group $G$.
\end{example}

\begin{remark}
Simpson \cite{Simpson4} originally gave a different definition in the case where $X$ is a scheme, modelling $X_{\mr{Hod}}$ as a groupoid in schemes living over $\bb A^1$.  First form the deformation to the normal cone of the diagonal map $\Delta \colon X \inj X \times X$.  This is a $\bb G_m$-equivariant scheme living over $\bb A^1$ whose fiber over $\lambda \ne 0$ is just $X \times X$ with $X$ included diagonally, and whose fiber over $0$ is the tangent space $TX$ with $X$ included as the zero section.  Form the formal completion of $X \times \bb A^1$ inside this total space.  This admits two maps to $X \times \bb A^1$ inherited from the two projections $X \times X \to X$,
\[\mr{Def}(\Delta)^{\wedge}_{X \times \bb A^1} \rightrightarrows X \times \bb A^1.\]
The Hodge prestack $X_{\mr{Hod}}$ is equivalent to the coequalizer of these arrows in the category of stacks. For $\lambda =1$, it coincides with the usual definition of the de Rham prestack $X_{\mr{dR}}$. For $\lambda =0$, the coequalizer of the trivial action $T_{\mr{form} }X \rightrightarrows X$ is the relative classifying space $B_XT_{\mr{form}}X$ of the sheaf $T_{\mr{form}} X$ of formal groups over $X$, which in turn is the same as $T_{\mr{form}}[1] X$ by the discussion below the Theorem \ref{GRtheorem}: the two prestacks arise from the same Lie algebra.
\end{remark}

With this apparatus in hand, one can construct twists of fiberwise formal algebraic gauge theories, as long as the twisting data is compatible with the structure map $\sigma \colon \bun_G \to \mc M$, so that a twist exists within the category of formal algebraic gauge theories.  Let $\mc M$ be a fiberwise formal algebraic gauge theory acted on by twisting data $(\alpha,Q)$ preserving the fibers of the map $\sigma$.  This condition will be necessary for a natural twist to exist within formal algebraic gauge theories.  Let's be clear about precisely what compatibility we require between the structure maps of out formal algebraic gauge theories and the $H$-action.

\begin{definition} \label{preserves_fibers_def}
Let $f \colon \mc X \to \mc Y$ be a morphism of derived stacks, and suppose that the supergroup $H$ acts on $\mc Y$.  We say that the $H$-action \emph{preserves the fibers} of the map $f$ if the image of the map $df \colon \bb T_{\mc X} \to f^*\bb T_{\mc Y}$ is invariant under the $H$-action.  In particular this makes the relative tangent complex $\bb T_{\mc X / \mc Y}$ into a sheaf of $H$-representations. 
\end{definition}

We will proceed by defining the canonical twist for the case of $\sigma$ and $\pi$ both being preserved by the twisting data and of $\sigma$ being preserved independently first and show that these two are compatible. 

\begin{definition}[Twisting a fiberwise formal algebraic gauge theory] \label{two_twist_defs}
Let $\mc M$ be a fiberwise formal algebraic gauge theory with $\sigma \colon \mr{Bun}_G\to \mc M$ and $\pi \colon \mc M\to \mr{Bun}_G$. We always assume that the action of $H$ on $\bun_G$ is trivial.
\vspace{-12pt}
\begin{enumerate}
\item \label{lie_algebra_twist_def} Suppose that the twisting data $(\alpha,Q)$ preserves the fibers of both $\sigma$ and $\pi$. Then $\mc M$ -- as a Lie algebra object in $\mr{IndCoh}(\mr{Bun}_G )$ by Theorem \ref{GRtheorem} -- has a twist $\mc M^Q$ in the same category by Proposition \ref{perturbative_twist_still_a_theory}, which in turn can be identified with a fiberwise formal algebraic gauge theory by Theorem \ref{GRtheorem}.
\item  \label{anchor_twist_def}  Suppose that the twisting data $(\alpha,Q)$ preserves the fibers of $\sigma$. An $H$-equivariant map $\sigma \colon \mr{Bun}_G\to \mc M$ gives an ind-coherent sheaf $\bb T_{\mr{Bun}_G / \mc M }$  with $H$-action, while an $H$-equivariant map $\mc M \to (\mr{Bun}_G)_{\mr{dR}}$ under $\mr{Bun}_G$ gives a map $\bb T_{\mr{Bun}_G / \mc M } \to \bb T_{\mr{Bun}_G}$ of ind-coherent sheaves with $H$-action by Theorem  \ref{GRtheorem2}. Hence we can define the \textit{twisted anchor map} as the twist of $\mr{anch}(\bb T_{ \mr{Bun}_G / \mc M })$ which is still an object of $\mr{IndCoh}(\mr{Bun}_G)_{/\bb T_{ \mr{Bun}_G}}$.
\end{enumerate} 
\end{definition}

Note that in the first case, one retains a Lie algebra structure, which by Theorem \ref{GRtheorem} gives rise to a pointed formal moduli over $\bun_G \times \bb A^1$. Note that the projection down to $\bb A^1$ supplies the structure of a twist in the sense of \ref{twist_definition}; the necessary section is given by composing the pointing with the map $\bb A^1 \to \bb A^1 \times \bun_G$ associated to a closed point of $\bun_G$. In the second case we only obtain an ind-coherent sheaf with an anchor map to $\bb T_{\mr{Bun}_G}$.  These two definitions of twist are compatible.

\begin{prop}
Given a fiberwise formal algebraic gauge theory $\mc M$ with twisting data preserving the fibers of both $\sigma$ and $\pi$, the anchor map of the twisted theory $\mr{anch}(\bb T_{\bun_G/\mc M^Q})$ is equivalent to the twist of the anchor $\mr{anch}(\bb T_{\bun_G/\mc M})$.
\end{prop}

\begin{proof}
Because the twist $\mc M^Q$ is still a fiberwise formal algebraic gauge theory, its anchor map is zero.  The underlying ind-coherent sheaves of both the theory obtained by applying the functor $\mr{anch}$ to the twisted theory $\mc M^Q$, and the theory obtained by twisting $\mr{anch}(\bb T_{\bun_G/\mc M})$ coincide, and hence we must only check that if our twisting data is equivariant for $\pi$ then the twisted anchor map defined above is zero.  In this case we can factor the anchor map $\bb T_{\bun_G/\mc M} \to \bb T_{\bun_G}$ through zero as maps of $H$-representations, by applying the functor of Theorem \ref{GRtheorem2} to the diagram
\[\xymatrix{
\mc M \ar[r]^\pi &\bun_G \ar[r] &(\bun_G)_{\mr{dR}} \\
\bun_G \ar[u]^\sigma \ar@{=}[ur] \ar[urr] &&
}\]
in formal moduli problems under $\bun_G$.  Because these maps are $H$-equivariant the twisted anchor map from the twist of $\bb T_{\bun_G/\mc M}$ still factors through the zero bundle, so is the zero map. 
\end{proof}

With this proposition in mind, we'll abuse notation and always refer to the twisted anchor map as $\mathbb{T}_{\bun_G /\mc M^Q}$, even if the twisting data does not preserve the fibers of $\pi$.  In some examples we can promote this anchor map to a unique Lie algebroid, and therefore to a unique formal algebraic gauge theory.

\begin{definition}
A deformation $L'$ of a Lie algebroid $L$ on a derived stack $\mc X$ is a Lie algebroid on $\mc X \times \bb A^1$ such that the moduli problem under $\mc X$ corresponding to $L$ via Theorem \ref{GRtheorem2} and the moduli problem under $\mc X$ obtained by restricting the moduli problem associated to $L'$ to $\mc X \times \{0\}$ coincide. 
\end{definition}

\begin{lemma} \label{twist_construction_lemma}
If the twisted family $\mathbb{T}_{\bun_G /\mc M^{\lambda Q}}  \in \text{IndCoh}(\bun_G) _{/ \mathbb{T}_{\bun_G }}$ for $\lambda \in \bb A^1$ is the image under the functor $\mr{Anch}$ of a deformation in $\text{LieAlgebroid}(\bun_G)$, deforming the Lie algebroid corresponding to $\mc M$ then there exists a formal moduli problem $\mc M'$ under $\mc M \times \bb A^1$, corresponding to a deformation of a Lie algebroid, with respect to the twisting data $(\alpha, Q)$.  If this object $\mr{Anch}^{-1}(\mathbb{T}_{\bun_G / \mc M^Q})$ is unique up to equivalence then so is the twisted derived stack $\mc M^Q$, among formal algebraic gauge theories.
\end{lemma}

\begin{proof}
This is a direct application of Theorem \ref{GRtheorem2}.
\end{proof}

\begin{remark} \label{canonical_twist_remark}
If, in addition, one can find a map $\mc M' \to \bb A^1$ so that the composite $\bb A^1 \to \bb A^1 \times \bun_G \to \mc M' \to \bb A^1$ is the identity for every closed point $P$ of $\bun_G$, then $\mc M'$ has the structure of a twist as in Definition \ref{twist_definition}.  We observed, following Definition \ref{two_twist_defs} that there is automatically such a map when $\mc M$ is fiberwise formal and the twisting data preserves the fibers of $\sigma$.  There will also naturally be such a map for examples built from the Hodge stack.  We will call such twists -- when they exist and are essentially unique -- \emph{canonical twists}.
\end{remark}

For reference later, we should spell out exactly what we've shown for fiberwise formal theories -- i.e. in situations where we twist a Lie algebra object, and the twisted theory does not develop a non-trivial anchor map.

\begin{corollary} \label{fiberwise_action_corollary}
If $\mc M$ is a fiberwise formal algebraic gauge theory acted on by twisting data $(\alpha,Q)$ preserving the fibers of the map $\pi \colon \mc M \to \bun_G$, then there exists a canonical twist $\mc M^Q$, which is itself a fiberwise formal algebraic gauge theory.
\end{corollary}

As well as fiberwise formal theories and twisting data preserving the fibers of $\pi$, we'll use Lemma \ref{twist_construction_lemma} for the following simple example.  A deeper understanding of the anchor map functor would allow for a more general theorem ensuring the existence of canonical twists of fiberwise formal algebraic gauge theories: i.e. twists of sheaves of Lie algebras into Lie algebroids.

\begin{example}
If $\mc M = T[1] \bun_G $ and the twisting data acts as a non-vanishing degree 1 vector field, then $\mc M^Q$ is $(\bun_G)_{\mr{dR}}$. This follows because the vector field amounts to $\mr{id} \colon \bb{T}_{\bun_G} \rightarrow \bb{T}_{\bun_G }$ as ind-coherent sheaves over $\bun_G$. Note that this object is the terminal object in $\text{IndCoh}(\bun_G) _{/ \mathbb{T}_{\bun_G}  }$ so is the image under the functor $\mr{Anch}$ of a unique Lie algebroid.  In this case there is a natural map $(\bun_G)_{\mr{Hod}} \to \bb A^1$, realizing $(\bun_G)_{\mr{dR}}$ as a twist of $T[1]\bun_G$.
\end{example}

Having defined a twisting procedure for fiberwise formal algebraic gauge theories, let's investigate the properties enjoyed by these twisted theories.  The twisted theory $\mc M^Q$ retains only a limited amount of supersymmetry: it is acted on by the $Q$-cohomology of the full supersymmetry algebra.  More precisely, we have the following statement at the perturbative level, which immediately implies an analogous result non-perturbatively.

\begin{prop}\label{action_on_twisted_theories_lemma}
Suppose twisting data $(\alpha, Q)$ comes from the action of a supersymmetry algebra $\mc A$.  The action of the Chevalley--Eilenberg cochains $C^\bullet(\mc A)$ on the theory $\mc E$ defines an action of $C^\bullet(H^\bullet(\mc A, Q))$  on the twisted theory $\mc E^Q$, where we think of $Q$ as a fermionic endomorphism of cohomological degree 0 acting on $\mc A$, and hence on $C^\bullet(\mc A)$.  Furthermore the action of the translation algebra factors through the action of this algebra.
\end{prop}

\begin{remark}
In particular, this tells us that $Q$-exact translations act trivially in the twisted theory.
\end{remark}

\begin{proof}
We use the fact that, since $\mc A$ acts by symmetries, $[A,B](\phi) = A(B(\phi)) - B(A(\phi))$.  First we'll show that the $\mc A$ action on $\mc E$ induces an $\mc A$-action on $\mc E^Q$ which is well-defined up to $Q$-exact symmetries.  Let $\phi$ and $\phi + Q\psi$ be equivalent fields in $\mc E^Q$, and let $A \in C^1(\mc A)$ be a symmetry.  The action of $A$ on $\phi + Q\psi$ is by
\begin{align*}
A(\phi + Q\psi) &= A\phi + AQ\psi \\
&= A\phi + QA\psi - [Q,A]\psi \\
&= A\phi - [Q,A]\psi
\end{align*}
since $QA\psi = 0$ in $\mc E^Q$.  This expression in turn equals $A\phi$ up to $Q$-exact elements of the supersymmetry algebra, so this yields a well-defined action of the $Q$-closed symmetries in $\mc A$.

Now, let $A = [Q, \lambda] \in C^1(\mc A)$ be a $Q$-exact symmetry.  The action of $A$ on a field $[\phi]$ in $\mc E^Q$ is by
\begin{align*}
A \phi &= [Q, \lambda] [\phi] \\
&= Q\lambda [\phi] - \lambda Q [\phi] \\
&= 0 - \lambda(0) = 0
\end{align*}
since $Q\lambda \phi$ and $Q \phi$ vanish in $\mc E^Q$.  Note that here we're using the well-defined action of $Q$-closed symmetries on $\mc E^Q$ from the previous paragraph, so if $\phi \in \mc E$ has $Q$-cohomology class $[\phi]$ then $[\lambda[\phi]] = \lambda [\phi]$.  In particular $Q[\phi] = [Q\phi] = [0]$.  Thus we've shown that $Q$-exact symmetries act trivially, which means we have a well-defined action of $H^\bullet(\mc A, Q)$ on $\mc E^Q$ as required.

For the last statement we only need to note that the action of translations on $\mc E^Q$ by pushing forward along infinitesimal symmetries of spacetime agrees with the action of translations given here (which is well-defined since all translations are $Q$-closed) by construction of the twist.
\end{proof}

We focus now on the two types of twist that we're principally interested in: holomorphic and topological twists.  
\begin{definition}
A classical perturbative field theory $E$ on $\RR^n$ is called \emph{topological} if it is translation invariant; That is if the action of the Lie algebra $\RR^n$ on the sheaf $\mc E$ by translations is homotopically trivial.  The theory $E$ is called \emph{holomorphic} if the analogous condition holds for the Lie algebra of holomorphic vector fields for a specified complex structure on $\RR^n$. 
\end{definition}

\begin{prop} \label{top_supercharge_top_twist}
If $Q$ is a topological (resp. holomorphic) supercharge, then the twisted perturbative theory $E^Q$ is topological (resp. holomorphic).
\end{prop}

\begin{proof}
If $Q$ is topological, then by definition all translations are $Q$-exact, so vanish in the $Q$-cohomology.  The action of translations is given by a cochain map from the Chevalley--Eilenberg cohomology
\[a \colon C^\bullet(\CC^n) \to \eend(\mc E^Q(\RR^n)).\]
This action factors through the action of the full supersymmetry algebra, i.e. through the map $C^\bullet(\CC^n) \to C^\bullet(\mc A)$ induced by projection onto the translations in the supersymmetry algebra.  Now apply Proposition \ref{action_on_twisted_theories_lemma}, and note that all translations must act trivially.

The holomorphic case proceeds identically.
\end{proof}

\section{Constructing Supersymmetric Gauge Theories} \label{constructing_theories_section}
We'll discuss two procedures for constructing supersymmetric gauge theories in four dimensions: \emph{dimensional reduction from 10 dimensions} and \emph{compactification from a supertwistor space}.  In this section we'll review both constructions for $N=4$ theories (though analogous constructions also give rise to theories in dimensions other than 4, and theories with less supersymmetry). The idea of dimensional reduction was developed by Cremmer and Scherk \cite{CremmerScherk} and by Scherk and Schwarz \cite{ScherkSchwarz} in the 1970's, and the application we're most concerned with is the construction of $N=4$ supersymmetric gauge theory in four dimensions from $N=1$ gauge theory in ten dimensions by Brink, Schwarz, and Scherk \cite{BrinkSchwarzScherk}.  We currently don't have a fully rigorous definition of dimensional reduction for our notion of classical field theories, so the construction via dimensional reduction from 10 dimensions should be thought of as motivational, while the construction via compatification from twistor space should be thought of as a true definition.

Before getting into the specifics we'll recall the general ideas behind compactification and dimensional reduction for classical field theories.  Throughout this section a classical field theory $\mc M$ will be a family of derived stacks with a shifted symplectic structure on the global section as in Definition \ref{classical_field_theory_def}.

\begin{definition}
If $p \colon X \to Y$ is a smooth and proper map of smooth complex varieties, then the \emph{compactification} of the theory along $p$ of a classical field theory $\mc M$ on $X$ is the pushforward assignment $p_*\mc M$. 
\end{definition}

\begin{prop}
The compactification of a classical field theory $\mc M$ is still a classical field theory.  
\end{prop}

\begin{proof}
We just have to note that the global sections of compactified theories still carry shifted symplectic structures compatibly with the structure maps, and that the shifted tangent complex at a point is still a perturbative classical field theory.  The survival of the shifted symplectic structure under the compactification along $p \colon X \to Y$ is obvious, since $p_*\mc M(Y) = \mc M(p^{-1}Y) = \mc M(X)$ by definition.  The shifted tangent complex certainly retains its invariant pairing coming from this symplectic pairing, and it retains the structure of an elliptic $L_\infty$ algebra, so it forms a perturbative field theory.  
\end{proof}

\begin{definition}[Definition sketch]
The \emph{dimensional reduction} of a classical field theory $\mc M$ on a smooth variety $X$ along a fiber bundle $p \colon X \to Y$ whose fiber is a homogeneous space for an algebraic group $G$ is the classical field theory on $Y$ whose sections on an open set $U \sub Y$ are the $G$-invariants in $\mc M(p^{-1}U)$ under the action induced from the $G$-action on the fibers of $p$.
\end{definition}

This definition is currently unsatisfactory; we expect to have to impose additional conditions on the theory and the fibration for the theory obtained by taking invariants to remain a classical theory.  As such, we'll refer to dimensional reduction purely in an informal sense.

\begin{remark}
Costello \cite[19.2.1]{CostelloSH} uses the term ``dimensional reduction'' for what we call ``compactification'', and he requires an additional piece of structure.  He requires perturbative classical field theories to arise as the sections of a finitely generated complex of vector bundles, which is broken by the pushforward.  Thus he defines the compactification to consist of a finitely generated complex of vector bundles whose sections carry the structure of a perturbative classical field theory as we define it, along with a homotopy equivalence to the pushforward of a perturbative classical field theory on $X$.  For our purposes we won't need this finiteness condition, so this subtlety won't arise. 
\end{remark}

It'll also be important to understand how compactification and twisting relate to one another.  If the compactified theory $p_* \mc M$ is locally supersymmetric as in Section \ref{superspace_section} then the original theory $\mc M$ also admits an action of the supersymmetry algebra by four-dimensional local isometries fixing the fibers.  If the theory $\mc M$ was a fiberwise formal algebraic gauge theory then the compactification $p_*\mc M$ still defines a family of pointed formal moduli problems over $\bun_G$, i.e. there are a pair of maps $p_*\mc M(U) \rightleftarrows \bun_G(p^{-1}U)$ satisfying the hypotheses of Definition \ref{algebraic_gauge_theory_def}.

Therefore if we have twisting data $(\alpha,Q)$ for $\mc M$ then it makes sense to twist \emph{either} the original theory \emph{or} the compactified theory.  Denote these twisted theories by $\mc M^Q$ and $(p_* \mc M)^Q$ respectively.  We'll describe the relationship perturbatively.
\begin{lemma} \label{twist_dim_red_lemma}
If $\sigma \colon \bun_G \rightleftarrows \mc M \colon \pi$ is a fiberwise formal algebraic gauge theory and $\mc M^Q$ is a twist of $\mc M$ with respect to twisting data that preserves the fibers of $\pi$ and $\sigma$, then there exists a quasi-isomorphism of classical field theories
\[p_*(\mc M^Q) \iso (p_* \mc M)^Q.\]
\end{lemma}

\begin{proof}
By Corollary \ref{fiberwise_action_corollary} it suffices to check this at the level of perturbative field theories on $Y$, i.e. taking the shifted relative tangent complexes as sheaves of dg Lie algebras on $\bun_G$ over $Y$.  Write $p_*(\mc E^Q)$ and $(p_* \mc E)^Q$ for these two sheaves.  Fixing an open set $U \sub Y$ in the base, by definition $p_* (\mc E^Q)(U)$ is obtained as the local sections $\mc E^Q(p^{-1}U)$.  Likewise, $(p_* \mc E)^Q(U)$ is obtained by taking the space of local sections $\mc E(p^{-1}U)$ and applying the twisting procedure with respect to the specified twisting data, which also recovers the space of local sections $\mc E^Q(p^{-1}U)$, so the two sheaves coincide, thus so do the global derived stacks.
\end{proof}

\subsection{$N=1$ Super-Yang--Mills in Ten Dimensions} \label{10dSYM_section}

We'll now give an informal description of a supersymmetric ten-dimensional field theory in terms of fields and an action functional, while explaining the action of the supersymmetry algebra (as described in Appendix \ref{SUSY_appendix}) as clearly as possible.  Let $G$ be a complex reductive group with Lie algebra $\gg$ (we'll describe a complexification of the usual super Yang--Mills theory).  There are two fields $A$ and $\Psi$, where $A$ is identified with a $\gg$-valued 1-form and $\Psi$ is a Weyl fermion: a section of the bundle $S_{10+} \otimes \gg$.  The Lagrangian density can be identified with 
\[\mc L(A,\Psi) = \tr\left(\frac 12 F_A \wedge \ast F_A + \Psi \wedge \ast \sd D_A \Psi \right)\]
 where $F_A = dA + \frac 12[A,A]$, $D_A \Psi = d\Psi + [A,\Psi]$, and where we define the Dirac operator $\sd D_A$ using Clifford multiplication.  Here the trace is defined by means of a specified faithful finite-dimensional representation of $G$.  Define $\rho$ to be the Clifford multiplication map thought of as a map of vector bundles $T^*\CC^{10} \otimes S_{10+} \otimes \gg \to S_{10-} \otimes \gg$, using the metric to identify the tangent and cotangent bundles.  We define $\sd D_A = \rho \circ D_A$.  The trace pairing here implicitly includes both the invariant pairing on the Lie algebra and the ten-dimensional spinor pairing $S_{10-} \otimes S_{10+} \to \CC$.

One can describe $N=1$ super Yang--Mills in the homological formalism of Section \ref{classical_field_theory_section}, expanding a more familiar definition for Yang--Mills in the second order formalism to an $N=1$ vector multiplet.  Consider the elliptic complex
\[\xymatrix{
\Omega^0_\CC(\RR^{10}; \gg_P) \ar[r]^{\mr d} &\Omega^1_\CC(\RR^{10}; \gg_P) \ar[r]^{\mr{d}*\mr{d}} &\Omega^{9}_\CC(\RR^{10}; \gg_P) \ar[r]^{\mr d} &\Omega^{10}_\CC(\RR^{10}; \gg_P) \\
&\Omega^0_\CC(\RR^{10}; S_{10+} \otimes \gg_P) \ar[r]^{\ast \sd {\mr d}} &\Omega^{10}_\CC(\RR^{10}; S_{10-} \otimes \gg_P) &}
\]
in degrees $0, 1, 2$ and 3, where we write $\Omega^i_\CC(\RR^{10})$ for the complexification $\Omega^i(\RR^{10}) \otimes_\RR \CC$.  This complex admits an invariant pairing built from the wedge-and-integrate pairing on forms and the ten-dimensional spinor pairing between $S_{10+}$ and $S_{10-}$.  There is a natural $L_\infty$-structure coming from the action, for which the pairing is invariant.  The only non-trivial brackets are given by the action of $\Omega^0(\CC^{10} ;\gg_P)$ on everything, the degree two brackets 
\begin{align*}
\ell_2^{\mr{Bos}} \colon \Omega^1_\CC(\RR^{10};\gg_P) \otimes \Omega^1_\CC(\RR^{10};\gg_P) &\to \Omega^{9}_\CC(\RR^{10};\gg_P) \\
(A \otimes B) &\mapsto [A \wedge \ast \mr d B] + [\ast \mr d  A \wedge B] + \mathrm{d} \ast[A \wedge B] \\
\ell_2^{\mr{Fer}} \colon \Omega^1_\CC(\RR^{10};\gg_P) \otimes \Omega^0_\CC(\RR^{10}; S_{10+} \otimes \gg_P) &\to \Omega^{10}_\CC(\RR^{10}; S_{10-} \otimes \gg_P) \\
(A \otimes \Psi) &\mapsto \ast \sd A \Psi
\end{align*}
and the degree three bracket
\begin{align*}
\ell_3 \colon \Omega^1_\CC(\RR^{10};\gg_P) \otimes \Omega^1_\CC(\RR^{10};\gg_P) \otimes \Omega^1_\CC(\RR^{10};\gg_P) &\to \Omega^{9}_\CC(\RR^{10};\gg_P) \\
(A \otimes B \otimes C) &\mapsto [A \wedge \ast[B \wedge C]] + [B \wedge \ast[C \wedge A]] + [C \wedge \ast[A \wedge B]].
\end{align*}

Now, we must define the action of the supersymmetry algebra.  The bosonic piece acts by isometries on $\CC^{10}$ itself, and on the fields by pullback.  The fermions $S_{10+}$ act by supersymmetries; we choose $\eps \in S_{10+}$ and consider the infinitesimal symmetry coming from $\eps$, $(A, \Psi) \mapsto (A + \delta A, \Psi + \delta \Psi)$.  We let
\begin{align*}
\delta A &= \Gamma(\Psi, \eps) \\
\delta \Psi &= \rho^2(F_A \otimes \eps)
\end{align*}
where $\Gamma$ is the usual pairing $S_{10+} \otimes S_{10+} \to \CC^{10}$, fiberwise (and again using the metric to identify vector fields and 1-forms), and where $\rho^2$ denotes the composite map 
\[\Omega^2_\CC(\RR^{10}) \otimes S_{10+} \to \Omega^1_\CC(\RR^{10})^{\otimes 2} \otimes S_{10+} \to \Omega^1_\CC(\RR^{10}) \otimes S_{10-} \to S_{10+}\]
where the first map is the natural inclusion, and the latter maps are Clifford multiplication.  That this gives a well-defined action of the supersymmetry algebra, at least on-shell, and that the Lagrangian is supersymmetric are proven in \cite{ABDHN}. 
\begin{remark}
The on-shell condition here will require some care to treat rigorously.  Rather than giving a well-defined Lie algebra action on the space of fields, the supersymmetry relations only hold up to terms that vanish after imposing the equations of motion.  A priori this should give a well-defined homotopy action on the derived space of solutions to the equations of motion.  A careful analysis of this action is beyond the scope of this paper. 
\end{remark}

Now, by the calculations above, considering the subspace of fields constant along the leaves of a foliation by six-dimensional affine subspaces produces a four-dimensional theory with $N=4$ supersymmetry.  This theory is called (pure) \emph{$N=4$ super Yang--Mills} in four dimensions.  One can explicitly describe the fields and the action functional \cite{BrinkSchwarzScherk} in this dimensionally reduced theory.  The gauge field $A$ breaks into a four-dimensional gauge field (which we'll also call $A$) and six scalar fields $\phi_1,\ldots, \phi_6$.  The Weyl spinor $\Phi$ breaks into four four-dimensional Dirac spinors $\chi_1, \ldots \chi_4$.  When we construct an $N=4$ from the twistor space perspective we'll observe that the field content is the same (one can also define an action on super twistor space which recovers the dimensionally reduced action functional here.  This was done by Boels, Mason, and Skinner \cite{BMS}).

\subsection{Twistor Space Formalism} \label{twistor_section}
Twistor space is a complex manifold whose geometry is closely related to that of (compactified) Minkowski space.  At its root, twistor space $\PT$ is just the complex manifold $\CC \PP^3$, but we can describe it in a way that explains why it might be related to the geometry of $\RR^{1,3}$.  Write $\bb T$ for the Dirac spinor representation $S = S_- \oplus S_+$ in signature $(1,3)$, a 4-complex-dimensional vector space.  This new notation is chosen for compatibility with the twistor literature.  The \emph{twistor space} $\PT$ is then the space of complex lines in $\bb T$.

\begin{remark}
Elsewhere when discussing four-dimensional spinors we've used Euclidean signature, and indeed since we're only discussing complex spinor representations here our classical field theories don't depend on a choice of signature.  We've used the language of Lorentzian signature in the above construction of twistor space because of certain other aspects of twistor theory that appear in the literature, for instance the existence of the Penrose correspondence between the space of null twistors and complexified Minkowski space, that suggest that twistors are really most naturally related to Lorentzian geometry.
\end{remark}

Fix a Hermitian inner product on the space $S_+$ of Weyl spinors.  The space $\bb T = S_- \oplus S_+$ therefore admits a pseudo-Hermitian structure by
\[((\alpha_1,\beta_1),(\alpha_2, \beta_2)) \mapsto \langle \alpha_1, \ol \beta_2 \rangle + \langle \ol \beta_1, \alpha_2 \rangle\]
using the canonical isomorphism $S_- \iso \ol S_+$, which we observe has signature $(2,2)$.  This is called the \emph{twistor norm}.  The space of twistors with vanishing twistor norm is denoted $\NN \sub \bb T$ and forms a seven-real-dimensional submanifold.  Looking at complex lines contained in $\NN$ defines $\PN \sub \PT$, a five-real-dimensional compact submanifold.  Removing this submanifold splits $\PT$ into two components, $\PT_+$ and $\PT_-$ corresponding to twistors with positive and negative twistor norm respectively.

There are two natural maps associated to twistor space which we should describe.  First define the \emph{Penrose map} associated to an identification $S_+ \iso \bb H$ with the quaternions to be the map
\[p \colon \PT \iso \bb{CP}^3 \to \bb{HP}^1 \iso S^4\]
with fibers isomorphic to $\bb{CP}^1$ (the \emph{twistor lines}).  The space of null twistors $\PN$ maps to an equator $S^3 \sub S^4$.  We choose a point in $p(\PN)$ as a ``point at infinity''.  The preimage $\PT \bs \bb{CP}^1$ of the complement is isomorphic to $\bb{CP}^1 \times \RR^4$ as a smooth manifold.

For concreteness, choose homogeneous coordinates $Z_0, Z_1, Z_2, Z_3$ on $\bb T$.  The Penrose map is then given by
\[(Z_0:Z_1:Z_2:Z_3) \mapsto (Z_0 + j Z_1:Z_2+j Z_3).\]
Say the point at infinity is $(1:0) \in \bb{HP}^1$.  The complement of the twistor line at infinity is the set $\{(Z_0:Z_1:Z_2:Z_3) \mid Z_2 \text{ and } Z_3 \text{ are not both } 0\}$.  This allows us to define a \emph{holomorphic} map
\begin{align*}
\pi \colon \PT \bs \bb{CP}^1 &\to \bb{CP}^1 \\
(Z_0:Z_1:Z_2:Z_3) &\mapsto (Z_2:Z_3).
\end{align*}
In more coordinate-free language we can identify $\PT \bs \bb{CP}^1$ with the total space of the rank 2 holomorphic vector bundle $\OO(1) \oplus \OO(1) \to \PP(S_+)$.  The map $\pi$ is the bundle map.

\begin{remark}
This is an instance of a more general construction due to Atiyah, Hitchin, and Segal \cite{AHS} that makes sense starting from any pseudo-Riemannian 4-manifold $X$ satisfying a certain curvature condition.  In short, one can take the total space of the projectivized negative Weyl spinor bundle $\PP(S_+)$ over $X$, and produce a canonical almost complex structure on this total space using the Clifford multiplication.  This almost complex structure is integrable if one imposes the appropriate curvature condition.  In the case where $X = \RR^{1,3}$ is Minkowski space we obtain the total space of the trivial $\PP(S_+)$ bundle, and the complex structure one defines is precisely the complex structure on $\PT \bs \bb{CP}^1$ defined above. 
\end{remark}

The twistor space itself admits a supersymmetric extension.
\begin{definition}
The \emph{super twistor space} associated to a complex vector space $W$ is the total space of the odd vector bundle
\[\PT^W = \Pi(\OO(1) \otimes W) \to \PT.\]
\end{definition}
If we restrict to the preimage of $\RR^4$ under the Penrose map $p$, we find a superspace which admits a natural action of the supersymmetry algebra $\mc A^W$ (where the R-symmetries act trivially).  We'll construct supersymmetric field theories on $\RR^4$ by compactification from theories on twistor space admitting manifest supersymmetry actions.

\subsection{Holomorphic Chern--Simons Theory on Super Twistor Space} \label{hCS_section}
The power of the twistor space formalism lies in its ability to relate theories involving the \emph{holomorphic} or \emph{algebraic} geometry of (super) twistor spaces, and the metric geometry of 4-manifolds.  We'll recall two types of theory modelling the theory of holomorphic principal bundles.  First, let $X = (\Pi E \to X_{\mr{even}})$ be a split algebraic supermanifold of complex dimension $n$, let $G$ be a complex reductive group, and let $P$ be a principal $G$-bundle on $X_{\mr{even}}$. 

The following theory of $BG$ valued holomorphic maps was discussed in \cite[Section 11.2]{CostelloSH} (as an instance of a more general theory of holomorphic maps into a complex target stack).  It will be an \emph{analytic} perturbative field theory, i.e. a sheaf of complexes over a complex manifold $X$ with respect to its analytic topology.  Lacking a good theory of derived analytic geometry we won't be able to literally promote this to a non-perturbative field theory, we'll only be able to describe an analogous theory using algebraic bundles and the Zariski topology.

\begin{definition}
The \emph{curved $\beta \gamma$ system} on $X$ (with target $BG$) near a holomorphic $G$-bundle $P$ is the cotangent theory, as in Definition \ref{cotangent_theory_definition}, whose base is the elliptic $L_\infty$ algebra
\[\Omega^{0,\bullet}(X_{\mr{even}}; \OO_{\Pi E} \otimes \gg_P), \ol \dd).\]
Hence the underlying elliptic complex is $(\Omega^{0,\bullet}(X_{\mr{even}}; \OO_{\Pi E} \otimes \gg_P) \oplus \Omega^{n,\bullet}(X_{\mr{even}};\OO_{\Pi E^\vee} \otimes \gg_P^*[n-3]), (\ol \dd, \ol \dd))$, and the invariant pairing is given by the canonical pairings between $\gg$ and $\gg^*$ and between $E$ and $E^\vee$, and the wedge pairing on forms.  
\end{definition}

This perturbative description ought to arise as a description of the cotangent theory to the moduli space of holomorphic $G$-bundles on $X$, because the Dolbeault complex with coefficients in $\gg_P$ controls deformations of holomorphic $G$-bundles on $X$.  This suggests an analogous \emph{algebraic, non-perturbative} version of the classical field theory.
\begin{definition}
The (algebraic) \emph{curved $\beta \gamma$ system} on $X$ (with target $BG$) is the cotangent theory whose local sections on $U \sub X$ are given by the derived stack
\[T^*[-1]\bun_G(U).\]
If $X$ is smooth and proper -- so $\bun_G(U)$ is finitely presented -- the global sections admit a natural shifted symplectic structure.
\end{definition}

\begin{remark} \label{cotangent_space_remark}
Since $\bun_G(U)$ is not locally of finite presentation for general $U$, its cotangent complex is generally not perfect and hence one cannot define the (shifted) cotangent bundle as in the conventions section. On the other hand, one can always define the total space of a given quasi-coherent sheaf $\mathcal{F}$ on $X$ in terms of the moduli problem whose $R$-points consists of maps $f \colon  \spec R  \to X$ together with sections $\Gamma ( \spec R, f^* \mathcal{F})$. We won't make this technical definition precise here; we're most interested in describing the global sections of classical field theories on smooth projective varieties $X$. This remark should also be applied for later appearances of a cotangent space of a derived stack which is not locally of finite presentation.
\end{remark}

In either the analytic or the algebraic setting we could instead consider a more general theory of holomorphic or algebraic maps into any target -- this would define a more general curved $\beta\gamma$-system.  

Starting from $N=1$ and $N=2$ super twistor space, one constructs supersymmetric gauge theories by taking the curved $\beta \gamma$ system on the complement of a twistor line in the super twistor spaces $\PT^{N=1}$ or $\PT^{N=2}$.  For $N=4$ super Yang--Mills however we'll do something different: we observe that the complex $\Omega^{0,\bullet}(X; \gg_P)$ where $X$ is the complement of the line in $N=4$ super twistor space (i.e. the restriction of the odd vector bundle defining super twistor space to $\PT \bs \bb{CP}^1 \sub \PT)$ \emph{already} admits a degree $-3$ invariant pairing, and so defines a field theory.  This is an instance of a more general family of theories.

\begin{example}
Let $X$ be a compact super Calabi--Yau variety of complex dimension $n|m$, as in Definition \ref{CY_def}.  Then the complex $\Omega^{0,\bullet}(X; \gg_P)$ admits a degree $-n$ invariant pairing by the invariant pairing on $\gg$ and the wedge pairing on forms.  This pairing naturally lands in the Berezinian, which yields a density by applying the Calabi--Yau structure, an isomorphism of vector bundles $\mr{Ber}(X) \to \ul \CC$.  If $n=3$, this defines a perturbative field theory on $X$ which we call \emph{holomorphic Chern--Simons theory}.   This perturbative theory admits an algebraic non-perturbative analogue, as above.  One can consider the non-perturbative algebraic classical field theory $\EOM(U) = \bun_G(U)$, with $(-1)$-shifted symplectic structure arising via the derived AKSZ formalism \cite[Theorem 2.5]{PTVV} from the 2-shifted symplectic structure on $BG$ and the Calabi--Yau structure on $X$.
\end{example}

\begin{remark}
There's a certain amount of ambiguity in the terminology for these classical field theories.  The theory we call the curved $\beta\gamma$ system with target $BG$ is itself called holomorphic Chern--Simons theory in \cite{CostelloWittenGenusI}.  In the case where $X$ is a super Calabi--Yau 3-fold then the two theories are closely related: the holomorphic Chern--Simons theory (in our terminology) has the curved $\beta \gamma$ system as its cotangent theory, as in the book of Costello and Gwilliam \cite{CostelloGwilliam2}.
\end{remark}

Now, let $X = \PT^{N=4} \bs \bb{CP}^1$, the complement of a line in $N=4$ super twistor space.  One observes (as noted by Witten \cite{WittenTST}) that this space is super Calabi--Yau by computing the Berezinian.  More generally, the Berezinian of the super projective space $\bb{CP}^{n|m}$ is computed to be
\begin{align*}
\mr{Ber}_{\bb{CP}^{n|m}} &\iso K_{\bb{CP}^n} \otimes_\OO \wedge^m(\OO(1) \otimes \CC^m) \\
&\iso \OO(-n-1) \otimes \OO(m) \iso \OO(m-n-1)
\end{align*}
(using a choice of trivialization of $\wedge^m \CC^m$) which is trivial if and only if $m = n+1$, for instance in the case $n=3, m=4$. 

\begin{remark} \label{non_compact_remark}
We should note that while $\bb{CP}^{3|4} \bs \bb{CP}^1$ is super Calabi--Yau, it is not \emph{compact} super Calabi--Yau.  While holomorphic Chern--Simons on $\PT^{N=4}$ is a genuine classical field theory as in Definition \ref{classical_field_theory_def} with shifted symplectic structure on the space $\bun_G(\PT^{N=4})$ of global solutions to the equations of motion given by the derived AKSZ formalism, the shifted symplectic form fails to be well-defined on the complement of a line.  We expect at least a shifted Poisson structure to survive here, but since we won't need this shifted symplectic structure for the untwisted $N=4$ moduli space in what follows -- we'll construct the twisted theories of interest on $\RR^4$, then generalize to arbitrary smooth algebraic surfaces by analogy -- we'll ignore this subtlety in the present work.
\end{remark}

Let's try to understand the theory we get when we perform compactification along the map $p \colon \PT^{N=4} \bs \bb{CP}^1 \to \RR^4$.  Specifically let's verify that the field content agrees with the fields we described at the end of Section \ref{10dSYM_section}.  Our argument will follow the argument for the ordinary Penrose--Ward correspondence given by Movshev \cite{Movshev}, and cohomology calculations given in Section 7.2 of the book of Ward and Wells \cite{WW}.  We'll use the phrase \emph{linearised} holomorphic Chern--Simons and $N=4$ super Yang--Mills to mean the perturbative field theories obtained by forgetting the brackets in the $L_\infty$ structure, leaving only a cochain complex.  We'll do this calculation for the analytic, perturbative theory.

\begin{remark}
Note that we needed to trivialize $\wedge^4 \CC^4$ in order to define the super Calabi--Yau structure.  This choice breaks the full $\gl(4;\CC)$ of R-symmetries to $\sl(4;\CC)$, as we remarked in Section \ref{supercharge_section}.
\end{remark}

\begin{prop}
The compactification of linearised holomorphic Chern--Simons theory along the Penrose map $p$ is equivalent to the linearised anti-self-dual $N=4$ super Yang--Mills theory.
\end{prop}

\begin{proof}
To show this, we need to pushforward the sheaf of solutions to the classical equations of motion in the holomorphic Chern--Simons theory along $p$.  This sheaf is just the complex $\Omega^{0,\bullet}(X; \gg_P)$ where $X$ is the complement of the line in $N=4$ super twistor space.  That is, the complex
\begin{align*}
\bigoplus_{i \ge 0} \left(\Omega^{0,\bullet}(\PT \bs \bb{CP}^1; \sym^i(\Pi \OO(-1)^4) \otimes_\OO \gg_P)\right) \iso \ & \bigoplus_{i \ge 0} \left(\Omega^{0,\bullet}(\PT \bs \bb{CP}^1; \wedge^i(\OO(-1)^4) \otimes_\OO \gg_P)\right) \\
\iso \ & \Omega^{0,\bullet}(\PT \bs \bb{CP}^1; (\OO \oplus \OO(-4)) \otimes_\OO \gg_P) \\
&\oplus \ \Omega^{0,\bullet}(\PT \bs \bb{CP}^1; (\OO(-1) \oplus \OO(-3)) \otimes_\OO \gg_P)^4 \\
&\oplus \ \Omega^{0,\bullet}(\PT \bs \bb{CP}^1; \OO(-2) \otimes_\OO \gg_P)^6.
\end{align*}
We've grouped the terms here judiciously -- they'll yield the gauge field, four spinor fields and six scalar fields we saw in Section \ref{10dSYM_section} respectively (with their corresponding antifields).  To check this, we must compute the hypercohomology of these terms, complete with their actions of the algebra $\so(4;\CC)$.  This becomes a little simpler after identifying $\PT \bs \bb{CP}^1$ with the total space of the rank two holomorphic vector bundle $\OO(1) \otimes S_- \to \PP(S_+)$.  What's more, the pullback of the bundle $\OO(k)$ on $\PP(S_+)$ under the map $\pi$ is precisely the vector bundle $\OO(k)$ given by restriction from $\PT = \bb{CP}^3$.  From this point of view we can identify
\begin{align*}
\Omega^{0,\bullet}(\PT \bs \bb{CP}^1; \OO(k) \otimes \gg_P) &\iso \pi^*\left( \bigoplus_{i+j=\bullet}\Omega^{0,i}(\PP(S_+); \OO(k) \otimes \gg_P \otimes \wedge^j(\OO(1) \otimes S_-))\right), \\
\text{so } p_*(\Omega^{0,\bullet}(\PT \bs \bb{CP}^1; \OO(k) \otimes \gg_P)) &\iso \Omega^0(\RR^4) \otimes \left( \bigoplus_{i+j=\bullet}\Omega^{0,i}(\PP(S_+); \OO(k) \otimes \gg_P \otimes \wedge^j(\OO(1) \otimes S_-))\right)
\end{align*}
as a sheaf on $\RR^4$.  We then compute the hypercohomology of the right hand side, which is just the cohomology of the coefficient coherent sheaf with an additional differential.  Indeed, we can think of the complex as bigraded by the $i$ and $j$ gradings, and the cohomology of the coefficient coherent sheaf is precisely the $E_1$ page of the spectral sequence of the double complex.  This page has form
\[C^\infty(\RR^4) \otimes \left(\vcenter{\xymatrix{H^0(\PP(S_+); \OO(k) \otimes \gg_P) \ar[r] &  H^0(\PP(S_+); \OO(k+1) \otimes S_- \otimes \gg_P) \ar[r] & H^0(\PP(S_+); \OO(k+2) \otimes \gg_P) \\
H^1(\PP(S_+); \OO(k) \otimes \gg_P) \ar[r] &H^1(\PP(S_+); \OO(k+1) \otimes S_- \otimes \gg_P) \ar[r] & H^1(\PP(S_+); \OO(k+2)  \otimes \gg_P).
}}\right)\]
The page is concentrated in a single row and therefore the spectral sequence converges at the $E_2$ page unless $k=-2$, in which case there's one additional differential (from $(i,j)=(1,0)$ to $(0,2)$) and the complex converges at the $E_3$ page.

We begin with the first line (the term of interest in the ordinary, non-supersymmetric Penrose--Ward correspondence, and the term considered by Movshev \cite{Movshev}).  The coefficient sheaf is isomorphic to $((\OO \oplus (\OO(1) \otimes S_-)[-1] \oplus \OO(2)[-2]) \oplus (\OO(-4) \oplus (\OO(-3) \otimes S_-)[-1] \oplus \OO(-2)[-2])) \otimes \gg_P$  whose cohomology is $\gg_P \oplus \gg_P[-1] \otimes (S_- \otimes S_+ \oplus \sym^2S_+) \oplus \gg_P[-2] \otimes (\sym^2 S_+ \oplus S_- \otimes S_+) \oplus \gg_P[-3]$.  Thus the corresponding term in the pushforward sheaf is
\[\Omega^0(\RR^4; \gg_P \otimes (\CC \oplus (V \oplus \sym^2 S_+)[-1] \oplus (V \oplus \sym^2S_+)[-2] \oplus \CC[-3]))\]
where $V \iso S_+ \otimes S_-$ is the vector representation of $\so(4;\CC)$.  To compute the differential, we start with the first summand in the pushforward sheaf, $\Omega^0(\RR^4) \otimes H^0(\OO \oplus (\OO(1) \otimes S_-)[-1] \oplus \OO(2)[-2])$.  This is the $E_1$ page of the spectral sequence of the double complex described above, and the differential is the image of the $\ol \dd$ operator.  Concretely, in coordinates this operator has form $\dd_i e^i$, where $x^i$ is a basis for $\RR^4$, $\dd_i = \frac \dd{\dd x^i}$, and $e^i$ is a degree 1 operator on $H^0(\OO \oplus (\OO(1) \otimes S_-)[-1] \oplus \OO(2)[-2])$ associated to $x^i$.  This operator arises by canonically identifying $H^0(\OO(1) \otimes S_-)$ with $V = \RR^4 \otimes_\RR \CC$ so that every global section of $\OO(1) \otimes S_-$ yields a degree 1 operator on $H^0(\wedge^\bullet(\OO(1) \otimes S_-))$ via the natural map $\wedge^\bullet(H^0(\OO(1) \otimes S_-)) \to H^0(\wedge^\bullet(\OO(1) \otimes S_-))$.  Unpacking this calculation, we find exactly the differential in the Atiyah--Singer--Donaldson complex
\[\Omega^0(\RR^4) \overset d \to \Omega^1(\RR^4) \overset {d_+} \to \Omega^2_+(\RR^4).\]
controlling an anti-self-dual connection.  The remaining summand is Serre dual to the first summand, so the overall complex is the complex controlling an anti-self-dual Yang--Mills field as required.

Similarly, we analyse the second line.  Now, the coefficient sheaf is isomorphic to $((\OO(-1) \oplus (\OO \otimes S_-)[-1] \oplus \OO(1)[-2]) \oplus (\OO(-3) \oplus (\OO(-2) \otimes S_-)[-1] \oplus \OO(-1)[-2]))\otimes \gg_P$, whose cohomology is $\gg_P[-1] \otimes (S_- \oplus S_+) \oplus \gg_P[-2] \otimes (S_- \oplus S_+)$ with the $\so(4;\CC)$ action indicated by the notation.  Thus the corresponding term in the pushforward sheaf is
\[(\Omega^0(\RR^4; \gg_P \otimes (S[-1] \oplus S[-2])))^4\]
where $S = S_+ \oplus S_-$.  We analyse the differential in a similar way to the above, focusing on the first summand $\Omega^0(\RR^4) \otimes H^0(S_-[-1] \oplus \OO(1)[-2])$ (the other term is Serre dual to this one).  Again, in a specified basis, the differential is of the form $\dd_i e^i$, where now the $e^i$ act according to the action of $x^i \in H^0(\OO(1) \otimes S_-)$ on the complex $H^0(\sym^\bullet(\OO(1) \otimes S_-) \otimes \OO(-1))$.  Unpacking, this action map (from $\sym^1$ to $\sym^2$) is given by the composite
\[S_- \overset {x^i \otimes 1} \to V \otimes S_- \iso S_+ \otimes S_- \otimes S_- \surj S_+ \otimes \wedge^2S_- \iso S_+.\]
This composite is exactly the Clifford multiplication $\rho(x^i)$ by the vector $x^i$, so our overall differential is $\dd_i \rho(x^i)$.  This is the Dirac operator $\sd d$, so combining this term with its Serre dual we obtain the complex
\[\left(\Omega^0(\RR^4; S) \overset {\sd d} \to \Omega^0(\RR^4; S)\right)^4\]
in degrees one and two, which is the linearised BV complex controlling four Dirac spinors.

Finally, we analyse the last line, which is the simplest algebraically, but whose differential is a little more subtle than the others.  The coefficient sheaf is isomorphic to $(\OO(-2) \oplus (\OO(-1) \otimes S_- )[-1] \oplus \OO[-2]) \otimes \gg_P$, whose cohomology is $\gg_P[-1] \oplus \gg_P[-2]$ with the trivial $\so(4;\CC)$ action.  Thus the corresponding term in the pushforward sheaf is
\[(\Omega^0(\RR^4; \gg_P[-1] \oplus \gg_P[-2]))^6.\]
To compute the differential we have to do a little more than we did for the earlier terms, because now the $E_1$ and $E_2$ pages of the spectral sequence coincide, but there's a differential on the $E_2$ page increasing the $j$ degree by two.  This differential is of the form $D = \dd_i\dd_j(e^i \ol \dd^{-1} e^j)$, where the operator $e^i \ol \dd^{-1} e^j$ is obtained from the composite 
\[\xymatrix{
&H^0(S_- \otimes \OO(1)) \otimes \Gamma(\Omega^{0,0}_{\bb{CP}^1}(-1) \otimes S_-) \ar[r] &\Gamma(\Omega^{0,0} \otimes \wedge^2S_-) \iso H^0(\Omega^{0,0}_{\bb{CP}^1}) \\
H^0(S_- \otimes \OO(1))^{\otimes 2} \otimes \Gamma(\Omega^{0,1}_{\bb{CP}^1}(-2)) \ar[r] &H^0(S_- \otimes \OO(1)) \otimes \Gamma(\Omega^{0,1}_{\bb{CP}^1}(-1) \otimes S_-) \ar[u]_{1 \otimes \ol \dd^{-1}} &
}\]
(where we've used $\Gamma$ for the global sections of the infinite-type vector bundles $\Omega^{i,j}$ to emphasise that we're considering all forms, not just the Dolbeault cohomology) applied to $x^i \otimes x^j \in H^0(S_- \otimes \OO(1))^{\otimes 2}$ and a representative for a cohomology class in $H^{0,1}(\bb{CP}^1; \OO(-2))$.  Here we use the fact that the operator $\ol \dd \colon \Omega^{0,0}_{\bb{CP}^1}(-1) \to \Omega^{0,1}_{\bb{CP}^1}(-1)$ induces an isomorphism on $H^0$. To compute the operator $e^i \ol \dd^{-1} e^j$ we follow the method of \cite[Theorem 7.2.5]{WW}.  There is a map of complexes
\[\xymatrix{
0 \ar[r] &\OO(-2) \ar[r]^-{e^i} \ar@{=}[d] &\OO(-1) \otimes S_- \ar[r]^-{e^i} \ar@{=}[d] &\OO \otimes \wedge^2 S_- \ar[d]^{\delta^{ij}} \ar[r] &0\\
&\OO(-2) \ar[r]^-{e^i} &\OO(-1) \otimes S_- \ar[r]^-{e^j} &\OO \otimes \wedge^2 S_-
}\]
where the top row is exact, which yields a map between the spectral sequences computing the hypercohomology of the two rows.  On the $E_2$ page of these spectral sequences, this map just yields a commutative square
\[\xymatrix{
H^0(\Omega^{0,1}_{\bb{CP}^1}(-2)) \ar[r] \ar[d]^{\mr{id}} &H^0(\Omega^{0,0} \wedge^2S_-) \ar[d]^{\delta^{ij}} \\
H^0(\Omega^{0,1}_{\bb{CP}^1}(-2)) \ar[r]^{e^i\ol \dd^{-1}e^j}  &H^0(\Omega^{0,0} \wedge^2S_-), 
}\]
and the top arrow is an isomorphism because the corresponding sequence of complexes was exact, so the operator $e^i \ol \dd^{-1} e^j$ is obtained from $\delta^{ij}$ by a change of coordinates, and the second order operator $D$ is conjugate to the Laplacian, as required.
\end{proof}

\begin{remark}
In the above calculation we've computed the BV complex for a perturbative classical field theory on $\RR^4$ as a cochain complex with a pairing only.  We \emph{haven't} described the pushforward of the $L_\infty$ structure.  In other words we've shown that we obtain the expected quadratic terms in the action for an anti-self-dual $N=4$ gauge theory, but we haven't checked that the correct interaction terms appear.  In what follows we take the compactification of holomorphic Chern--Simons on twistor space as the \emph{definition} of untwisted $N=4$ anti-self-dual super Yang--Mills.
\end{remark}

We won't investigate the action in detail, but the holomorphic Chern--Simons action functional yields an \emph{anti-self-dual} super Yang--Mills theory after compactifying the twistor lines.  There's an extra term that we can introduce into the action, of form
\[S_2(A) = \int_{\RR^{4|8}} d\mu \log \det (\ol \dd|_{p^{-1}(\mu)}).\]
Boels, Mason, and Skinner \cite{BMS} prove that the holomorphic Chern--Simons theory on $N=4$ super twistor space with this additional term incorporated into the action recovers $N=4$ super Yang--Mills after compactifying along the twistor lines. 

\begin{remark} \label{non_holomorphic_issue}
We run into trouble when we try to define untwisted $N=4$ super Yang--Mills theory non-perturbatively via compactification along the twistor fibers, because the Penrose map $p$ is \emph{not} holomorphic for any complex structure on $\RR^4$.  As such, a Zariski open set $U \sub \CC^2$ does not lift to a Zariski set $p^{-1}(U) \sub \PT \bs \bb{CP}^1$.  This is not a problem in the analytic setting; any open set in a complex manifold admits a canonical complex structure, but generally not an \emph{algebraic} structure.  It is not particularly surprising that we encounter such problems: there's no reason that a metric-dependent theory like untwisted $N=4$ gauge theory should admit a description purely in terms of algebraic geometry.
\end{remark}

\section{Equations of Motion in the Twisted Theories} \label{classical_states_section}
We'll now investigate the form of the classical field theories obtained from applying our holomorphic and topological twists to this $N=4$ theory.  The holomorphic twist will be the simplest, conceptually: the holomorphic twisting data is compatible with the structure of $\bun_G(\PT^{N=4})$ as a fiberwise formal algebraic gauge theory over $\PT \bs \bb{CP}^1$, so a canonical holomorphic twist exists by Corollary \ref{fiberwise_action_corollary}, which can naturally be thought of as a fiberwise formal algebraic gauge theory over $\CC^2$, and which generalizes to describe a fiberwise formal algebraic gauge theory over a \emph{compact} complex algebraic surface $X$ whose global sections are given by
\[\EOM_{\mr{hol}}(X) \iso T_{\mr{form}}[1]\ul{\mr{Map}}(\Pi TX, BG).\]
The A and B topological twists are more subtle, because they each break structures that survive the holomorphic twist: the B-twist breaks the section $\bun_G(U) \to \EOM_{\mr{hol}}(U)$, while the A-twist breaks the projection map $\EOM_{\mr{hol}}(U) \to \bun_G(U)$.  However, we'll construct natural twists using Example \ref{de_rham_stack_twist_example}: the A-twist deforms the outer shifted tangent bundle to the de Rham prestack, while the B-twist deforms the source of the mapping stack to $X_{\mr{dR}}$, yielding the cotangent theory to the moduli of $G$-local systems.

\subsection{The Holomorphic Twist} \label{holo_twist_section}
First, recall that according to the superspace formalism, to define the holomorphically twisted theory we need to specify a complex structure on a 4-manifold.  The perturbative piece of this calculation is contained in Costello's 2011 paper \cite{CostelloSH}, but is included here for the reader's convenience.  Recall that a \emph{$G$-Higgs bundle} on a complex variety $X$ is an algebraic $G$-bundle $P$ equipped with a section $\phi \in H^0(X, T^*_X \otimes \gg_P )$ such that $[\phi, \phi]=0$.  We'll write $\higgs_G(X)$ for the moduli stack of $G$-Higgs bundles, and $\higgs_G^{\mr{fer}}(X)$ for the moduli stack of $G$-Higgs bundles where the Higgs field is placed in \emph{fermionic degree} (so the underlying bosonic piece is just $\bun_G(X)$).  This moduli space is described by the mapping stack $\ul{\mr{Map}}(\Pi TX, BG)$.

The Penrose--Ward correspondence tells us that $N=4$ anti-self-dual super Yang--Mills corresponds to the compactification of holomorphic Chern--Simons on super twistor space along the Penrose map $p$, where the bundles are constrained to be trivializable along the twistor lines.  As we remarked in \ref{non_holomorphic_issue} this is problematic when working algebraically, because the map $p$ is not holomorphic, so the compactification is not well-defined.  We'll motivate a definition of holomorphically twisted $N=4$ theory by computing the twist of the holomorphic Chern--Simons theory (since, by Lemma \ref{twist_dim_red_lemma} the compactification of this twist is the desired twist of $N=4$ theory).

We use the following trick: find a closed embedding $\iota \colon Z \sub \PT \bs \bb{CP}^1$ such that the Penrose map $p$ maps $Z$ diffeomorphically onto $\RR^4$.  We define the compactification of an algebraic gauge theory along $p$ to be the restriction of the theory to $Z$.

First, we'll check that the twisting data we've been discussing preserves the fibers of the maps $\sigma$ from $\bun_G$ and $\pi$ to $\bun_G$ as in Corollary \ref{fiberwise_action_corollary}, so the twist remains fiberwise formal.

\begin{prop}
The twisting data associated to the holomorphic twist preserve the fibers of the zero section map $\sigma \colon \bun_G(U) \to \EOM(U)$ and the projection map $\pi \colon \EOM(U) \to \bun_G(U)$ for an open set $U \sub \CC^2$, as in Definition \ref{preserves_fibers_def}.
\end{prop}

\begin{proof}
We can check the holomorphic twist preserves the fibers at the super twistor space level.  For holomorphic Chern--Simons theory on super-twistor space the relevant map $\pi\colon \bun_G(\PT^{N=4} \bs \bb{CP}^1) \to \bun_G(\PT \bs \bb{CP}^1)$ is given by pulling back under the zero section of the super vector bundle $\Pi \OO(-1)^4$.  The twisting data acts by pairing with a section of $\OO(1) \inj (\OO(-1)^4)^*$, the dual to the first factor, which acts on the fibers by multiplication by that section in the coefficient $\sym(\Pi(\OO(1)^4))$.  In particular, the fibers are preserved, so the twisting data acts trivially on the image of $d\pi$.  Also pairing with such a section preserves the zero-section of the bundle over $\PT$, thus the image of the section $\sigma$ and therefore the twisting data acts trivially on the image of $d\sigma$.
\end{proof}

As such, we can compute the holomorphic twist by computing the restricted relative shifted tangent complex as a sheaf over $\bun_G$, twisting the fibers, and applying Gaitsgory--Rozenblyum's theorem as in Corollary \ref{fiberwise_action_corollary}.  

\begin{theorem} \label{holo_twist_theorem}
The solutions to the equations of motion in the holomorphically twisted $N=4$ SYM theory on $\CC^2$ near an open set $U$ are given by
\[\EOM_{\mr{hol}}(U) \iso T^*_{\mr{form}}[-1]\higgs^{\mr{fer}}_G(U).\]
\end{theorem}

Note that Remark \ref{cotangent_space_remark} applies for this theorem for general open sets $U$.  The choice of holomorphic supercharge we made corresponds to a choice of complex structure on the base space $\RR^4$ of the Penrose map. For concreteness, let us note that for a holomorphic $G$-bundle $P$ on $U \subset \CC^2$, thought of as a Higgs bundle with trivial Higgs field, one has
\[\bb T_{P}[-1]\higgs^{\mr{fer}}_G(U) \iso  \OO(U; \gg_P) \oplus \Omega^{\ge 1}_{\mr{alg}}(U; \gg_P) \iso \Omega^{\natural}_{\mr{alg}}(U; \gg_P),\]
with zero differential, where $\Omega^p_{\mr{alg}}$ is naturally in fermionic degree $p \mod 2$ and cohomological degree 0. Here the first summand of the complex describes deformations of the holomorphic bundle $P$ and the second summand describes deformations of the Higgs field $0 \in \Pi \Omega^1_{\mr{alg}}(U; \gg_P)$. We will see in the proof that the homomorphically twisted theory is the cotangent theory with the base $\higgs^{\mr{fer}}_G(U) $, namely, \[\bb T_{P}[-1]\EOM_{\mr{hol}}(U) = \Omega^\natural_{\mr{alg}}(U; \gg_P) \oplus  \Omega^\natural_{\mr{alg}}(U; \gg_P)^\vee[-3]\]
with the Lie algebra structure being the base acting on the shifted cotangent fiber in a canonical way.

\begin{remark}
A priori, the twists of the full $N=4$ super Yang--Mills theory and its anti-self-dual piece might differ.  However, this is actually not the case.  In the appendix of Costello's paper on supersymmetric field theories \cite{CostelloSH} it is shown that the $Q_{\mr{hol}}$ twist of perturbative $N=4$ anti-self-dual Yang--Mills doesn't admit any deformations as a perturbative field theory.  If the twist of the full theory differed from the twist of the anti-self-dual theory, then there would be a path of twisted theories deforming one into the other (by sending the additional term in the action for the full theory to zero), thus a non-trivial deformation of the perturbative theory.  Hence we can compute our twist using twistor space without worrying about the additional Boels--Mason--Skinner term in the action: this is guaranteed to be $Q_{\mr{hol}}$-exact.
\end{remark}

\begin{proof}
We'll begin with a summary of the global structure of the proof. First, in view of Lemma \ref{twist_dim_red_lemma} we'll compute the twist of holomorphic Chern--Simons theory on super twistor space. This amounts to computing the shifted tangent complex and performing the twisting construction to get a new family over $\PT \bs \bb{CP}^1$, with the structure of a family of pointed formal moduli problems over $\bun_G$.  In order to obtain the compactified theory on $\CC^2$, we will use the trick described above: we'll find a closed embedding $\iota \colon Z \sub \PT \bs \bb{CP}^1$ such that the Penrose map $p$ induces a diffeomorphism $Z \cong \RR^4$ (and hence defines a complex structure on $\RR^4$) and define the compactification to be the restriction of the family from $\PT \bs \bb{CP}^1$ to $Z$. Since the result is a family over $Z \cong \CC^2$ of pointed formal moduli problems over $\bun_G$, the above computation determines the moduli space of solutions in the twisted, compactified theory, using Theorem \ref{GRtheorem}. 

We will compute the twisted theory at the level of twistor space.  Choose an open set $U \sub \PT \bs \bb{CP}^1$, an affine derived scheme $V$, and a smooth map $f \colon V \to \bun_G(U)$.  The shifted tangent complex at the map $f$ to the $N=4$ super twistor space theory was canonically quasi-isomorphic to
\[\Gamma(p^{-1}(U) \times V; f^*\gg) \iso \Gamma(U \times V; \pi_1^*\sym(\Pi \OO(-1)|_U^4) \otimes f^*\gg),\] 
where we write $f^*\gg$ to denote the sheaf of Lie algebras on $U \times V$ obtained by pulling back $\gg = \bb T[-1]BG$ under a closed point $f$ of $\ul{\mr{Map}}(V, \bun_G(U)) \iso \ul{\mr{Map}}(U \times V, BG)$, and where $\pi_1 \colon U \times V \to U$ is the projection.  From now on we'll just write $\OO(k)$ for the restriction $\OO(k)|_U$ when our arguments are independent of $U$.  Recall that when we twist we modify the sections of our theory over $U$ by adding a $\CC^\times$ weight to the cohomological grading then introducing a new differential coming from the supercharge.  We'll choose a $\CC^\times$-action such that the first copy of $\OO(-1)$ (corresponding to $e_1^* \in W^*$) has weight $-1$, the third copy of $\OO(-1)$ (corresponding to $f_1^* \in W^*$) has weight 1, and the remaining two copies (corresponding to $e_2^*$ and $f_2^* \in W^*$) have weight 0.

The holomorphic supercharge $Q_{\mr{hol}} = \alpha_1 \otimes e_1$ can be thought of as a section of $\Pi \OO(1)$ which pairs non-trivially with the first factor of $\OO(-1)^4$ (generated by $e_1^* \in W^*$) to define a map $\OO(-1)^4 \to \OO$, which extends to a sym-degree $-1$ derivation of $\sym(\Pi\OO(-1)^4)$.  The section of $\Pi \OO(1)$ in question, corresponding to $\alpha_1 \in S_+$, is given on the open set $U$ by the homogeneous polynomial $Z_2$ in twistor coordinates, so the differential given by $Q_{\mr{hol}}$ is generated by the map that multiplies a section of $\OO(-1)$ on the set $U$ by $Z_2$.  This preserves the cohomological grading, but increases the weight by 1, since it reduces the number of $e_1^*$ factors by 1.  The map ``multiply by $Z_2$'' from $\OO(k)$ to $\OO(k+1)$ is injective, and has cokernel isomorphic to $\OO_Z(k+1) = \iota_* \OO_{\{Z_2 = 0\}}(k+1)$ where $\OO_Z$ is the structure sheaf of the zero locus of $Z_2$.  Thus we compute the $Q_{\mr{hol}}$-twisted shifted tangent complex to be the space of global sections of the sheaf
\[\pi_1^*\left(\OO_Z \oplus \OO_Z(-2) \oplus \Pi \OO_Z(-1)^2\right) \otimes f^*\gg \oplus \pi_1^*\left(\OO_Z(-3) \oplus \OO_Z(-1) \oplus \Pi \OO_Z(-2)^2\right) \otimes f^*\gg[-1],\]
arising from the cohomology of the operator
\[\xymatrix{
{}^{-1} && \OO(-1)_1 \ar[dl] &\Pi\OO(-2)^2_{12,14}\ar[dl] &\OO(-3)_{124}\ar[dl] & \\
{}^0 &\OO \quad  &\Pi\OO(-1)^2_{2,4} &\OO(-2)_{13} \oplus \OO(-2)_{24}\ar[dl] &\Pi\OO(-3)^2_{123,134}\ar[dl] &\OO(-4)_{1234}\ar[dl] \\
{}^1 &&\OO(-1)_3 &\Pi\OO(-2)^2_{23,34} &\OO(-3)_{234} &
}\] 
where in the diagram cohomological degree runs vertically, and the subscripts represent symmetric products of the four factors of $\Pi\OO(-1)^4$.  This result actually defines the BV complex of a cotangent theory whose base is the first factor -- $\pi_1^*\left(\OO_Z \oplus \OO_Z(-2) \oplus \Pi \OO_Z(-1)^2\right) \otimes f^*\gg$ -- alone, since there is a canonical quasi-isomorphism of complexes
\[\left(\OO(k-1)[1] \to \OO(k)\right)^! \iso \left(\OO(-k-4) \to \OO(-k-3)[-1]\right)[3]\]
for each $k$ by identifying the sheaf of densities with $\OO(-4)[3]$ (the canonical sheaf shifted so that its cohomology is concentrated in degree zero) -- where the morphisms are given by pairing with the section $\alpha_1$ of $\OO(1)$ -- and therefore an invariant pairing on $\gg$ provides an isomorphism of coherent sheaves
\[\pi_1^*\left(\OO_Z(-3) \oplus \OO_Z(-1) \oplus \Pi \OO_Z(-2)^2\right) \otimes f^*\gg \iso \pi_1^*\left(\left(\OO_Z \oplus \OO_Z(-2) \oplus \Pi \OO_Z(-1)^2\right) \otimes f^*\gg\right)^![-2].\]
Since the original Lie algebra structure comes from the tensor product of sheaves and the Lie algebra structure on $f^*\gg$ (in the diagram, this pairs objects with their reflection through the center, with complementary subscripts), the induced Lie structure is that of a cotangent theory, using the nondegenerate invariant pairing. 

After identifying $\OO_Z(-1)^2 \iso \Omega^1_{Z, \mr{alg}}$ by choosing a trivialization, we obtain an isomorphism of coherent sheaves of graded Lie algebras
\[\pi_1^* \left(\OO_Z \oplus \OO_Z(-2) \oplus \Pi \OO_Z(-1)^2\right) \otimes f^*\gg \iso (\pi_1')^*\Omega^{\natural}_{Z, \mr{alg}} \otimes f^*\gg \]
over $Z$, where $\Omega^{1}$ is fermionic but in cohomological degree 0 and $\pi_1'$ is the projection $Z \times V \to Z$. With this, the holomorphically twisted shifted tangent complex becomes \[\Gamma( U \times V ; (\pi_1')^* \Omega^\natural_{Z,\mr{alg}} \otimes f^*\gg) = \Omega_{\mr{alg}}^\natural ( (U \cap Z) \times V ; f^*\gg) \] where we abuse notation to write $f^*\gg$ both for the sheaf on $U \times V$ and for its restriction to $(U \cap Z) \times V$. 

Now, we have to compactify the twisted complex along the Penrose map. We might worry that this is undefined since $p$ is not holomorphic, but we note that $p$ maps $\{Z_2 = 0\}$ diffeomorphically onto $\RR^4$ and henceforth identify $Z$ as $\RR^4$ (thus in particular defining a complex structure on $\RR^4$). Then for $U \subset \CC^2$, and a smooth map $f \colon V \to \bun_G(U)$, one obtains the shifted tangent complex of the cotangent theory whose base is perturbatively given by $\Omega^{\natural}_{\mr{alg}}(U \times V; f^*\gg)$ with zero differential, and where $\Omega^i$ is placed in fermionic degree $i \mod 2$.

It remains to globalize our computation using Theorem \ref{GRtheorem}. By definition of the tangent complex as a quasi-coherent sheaf, it is enough to check that for any affine derived scheme $V$ over $\bun_G(U)$, the local sections on $V$ of the restricted shifted tangent complexes to $\EOM_{\mr{hol}}(U)$ and $T^*[-1]\higgs^{\mr{fer}}_G(U)$ are equivalent as dg Lie algebras.  This is exactly what we checked above: the local sections on $V$ of the restricted tangent complex to $\higgs_G^{\mr{fer}}(U)$ are precisely given by $\Omega^{\natural}_{\mr{alg}}(U \times V; f^*\gg)$ with zero differential, and with $\Omega^i$ in fermionic degree $i \mod 2$, so the calculation above of the restricted shifted tangent complex to the holomorphically twist moduli space provides the desired dg Lie algebra equivalence for each $f$.  Thus we obtain an equivalence
\[\bb T_{\EOM_{\mr{hol}}(U) } [-1]\iso \bb T_{T^*_{\mr{form}}[-1]\higgs_G^{\mr{fer}}(U)} [-1]\]
of sheaves of dg Lie algebras, and therefore by Theorem \ref{GRtheorem} an equivalence of derived stacks as required.
\end{proof}

\begin{remark}
If we were working in an analytic framework, we could do this calculation by literally compactifying along the twistor lines.  If $U \sub \CC^2$ is an \emph{analytic} open set then its pullback $p^{-1}U$ to twistor space admits a canonical complex structure despite $p$ not being holomorphic.
\end{remark}

Given the above calculation, we can \emph{define} the holomorphic twist of $N=4$ theory on any complex proper algebraic surface $X$ using the superspace formalism of Section \ref{superspace_section}.  

\begin{definition} \label{holo_twist_definition}
The \emph{holomorphically twisted $N=4$ theory} on a complex proper algebraic surface $X$ is the assignment of derived stacks with
\[\EOM_{\mr{hol}}(U) = T^*_{\mr{form}}[-1]\higgs^{\mr{fer}}_G(U)\]
where $U \sub X$ is a Zariski open set, with the canonical $(-1)$-shifted symplectic structure on the global sections. 
\end{definition}

\subsection{The B-twist} \label{B_twist_section}
We'll now proceed to compute the B-twist of $N=4$ super Yang--Mills on a complex proper algebraic surface $X$.  This will again be a cotangent theory, but now to the moduli space $\Flat_G(X)$ of $G$-bundles with \emph{flat connection}.  As before, we'll compute the B-twist on flat space first -- computing the twist of the holomorphically twisted theory on $\CC^2$ with respect to the further B supercharge -- then note that the superspace formalism allows us to extend the theory to one on general complex (proper) algebraic surfaces.

Unlike the example of the holomorphic twist in the previous section, the B supercharge will preserve the fibers of the projection map $\pi \colon \EOM_{\mr{hol}}(U) \to \bun_G(U)$, but \emph{not} of the section $\sigma \colon \bun_G(U) \to \EOM_{\mr{hol}}(U)$.  As such we will not be able to directly apply Theorem \ref{GRtheorem2} to describe a canonical twist.  Instead, we'll observe that the moduli space $\EOM_{\mr{hol}}(U)$ has the structure of a mapping space, and the twisting data acts on the source of the mapping space alone, which \emph{does} admit a natural deformation describable by Theorem \ref{GRtheorem2}, yielding a natural B-twist.

We begin by describing $\EOM_{\mr{hol}}(U)$ in a slightly different way.  Using the language of the Hodge prestack, as in Example \ref{de_rham_stack_twist_example}, we can rewrite the moduli space of solutions to the equations of motion in the holomorphic twist in a way natural for constructing our further A- and B-twists.  There is a $\CC^\times$ action $\alpha$ on $\EOM_{\mr{hol}}(U)$, which acts on the base space $\higgs^{\mr{fer}}_G(U)$ of the shifted cotangent bundle in a way that on the fibers of the projection $\higgs^{\mr{fer}}_G(U) \to \bun_G(U)$ it does with weight minus one by rescaling the Higgs field.  

\begin{definition}
We'll write $\higgs_G^{\mr{bos}}(U)$ for the formal completion
\[\higgs_G^{\mr{bos}}(U) = \higgs_G(U)^\wedge_{\bun_G(U)} = \ul{\mr{Map}}(T[1]U, BG)^\wedge_{\ul{\mr{Map}}(U, BG)}.\]
\end{definition}

The superscript ``bos'' (for bosonic) is intended to contrast with the fermionic Higgs moduli space of the previous section, and to remind the reader that this formal Higgs moduli space differs slightly from the definition that more normally appears in the literature. 

\begin{lemma} \label{holo_as_dolbeault_lemma}
The regrading of the moduli space $\EOM_{\mr{hol}}(U)$ for a smooth surface $U$ with respect to this $\CC^\times$-action $\alpha$ is equivalent to the mapping stack
\[T^*_{\mr{form}}[-1] \ul{\mr{Map}}(U_{\mr{Dol}}, BG)^\wedge_{\bun_G(U)} \iso T^*_{\mr{form}}[-1] \higgs_G^{\mr{bos}}(U).\]
\end{lemma}

\begin{proof}
We saw in Theorem \ref{holo_twist_theorem} for the surface $\CC^2$, which we used as a definition for more general surfaces, that 
\begin{align*}
\EOM_{\mr{hol}}(U) &\iso T_{\mr{form}}^*[-1]\higgs_G^{\mr{fer}}(U) \\
&\iso T_{\mr{form}}^*[-1]\ul{\mr{Map}}(\Pi TU, BG) \\
&\iso T_{\mr{form}}^*[-1](\ul{\mr{Map}}(\Pi TU, BG)^\wedge_{\bun_G(U)}).
\end{align*}
The $\CC^\times$-action we've described acts on the fiber of $\Pi TU$ with weight one, so the regraded space is equivalent to
\[\EOM_{\mr{hol}}^\alpha(U) \iso  T_{\mr{form}}^*[-1](\ul{\mr{Map}}(T[1]U, BG)^\wedge_{\bun_G(U)}).\]
In turn, the shifted tangent bundle $T[1]U$ is equivalent to $U_{\mr{Dol}}$ (because $U$ is a smooth scheme, so $T[1]U \iso T_{\mr{form}}[1]U$), so $\EOM_{\mr{hol}}^\alpha(U) \iso T^*_{\mr{form}}[-1]\higgs_G^{\mr{bos}}(U)$ as required.
\end{proof}

\begin{remark}
The formal completion at $\bun_G(U)$ is necessary for the bosonic but not the fermionic Higgs moduli space because, while the fibers of the map $\higgs_G^{\mr{fer}}(U) \to \bun_G(U)$ are purely fermionic, and therefore formal, the map $\higgs_G(U) \to \bun_G(U)$ has non-formal fibers, so the map is not a nil-isomorphism.  Taking the completion while we regrade is necessary for the regraded theory to still be a formal algebraic gauge theory.
\end{remark}

Now, let's describe a twist of the holomorphic theory with respect to the B-supercharge.  The idea is that, viewing $\EOM_{\mr{hol}}(U)$ as a mapping space as in Lemma \ref{holo_as_dolbeault_lemma} we can canonically deform the \emph{source} from $U_{\mr{Dol}}$ to $U_{\mr{dR}}$, for instance by applying Theorem \ref{GRtheorem2} to the symmetry generated by a non-vanishing degree one vector field on $T[1]U$.  This will contrast with the A-twist in the next section, where we'll deform the global shifted cotangent bundle construction in a similar way.  


\begin{theorem} \label{B_twist_EOM_theorem}
The algebraic classical field theory $\EOM_B$ which assigns to a complex algebraic surface $U$ the derived stack
\[\EOM_B(U) = T^*_{\mr{form}}[-1]\Flat_G(U)\]
arises as a natural deformation of $\EOM_{\mr{hol}}(U)$ which, if $U = \CC^2$, defines a twist of $N=4$ super Yang--Mills theory with respect to the topological supercharge $Q_B$. 
\end{theorem}

\begin{remark}
As we noted in Remark \ref{non_compact_remark}, this theory is only a true algebraic classical field theory according to Definition \ref{classical_field_theory_def} if $U$ is proper, ensuring that $\Flat_G(U)$ is finitely presented, so has a perfect tangent complex.  In general the theory exists as an assignment of (possibly infinite type) derived stacks, but the presymplectic form on the shifted cotangent complex may be degenerate.
\end{remark}

\begin{proof}
We'll build a canonical twist as discussed in Remark \ref{canonical_twist_remark}.  More specifically, we'll describe a deformation of the regrading $\EOM_{\mr{hol}}^\alpha(U)$ for a general surface $U$, then observe that if $U$ is a Zariski open subset of $\CC^2$ then it satisfies the conditions of Definition \ref{twist_definition}. 

For a fixed complex algebraic surface $U$, we consider the derived stack $\mc M' (U)= T^*_{\mr{form}}[-1]\ul{\mr{Map}}_{\bb A^1}(U_{\mr{Hod}}, BG \times \bb A^1)$: the formal shifted cotangent to the mapping stack relative to $\bb A^1$.  This admits a flat map to $\bb A^1$ whose fiber over $t$ is canonically equivalent to $T^*_{\mr{form}}[-1]\ul{\mr{Map}}(U_{t \text{-dR}}, BG)$ -- as in Example \ref{de_rham_stack_twist_example} -- so the general fiber is equivalent to $\mc M^{Q_B}(U) = T^*_{\mr{form}}[-1]\Flat_G(U)$, and whose fiber over zero is equivalent to $T^*_{\mr{form}}[-1]\higgs_G(U)$.   We've therefore defined a deformation of the regrading $\mc M^\alpha (U) = \EOM_{\mr{hol}}^\alpha(U)$, via the embedding $\EOM_{\mr{hol}}^\alpha(U) \to \higgs_G(U)$, whose general fiber is the desired twisted moduli space.

Now, we must check the hypotheses of Definition \ref{twist_definition}; that is, that for every closed point $P \in \bun_G(U)$ we can find a section $s$ such that $s(0) = \sigma_\alpha(P)$ and such that the relative shifted tangent complex agrees with the twist of the zero fiber as a perturbative field theory.  For every closed point of $\mc M^{Q_B} (U)$ -- just a closed point $A = (P,\nabla)$ of the base space $\Flat_G(U)$ -- there's a natural section $s \colon \bb A^1 \to \mc M'$ given by rescaling the connection, such that the shifted tangent complex restricted to $s$ is equivalent to the $\CC[t]$-module 
\[s^*\bb{T}_{\mc M'}[-1] = \left( ( \Omega^{\bullet}_{\mr{alg}}(U; \gg_P) \oplus \Omega^{\bullet}_{\mr{alg}}(U; \gg_P) ^\vee[-3] )  \otimes \CC[t] ), (td_A, td_A)\right)\]
where $d_A$ is the algebraic covariant derivative associated to the flat connection $\nabla$ on $U$.  This defines a twist of the perturbative field theory $\bb{T}_P[-1]\EOM_{\mr{hol}}(U) = \Omega^{\natural}_{\mr{alg}}(U; \gg_P) \oplus \Omega^{\natural}_{\mr{alg}}(U; \gg_P) ^\vee[-3]$ by the B-twisting data.   
\end{proof}

It is immediate to identify compactification of the twisted theory along an algebraic curve.

\begin{corollary}
For a product $\Sigma_1 \times \Sigma_2$ of algebraic curves, the B-twist of $N=4$ super Yang--Mills theory satisfies
\[\EOM_B(\Sigma_1 \times \Sigma_2 ) = T^*_{\mr{form}}[-1] \underline{\mr{Map}}( (\Sigma_1)_{\mr{dR}},  \Flat_G(\Sigma_2)). \]
\end{corollary}

\begin{proof}
This follows from the definition $\Flat_G(X) = \underline{\mr{Map}}(X_{\mr{dR}}, BG)$ and the adjunction 
\[\underline{\mr{Map}}(X \times Y , Z ) = \underline{\mr{Map}} (X , \underline{\mr{Map}}(Y,Z)).\]
\end{proof}

\begin{remark}
One can read this corollary as saying that the B-twisted theory compactifies to the B-model with target $\Flat_G(\Sigma_2)$. A completely perturbative description was given by Costello \cite{CostelloSH}, which was not enough to identify $\Flat_G(\Sigma_2)$ as an algebraic stack. One should note that here we identify the target as the moduli stack of de Rham local systems, as opposed to Betti local systems, which is more aligned with the usual formulation of the geometric Langlands correspondence. This result is somewhat surprising, because it has been widely believed that the Kapustin--Witten story can only capture the topological aspects of the correspondence. 

\vspace{-5pt}
One might worry that one shouldn't expect the twist by a \emph{topological} supercharge to depend on a choice of complex structure on spacetime, which our examples clearly do.  Because this theory on $X$ didn't necessarily arise from twisting a theory with respect to global topological twisting data, there's no reason that the moduli space $\EOM_B(U)$ shouldn't depend on a complex algebraic structure on $U$, and in general it \emph{does} depend on this choice.

\vspace{-5pt}
A more familiar example of this phenomenon is provided by Donaldson--Witten theory as a topological twist of $N=2$ super Yang--Mills.  While the theory on a flat space is truly topological, if one uses the superspace formalism to extend this theory to a general 4-manifold one finds that the moduli space of solutions to the equations of motion is built from the moduli space of instantons, which -- if $b_2^+ = 1$ -- may depend on the metric of the underlying 4-manifold, not just its diffeomorphism type.  A discussion in the physics literature can be found in the 1998 paper of Moore and Witten \cite{WittenMoore}.

\vspace{-5pt}
From the point of view of the current work, this subtlety is necessary if we intend to recover a statement as the geometric Langlands conjecture, which is dependent on changes in the algebraic/holomorphic structure on a curve from a topologically twisted theory.  We will return to this in future work. 
\end{remark}

\begin{remark}
In theories like the B-twist, we would like to be able to talk about the germs of solutions to the equations of motion near some (smooth) submanifold of positive real codimension, especially codimension 1 submanifolds of form $\Sigma \times S^1 $, where $\Sigma$ is an algebraic curve: these germs of solutions correspond to the classical phase space in the 2d theory obtained by compactification along $\Sigma$.  With the ideal, analytic Definition \ref{idealdefinition} of a classical field theory this would be possible: one could define the space of \emph{germs of solutions to the equations of motion} along a submanifold $Y \sub X$ to be the inverse image $\iota^{-1}\mc M$, where $\iota \colon Y \inj X$ was the inclusion map.  As we'll see, this would give very natural examples for an analytic version of the B-twisted classical field theory, but using our algebraic definition we'll need to use a slightly different construction.

\vspace{-5pt}
Suppose we indeed had an algebraic model for the holomorphically twisted $N=4$ theory with open sections on an \emph{analytic} open set $U$ given by $T^*[-1]\higgs_G^{\mr{fer}}(U)$, interpreted in some natural way. Then we could make a claim of the following sort.  
\begin{claim}
If $Y \sub X$ is a compact oriented codimension $k$ submanifold, then the germs of solutions to the equations of motion near $Y$ in a $B$-twisted $N=4$ theory are given by
\[\EOM_B(Y) = T^*_{\mr{form}}[k-1]\Flat_G(Y)\] 
where $\Flat_G(Y)$ is the space of germs of flat connections near $Y \sub X$. 
\end{claim}

\begin{proof}
To identify the moduli space of germs along $Y$ we choose a tubular neighborhood $U$ of $Y$ in $X$, and use Poincar\'e duality to identify the compactly supported sections of the shifted tangent complex on $Y$ with the compactly supported sections of the complex $\Omega^\bullet(U; \gg_P)[1]$ of $\Flat_G(U)$ plus a shift of its dual.  Indeed, global sections of the inverse image $\iota^{-1}\EOM_B(Y)$ are just compactly supported sections of $\EOM_B$ on a tubular neighborhood $U$ of $Y$. We have quasi-isomorphisms
\begin{align*}
\left(\Omega_c^\bullet(U; \gg_P)[1]\right)^\vee &\iso \left(\Omega^\bullet(Y; \gg_P)[1]\right)^\vee \\
 &\iso \Omega^\bullet(Y; \gg_P^*)[\dim Y - 1] \\
 &\iso \Omega^\bullet(Y; \gg_P^*)[3-k] \\
 &\iso \left(\Omega^\bullet(Y; \gg_P^*[1])[1]\right)[1-k] \\
\end{align*}
which gives the total compactly supported tangent complex $\Omega_c^\bullet(U; \gg_P \oplus \gg_P^*[1])[1]$ a $(k-1)$-symplectic structure which splits globally as the sum of a sheaf of complexes and a shift of its dual.  Thus, after an application of a version of Theorem \ref{GRtheorem} in analytic derived geometry we identify the moduli space of solutions with the appropriate shifted cotangent bundle.
\end{proof}
\vspace{-10pt}
We'll give an algebraic version of this claim for manifolds of form $\Sigma \times U$ for $U = S^1$ or $U= \mr{pt}$ below.
\end{remark}

As discussed in the remark, we would like to make sense of what a theory assigns to a submanifold of nonzero codimension. Because our framework uses an algebraic structure of a submanifold in an essential way -- we defined the B-twist by twisting theories only naturally defined for algebraic varieties -- we'll need to extend our formalism.  One observes that the base of the cotangent sheaf defining the B-twist can be described by $U \mapsto \Flat_G(U) = \mr{Map} (  U_{\mr{dR}} , BG)$ for $U \subset X$ and that this assignment makes sense for a more general class of derived stacks than just algebraic varieties.

Specifically, let's consider compact connected manifolds $U$ so that $U \times \Sigma$ has dimension less than four (formally, we're considering spaces of positive codimension for the 2-dimensional theory obtained by compactification along $\Sigma$): the only possibilities are the circle and the point.  These are modelled by derived stacks $S^1_B$ and $\mr{pt}$, so we will simply consider $\Sigma \times U \mapsto \ul{\mr{Map}}((\Sigma \times U ) _{\mr{dR}} , BG)$ for $U=S^1_B$ or $U=\mr{pt}$.

While it is natural to consider the assignment $V \mapsto \Flat_G(V)$ to such extended objects, the $(-1)$-shifted cotangent bundle is not: the degree of the shift must change depending on the dimension of $V$. In order to understand what this means, let us view $\mr{EOM}_B(X) = T_{\mr{form}}^*[-1] \Flat_G(X)$, where $X$ is a smooth and proper algebraic surface, as arising by applying Theorem \ref{GRtheorem} to a sheaf of dg Lie algebras over $\Flat_G(X)$ given by the dg Lie algebra equivalence
\begin{eqnarray*}
\mathbb{T}_{T_{\mr{form}}^*[-1] \Flat_G(X)  }[-1] & =&  \mathbb{T}_{\Flat_G(X) } [-1] \oplus  (\mathbb{T}_{\Flat_G(X) } [-1] )^\vee[-3]\\
& =&  \mathbb{T}_{\Flat_G(X) } [-1] \oplus  \mathbb{L}_{\Flat_G(X) } [-2]\\
&= &  \mathbb{T}_{\Flat_G(X) } [-1] \oplus  \mathbb{T}_{\Flat_G(X) },
\end{eqnarray*}
where we use the $(-2)$-shifted symplectic structure of $\Flat_G(X) = \ul{\mr{Map}}(X_{\mr{dR}}, BG)$ obtained from the AKSZ construction using the 4-orientation on $X_{\mr{dR}}$ to identify the $(-2)$-shifted cotangent complex with the tangent complex \cite[Theorem 2.5]{PTVV}. This is an equivalence of dg Lie algebras, where the second summand is treated as a module for the first summand.  We'll extend this description of the moduli space of solutions to the equations of motion, to \emph{define} the moduli space for spaces of form $\Sigma \times U$. 

\begin{definition} \label{phase_space_definition}
For $U= S^1_B$ or $U=\mr{pt}$, we define $\mr{EOM}_B(\Sigma \times U)$ on $X$ to be the derived stack obtained by applying the Theorem \ref{GRtheorem} to the sheaf  
$ \mathbb{T}_{\Flat_G(\Sigma \times U) } [-1] \oplus  \mathbb{T}_{\Flat_G(\Sigma \times U) }$ of Lie algebras over $\Flat_G(\Sigma \times U)$.
\end{definition}

\begin{corollary}\label{Bphasespacecorollary}
There is an equivalence of derived stacks 
\[\mr{EOM}_B (\Sigma \times S_B^1) = T^*_{\mr{form}} (\LL  \Flat_ G( \Sigma)).\]
\end{corollary}

\begin{proof}
By definition, it is enough to compare the shifted tangent complexes of $\mr{ EOM}_B (\Sigma \times S_B^1)$ and $T^*_{\mr{form}} (\LL  \Flat_ G( \Sigma))$ as sheaves of Lie algebras over $\LL  \Flat_ G( \Sigma)$.  There are Lie algebra equivalences
\begin{eqnarray*}
\mathbb{T}_{\mr{EOM}_B(\Sigma \times S^1_B)}[-1]  & = & \mathbb{T}_{\Flat_G(\Sigma \times S^1_B) } [-1] \oplus  \mathbb{T}_{\Flat_G(\Sigma \times S^1_B) }\\
& = & \mathbb{T}_{\Flat_G(\Sigma \times S^1_B) } [-1] \oplus  \mathbb{L}_{\Flat_G(\Sigma \times S^1_B) }[-1]\\
& =&  \mathbb{T}_{\LL  \Flat_ G( \Sigma) }[-1] \oplus(  \mathbb{T}_{\LL  \Flat_ G( \Sigma)) }[-1])^\vee [-2]\\
&=&  \mathbb{T}_{T^*_{\mr{form}} (\LL  \Flat_ G( \Sigma)) }[-1] 
\end{eqnarray*}
where we use the $(-1)$-shifted symplectic structure of $\Flat_G(\Sigma \times S^1_B) = \ul{\mr{Map}}((\Sigma \times S^1_B)_{\mr{dR}}, BG) \iso \ul{\mr{Map}}(\Sigma_{\mr{dR}} \times S^1_B, BG)$ provided by the AKSZ construction, using the 2-orientation on $\Sigma_{\mr{dR}}$ and the 1-orientation on $S^1_B$.
\end{proof}

Note that the result is a 0-shifted symplectic derived stack. This is an expected property of a \textit{phase space} in a classical field theory, i.e. the space the theory assigns to a proper codimension 1 submanifold.  According the Kapustin--Witten program, this space should -- under geometric quantization -- yield the Hochschild homology of the category the relevant \emph{extended} 2d topological quantum field theory assigns to the point, expected to be the category on the B-side of the geometric Langlands correspondence.  We intend to address this in the sequel to this work.

Finally, we can similarly understand what the B-twisted theory assigns to spaces of codimension 2.

\begin{corollary}
For a smooth projective curve $\Sigma$, the moduli space of germs of solutions to the equations of motion on $\Sigma \times \mr{pt}$ is given by
\[\EOM_B(\Sigma \times \mr{pt} ) \iso T^*_{\mr{form}}[1]\Flat_G(\Sigma).\]
\end{corollary}

\begin{proof}
The argument here is very similar to the computation of the phase space in corollary \ref{Bphasespacecorollary}.  We apply Theorem \ref{GRtheorem} to the sheaf 
\[\bb T_{\Flat_G(\Sigma)}[-1] \oplus \bb T_{\Flat_G(\Sigma)}\]
on $\Flat_G(\Sigma)$.  There are dg Lie algebra equivalences
\begin{align*}
\bb T_{\Flat_G(\Sigma)}[-1] \oplus \bb T_{\Flat_G(\Sigma)} &\iso \bb T_{\Flat_G(\Sigma)}[-1] \oplus \bb L_{\Flat_G(\Sigma)} \\
&\iso \bb T_{T^*_{\mr{form}}[1] \Flat_G(\Sigma)}
\end{align*}
using the 0-shifted symplectic structure on $\Flat_G(\Sigma)$.  Again, applying Theorem \ref{GRtheorem} completes the proof.
\end{proof}

\begin{remark} \label{derhamderham_remark}
In order to perform this calculation, we were forced to extend a natural calculation of $\EOM_B$ for algebraic varieties to spaces of form $\Sigma \times U_B$ by hand.  In order to obtain a theory compatible with geometric Langlands, as proposed by Kapustin and Witten, we are forced to perform this procedure, where we replace a theory which is ``de Rham'' in all four directions with a theory that is de Rham in two directions and Betti (purely topological) in the remaining two.  It is worth noting that these theories are very different: the purely de Rham theory is determined entirely by its local operators, whereas the de Rham-Betti theory admits non-trivial line operators (indeed, these are critical for the geometric Langlands program).  Having made this modification, one can go further to investigate a theory in which all four directions are topological; an understanding of such a theory should lead to a physical description of the ``Betti Langlands correspondence'' of Ben-Zvi and Nadler \cite{BZNBettiLanglands}.
\end{remark}

\subsection{The A-twist as a Limit of Holomorphic-Topological Twists of A-type} \label{A_twist_section}
Understanding the A-twisted theory will be slightly different to our calculation for the B-twist, because the A-twisted theory is no longer a cotangent theory. However, it will be a cotangent theory upon a certain compactification. In fact, we will realize that the A-twist arose as a limit of A-type deformations of holomorphic-topological twists, all of which yield cotangent theories upon such a compactification.
 
We'll begin by calculating the solutions to the equations of motion in the A-type deformations of holomorphic-topological twists by an analogous procedure to the one we used for the B-twist. A crucial difference from the previous twists is that the relevant twisting data fails to preserve the fibers of the morphism $\pi \colon \EOM_{\mr{hol}}^\alpha(X) \to \bun_G(X)$ defining the fiberwise formal algebraic gauge theory. However, for the A-twist, the fibers of the morphism $\sigma \colon \bun_G(X) \to \EOM_{\mr{hol}}^\alpha(X)$ \emph{are} preserved from the twisting data, so it's possible to define a canonical twist by applying the general construction in Lemma \ref{twist_construction_lemma} based on the general Gaitsgory--Rozenblyum correspondence in Theorem \ref{GRtheorem2}.

Let $Q_\lambda = Q_{\text{hol}} - \lambda(\alpha_2^\vee \otimes f_2^*) + (\alpha_2 \otimes e_2)$ be an A-type deformation of a holomorphic-topological supercharge as described at the end of Section \ref{supercharge_section} (so $Q_\lambda \to Q_A$ as $\lambda \to 0$).  We'll first consider a twisted theory with respect to these supercharges where $\lambda \in \CC^\times$ on a space of form $X = \Sigma_1 \times \Sigma_2$, where $\Sigma_i$ are smooth algebraic curves.  We'll have to be careful: if $\lambda \ne 0$ then the twisting data is equivariant \emph{neither} for the projection $\pi$, \emph{nor} for the section $\sigma$, so there is no chance of constructing the twist canonically from formal, linear algebraic data.  We will however describe a natural deformation of the holomorphically twisted theory, for each $\lambda$, including $\lambda = 0$ that yields a twist as defined in Section \ref{classical_field_theory_section}, guided by the superspace description of the supersymmetry action.

Recall that the twist $Q_{\lambda}$ for $\lambda \in \CC^\times$ corresponds -- in the superspace formalism -- to the vector field $\overline{\del}_{\Sigma_1} + d_{\Sigma_2} + \delby{\eps}$ on $\Sigma_1 \times \Sigma_2$.  The $Q_\lambda$-twisted theory admits a description in terms of moduli space of $\lambda$-connections, as in Definition \ref{lambda_connection}; let's describe this.  Let $U_1$ and $U_2$ be smooth complex curves; we'll describe $\EOM_\lambda(U_1 \times U_2)$, where the supercharge $Q_\lambda$ acts holomorphically in the first complex direction and topologically in the second direction. Since the twisting procedure for a supercharge $Q$ that splits as $Q' + Q''$ with $Q'$ purely of positive helicity and $Q''$ purely of negative helicity can be performed in steps without changing the result, as in Remark \ref{remark_successive_twisting}, or more concretely by performing two deformations, then obtaining a composite deformation by restricting to the diagonal $\bb A^1 \sub \bb A^1 \times \bb A^1$, we first consider the twist by the vector field $\overline{\del}_{\Sigma_1} + d_{\Sigma_2}$ and then by $\delby{\eps}$.

When we twist with respect to the holomorphic-topological supercharge $\overline{\del}_{\Sigma_1} + d_{\Sigma_2}$, it is clear from a similar line of reasoning to the one employed in Theorem \ref{B_twist_EOM_theorem} that there is a natural twisted moduli space of solutions to the equations of motion on $U_1 \times U_2$ given by the $(-1)$-shifted formal cotangent space to the moduli stack of principal $G$-bundles on $U_1 \times U_2$ together with a formal Higgs field on $U_1$ and a flat $\lambda$-connection on $\Sigma_2$, that is, the mapping space
\[T^*_{\mr{form}}[-1] \left(\underline{\mr{Map}}\left( (U_1)_{\mr{Dol}} \times (U_2)_{\lambda\text{-dR}}, BG\right)^\wedge_{\ul{\mr{Map}}(U_1 \times (U_2)_{\lambda\text{-dR}}, BG)}\right).\]
More precisely, there is a deformation of the holomorphically twisted moduli space given by the relative mapping space
\[T^*_{\mr{form}}[-1] \left(\underline{\mr{Map}}_{\bb{A}^1}\left( (U_1)_{\mr{Dol}} \times (U_2)_{\mr{Hod}}, BG \times \bb{A}^1 \right)^\wedge_{\ul{\mr{Map}}_{\bb{A}^1}(U_1 \times (U_2)_{\mr{Hod}}, BG \times \bb{A}^1 )}\right),\]
whose fiber over $\lambda$ is given by the mapping space above, and when $U_1$ and $U_2$ are both Zariski open subsets of $\CC$ this defines a twist in the sense of Definition \ref{twist_definition}.

As for the second summand, $\delby{\eps}$, this supercharge has a very natural description when $U=X$ is proper, in which case it becomes the non-vanishing vector field of degree 1, because $T^*_{\mr{form}}[-1] \underline{\mr{Map} }( X_{\mr{Dol}} , BG ) = T_{\mr{form}}[1] \underline{\mr{Map}}(X_{\mr{Dol}} ,BG)$ using the $(-2)$-shifted symplectic structure of the mapping stack from the AKSZ construction \cite[Theorem 2.5]{PTVV}.

The following proposition describes what happens when we perform the two supercharges successively.

\begin{prop} \label{HT_twist_first_description}
If $\Sigma_1$ and $\Sigma_2$ are proper smooth curves, the moduli space of solutions to the equations of motion in the $Q_\lambda$ twist of $N=4$ gauge theory is equivalent to the de Rham prestack
\[\EOM_\lambda(\Sigma_1 \times \Sigma_2) \iso \left( \underline{\mr{Map}}\left( (\Sigma_1)_{\mr{Dol}} \times (\Sigma_2)_{\lambda\text{-dR}}, BG\right)^\wedge_{\ul{\mr{Map}}(\Sigma_1 \times (\Sigma_2)_{\lambda\text{-dR}}, BG)}\right)_{\mr{dR}}.\]
\end{prop}

\begin{proof}
Since $\Sigma_1$ and $\Sigma_2$ are proper, the mapping space $\mc X = \underline{\mr{Map}}\left( (\Sigma_1)_{\mr{Dol}} \times (\Sigma_2)_{\lambda\text{-dR}}, BG\right)$ and its formal completion are $(-2)$-shifted symplectic by the AKSZ construction.  Indeed, $BG$ is naturally 2-shifted symplectic and $(\Sigma_1)_{\mr{Dol}}$ and $(\Sigma_2)_{\lambda\text{-dR}}$ are both $\OO$-compact and $\OO$-2-oriented by their fundamental classes.  Using this shifted symplectic form, we can identify $T^*[-1] \mc X$ with $T[1]\mc X $.  The result then follows by Example \ref{de_rham_stack_twist_example}.
\end{proof}

This $Q_\lambda$-twisted moduli space has another description, which realizes the compactified theory as a cotangent field theory on $\Sigma_1$. For a convenient future reference, we first note the following lemma on some useful canonical equivalences of derived stacks.

\begin{lemma} \label{mapping_stack_claims}

\begin{enumerate}
\item For a reduced scheme $Y$ and any prestack $\mc X$, there is an equivalence \[\ul{\mr{Map}}(Y, \mc X_{\mr{dR}}) \iso \ul{\mr{Map}}(Y, \mc X)_{\mr{dR}}.\]
\item For a smooth projective curve $\Sigma$ and a $k$-shifted symplectic derived stack $\mc X$, there is an equivalence \[T^*_{\mr{form}}[k-2]\ul{\mr{Map}}(\Sigma, \mc X) \iso \ul{\mr{Map}}(T[1]\Sigma, \mc X)^\wedge_{\ul{\mr{Map}}(\Sigma, \mc X)}.\]
\item For a derived Artin stack $\mc X$ locally of finite presentation, there is an equivalence
\[T^*_{\mr{form}}[k]T_{\mr{form}}[\ell]\mc X \iso T_{\mr{form}}[\ell]T_{\mr{form}}^*[k-\ell] \mc X\]
for all integers $k$ and $\ell$. 
\end{enumerate}
\end{lemma}

\begin{proof}
\begin{enumerate}
\item We analyse the $S$-points for an arbitrary cdga $S$.  There are equivalences
\begin{align*}
\ul{\mr{Map}}(Y,\mc X)_{\mr{dR}}(S) &\iso \ul{\mr{Map}}(Y,\mc X)(S^{\mr{red}})\\
&\iso \mr{Map}(Y \times \spec S^{\mr{red}} ,\mc X)\\
&\iso \mr{Map}(Y^{\mr{red}} \times \spec S^{\mr{red}} ,\mc X)\\
&\iso \mr{Map}(Y \times \spec S ,\mc X_{\mr{dR}})\\
&\iso \ul{\mr{Map}}(Y,\mc X_{\mr{dR}})(S).
\end{align*}

\item Note that both the left-hand and right-hand sides are pointed formal moduli problems over the mapping space $\ul{\mr{Map}}(\Sigma, \mc X)$, so by Theorem \ref{GRtheorem} it suffices to provide an equivalence of their shifted relative tangent bundles as sheaves of dg Lie algebras.  We observe that
\begin{align*}
\bb T_{T^*[k-2]\ul{\mr{Map}}(\Sigma, \mc X)/\ul{\mr{Map}}(\Sigma, \mc X)}[-1] &\iso \bb L_{\ul{\mr{Map}}(\Sigma, \mc X)}[k-2][-1]\\
\text{and } \bb T_{\ul{\mr{Map}}(T[1]\Sigma, \mc X)/\ul{\mr{Map}}(\Sigma, \mc X)}[-1] &\iso (\bb T_{\ul{\mr{Map}}(\Sigma, \mc X)} \to \sigma^*\bb T_{\ul{\mr{Map}}(T[1]\Sigma, \mc X)})[-1] \\
&\iso (\bb T_{\ul{\mr{Map}}(\Sigma, \mc X)} \to \sigma^*\bb L_{\ul{\mr{Map}}(T[1]\Sigma, \mc X)}[k-2])[-1]
\end{align*}
where $\sigma$ is the morphism of mapping stacks obtained by precomposition with the projection $T[1] \Sigma \to \Sigma$, and where on the last line we used the $(k-2)$-shifted symplectic structure on $\ul{\mr{Map}}(T[1]\Sigma, \mc X) \iso \ul{\mr{Map}}(\Sigma_{\mr{Dol}}, \mc X)$ obtained by the AKSZ construction. Note that the Lie algebra structure is trivial on both sides.  The two-step complexes on the right-hand side just spell out the definition of the relative tangent complex, as an object of the derived category of sheaves.

The map $\sigma$ induces a map of sheaves
\[\bb T_{\ul{\mr{Map}}(\Sigma, \mc X)}[-1] \to \sigma^* \bb T_{\ul{\mr{Map}}(T[1]\Sigma, \mc X)}[-1]\]
or dually, with a shift, a map
\[\sigma^* \bb L_{\ul{\mr{Map}}(T[1]\Sigma, \mc X)}[k-3] \to \bb L_{\ul{\mr{Map}}(\Sigma, \mc X)}[k-3].\]
We'll show that the kernel of this map is equivalent to $\bb T_{\ul{\mr{Map}}(\Sigma, \mc X)}$, and therefore the induced map between relative tangent complexes is an equivalence.  It suffices to check this claim for the fiber at each map $f \colon \Sigma \to \mc X$.  At such a fiber, the map of sheaves induced by $\sigma$ is given by the projection
\[\Gamma(\Sigma; \bb L_{\mc X} \otimes (\OO_\Sigma[2] \oplus K_\Sigma[1]))[k-3] \to \Gamma(\Sigma; \bb L_{\mc X} \otimes K_\Sigma)[k-2].\]
On the other hand, the inclusion of a fiber of $\bb T_{\ul{\mr{Map}}(\Sigma, \mc X)}[-1]$ is given by the composite
\[\Gamma(\Sigma; \bb T_{\mc X})[-1] \to \Gamma(\Sigma; \bb T_{\mc X} \otimes (\OO_\Sigma \oplus K_\Sigma[-1])[-1] \iso \Gamma(\Sigma; \bb L_{\mc X}[k] \otimes (\OO_\Sigma \oplus K_\Sigma[-1])[-1],\]
whose image is precisely the kernel of the projection, as required.  Therefore the relative tangent complexes to our two derived stacks are equivalent, so the derived stacks themselves are equivalent, as required.

\item Since both $T^*_{\mr{form}}[k]T_{\mr{form}}[\ell]\mc X$ and $T_{\mr{form}}[\ell]T^*_{\mr{form}}[k-\ell]\mc X$ define pointed formal moduli problems over $\mc X$, it suffices by Theorem \ref{GRtheorem} to prove an equivalence for the restricted shifted tangent complexes as sheaves of Lie algebras over $\mc X$.  We realize such an equivalence as the composite
\begin{align*}
\sigma^*\bb T_{T^*_{\mr{form}}[k]T_{\mr{form}}[\ell]\mc X}[-1] &\iso ((\bb T_{\mc X} \oplus \bb T_{\mc X}[\ell]) \oplus (\bb L_{\mc X} \oplus \bb L_{\mc X}[-\ell])[k])[-1] \\
&\iso (\bb T_{\mc X} \oplus \bb T_{\mc X}[\ell] \oplus \bb L_{\mc X}[k] \oplus \bb L_{\mc X}[k-\ell])[-1]\\
&\iso ((\bb T_{\mc X} \oplus \bb L_{\mc X}[k-\ell]) \oplus (\bb T_{\mc X} \oplus \bb L_{\mc X}[k-\ell])[\ell])[-1] \\
&\iso \sigma^*\bb T_{T_{\mr{form}}[\ell]T^*_{\mr{form}}[k-\ell]\mc X}[-1]
\end{align*}
of dg Lie algebra equivalences, where the Lie structure on the second line is given by the bracket on the first factor, the action of the first factor on each of the others, and the pairing between the second and fourth factors, taking values in the third factor.
\end{enumerate}
\vspace{-10pt}
\end{proof}

\begin{remark} \label{de_rham_deformation_under_equivalence_remark}
\begin{enumerate}
\item The equivalence $\ul{\mr{Map}}(Y, \mc X_{\mr{dR}}) \iso \ul{\mr{Map}}(Y, \mc X)_{\mr{dR}}$ arises as an equivalence of the full Hodge stack. For this, it is enough to observe that $\ul{\mr{Map}}(Y, T_{\mr{form}}[1] \mc X) \iso T_{\mr{form}}[1] \ul{\mr{Map}}(Y, \mc X)$ has the same relative shifted tangent complex over $\ul{\mr{Map}}(Y, \mc X)$, which is immediate.
\item The third equivalence for $\ell=1$ is also compatible with its de Rham deformation. More precisely, under the equivalence $T_{\mr{form}}[1]T^*_{\mr{form}}[-k]\mc X \iso T^*_{\mr{form}}[1-k]T_{\mr{form}}[1]\mc X$, we can transfer the natural deformation of the shifted tangent complex on the left-hand side corresponding to the family of sheaves $t \cdot \mr{id} \colon \pi^*\bb T_{T^*_{\mr{form}}[-k]\mc X} \to  \pi^*\bb T_{T^*_{\mr{form}}[-k]\mc X}$ over $\bb A^1$, where $\pi$ is the projection $T_{\mr{form}}[1]T^*_{\mr{form}}[-k]\mc X \to T^*_{\mr{form}}[-k]\mc X$, to the right-hand side.  The result is the pullback under the map $T^*_{\mr{form}}[1-k]T_{\mr{form}}[1]\mc X \to T_{\mr{form}}[1]\mc X$ of the deformation  $t \cdot \mr{id} \colon \pi'^*\bb T_{\mc X} \to  \pi'^*\bb T_{\mc X}$, where now $\pi'$ is the projection $T_{\mr{form}}[1]\mc X \to \mc X$.  Now, we can consider the formal completions of both sides of our equivalence with respect to $T^*_{\mr{form}}[-k]\mc X$ to obtain a pair of equivalent pointed formal moduli problems over $T^*_{\mr{form}}[-k]\mc X$.  By Theorem \ref{GRtheorem} these are determined by (equivalent) sheaves of dg Lie algebras over $T^*_{\mr{form}}[-k]\mc X$, and we've described equivalent 1-parameter deformations of these sheaves, and therefore of the resulting formal moduli problems under $T^*_{\mr{form}}[-k]\mc X$.  The fibers over 1 of these deformed moduli problems are given by
\[(T^*_{\mr{form}}[-k]\mc X)_{\mr{dR}} \iso T^*_{\mr{form}}[1-k](\mc X_{\mr{dR}})\]
where the latter is a formal moduli problem under $T^*[-k]\mc X$ by the composite
\[T^*_{\mr{form}}[-k]\mc X \to (T^*_{\mr{form}}[-k]\mc X)_{\mr{dR}} \to \mc X_{\mr{dR}} \iso T^*_{\mr{form}}[1-k](\mc X_{\mr{dR}}).\]
\end{enumerate}
\end{remark}

\begin{theorem} \label{HT_twist_EOM}
The moduli space of solutions to the equations of motion on the product $\Sigma_1 \times \Sigma_2$ of two smooth projective curves after applying the $Q_\lambda$-twist is equivalent to
\[\EOM_\lambda(\Sigma_1 \times \Sigma_2) \iso T^*_{\mr{form}}[-1] \underline{\mr{Map}}\left(\Sigma_1 , \Flat^\lambda_G(\Sigma_2)_{\mr{dR}} \right)\]
in a canonical way.
\end{theorem}

\begin{remark}
This statement is not contentless, despite the fact that it involves the cotangent bundle of a de Rham stack, which is necessarily trivial. Indeed, the equivalence is compatible with the deformation to the whole Hodge stack. All such statements appearing in the paper arise as specializations of equivalences of Hodge stacks.
\end{remark}

\begin{proof}
We begin with the derived stack on the right-hand side.  Since $\Flat^\lambda_G(\Sigma_2) = \ul{\mr{Map}}(\Sigma_{\lambda\text{-dR}}, BG)$ is 0-shifted symplectic by the AKSZ construction, there is an equivalence $T[1]\Flat^\lambda_G(\Sigma_2) \iso T^*[1]\Flat^\lambda_G(\Sigma_2)$, so in particular $\Flat^\lambda_G(\Sigma_2)_{\mr{Dol}} = T_{\mr{form}}[1]\Flat^\lambda_G(\Sigma_2)$ is 1-shifted symplectic. We have equivalences
\begin{align*}
T^*_{\mr{form}}[-1] \underline{\mr{Map}}\left(\Sigma_1, \Flat^\lambda_G(\Sigma_2)_{\mr{dR}} \right)&\iso T^*_{\mr{form}}[-1]\left(\ul{\mr{Map}}(\Sigma_1, \Flat_G^\lambda(\Sigma_2))_{\mr{dR}}\right) \\
&\iso \left(T^*_{\mr{form}}[-2]\ul{\mr{Map}}(\Sigma_1, \Flat_G^\lambda(\Sigma_2)) \right)_{\mr{dR}} \\
&\iso \left(\ul{\mr{Map}}(T[1]\Sigma_1, \Flat_G^\lambda(\Sigma_2))^\wedge_{\ul{\mr{Map}}(\Sigma_1, \Flat_G^\lambda(\Sigma_2))}\right)_{\mr{dR}} \\
&\iso \left( \left( \underline{\mr{Map}}\left( (\Sigma_1)_{\mr{Dol}}, \Flat_G^\lambda(\Sigma_2)\right)\right)^\wedge_{\ul{\mr{Map}}(\Sigma_1, \Flat_G^\lambda(\Sigma_2))}\right)_{\mr{dR}} \\
&\iso \left( \left( \underline{\mr{Map}}\left( (\Sigma_1)_{\mr{Dol}} \times (\Sigma_2)_{\lambda\text{-dR}}, BG\right)\right)^\wedge_{\ul{\mr{Map}}(\Sigma_1, \Flat_G^\lambda(\Sigma_2))}\right)_{\mr{dR}} \\
&= \EOM_\lambda(\Sigma_1 \times \Sigma_2),
\end{align*}
where on the first line we used Lemma \ref{mapping_stack_claims} part 1, on the second line we used Remark \ref{de_rham_deformation_under_equivalence_remark} part 2, and on the fifth line we used the adjunction $\underline{\mr{Map}}( (\Sigma_1 )_{\mr{Dol}} \times (\Sigma_2)_{\lambda\text{-dR} } , BG  ) = \underline{\mr{Map}} ( (\Sigma_1)_{\mr{Dol}} , \underline{\mr{Map}}( (\Sigma_2 )_{\lambda\text{-dR} }, BG  ) )$. Now in view of Remark \ref{de_rham_deformation_under_equivalence_remark}, one can note that the whole equivalences work at the level of Hodge stacks. 
\end{proof}

\begin{remark}
We have two apparently different-looking descriptions of our moduli space, but the point is that one can use either one. For the rest of the paper, we won't use this latter description. On the other hand, when $\lambda=0$, this theorem amounts to identifying the compactification of the A-twisted theory along $\Sigma_2$ with the A-model with target $\higgs_G(\Sigma_2)$, as expected from the physics literature. This can also be understood as an algebraization and globalization of Costello's perturbative description of the A-model in the smooth category  \cite{CostelloSH}.
\end{remark}

Let's now discuss what this assigns to objects of nonzero codimension as we did in Section \ref{B_twist_section}: \[\EOM_\lambda(\Sigma \times U) \iso \left( \underline{\mr{Map}}\left( \Sigma_{\mr{Dol}} \times U_{\lambda\text{-dR}}, BG\right)^\wedge_{\ul{\mr{Map}}(\Sigma \times U_{\lambda\text{-dR}},BG)}\right)_{\mr{dR}}\]
as in Proposition \ref{HT_twist_first_description} the assignment naturally extends to $U = S^1_B$ or $U=\pt$.

We'll describe it in a way designed to illustrate the connection with geometric Langlands. However, the argument we gave for Theorem \ref{HT_twist_EOM} no longer applies.  Instead of a $(-1)$-shifted cotangent space, we'll produce a 0-shifted cotangent space. In the A-twist, the degree of shifting comes naturally so we don't need any auxiliary step: de Rham stack can be regarded as $k$-shifted symplectic for any $k$, but that being realized as a Hodge stack over $\bb{A}^1$ determines the unique number $k$ in such a way that ensures compatibility for any $t \in \bb{A}^1$. 

\begin{prop} \label{A_twist_phase_space_prop}
The phase space $\EOM_\lambda(\Sigma \times S^1_B)$ in the $Q_\lambda$-twisted theory is equivalent to
\[T^*_{\mr{form}}(\underline{\mr{Map}}(S^1_B, \bun_G(\Sigma))_{\mr{dR}} ).\]  
In particular the result is independent of the value of $\lambda$.  The equivalence arises by taking the fiber at 1 of an equivalence of deformations, whose fiber at 0 is an equivalence 
\[\ul{\mr{Map}} \left(\Sigma_{\mr{Dol}} \times (S^1_B)_{\lambda \text{-dR}}, BG \right)_{\mr{Dol}} \iso T^*_{\mr{form}}T_{\mr{form}}[1]\underline{\mr{Map}}(S^1_B, \bun_G(\Sigma)).\]
\end{prop}

\begin{proof}
First, observe that $(S^1_B)_{\lambda \text{-dR}} \iso S^1_B$ for all $\lambda \in \CC$.  Indeed, any topological space $Y$ viewed as a derived stack has trivial tangent complex, so $(Y_B)_{\mr{Hod}} \iso Y_B \times \bb A^1$. According to Proposition \ref{HT_twist_first_description} and Lemma \ref{mapping_stack_claims} part 2 we have
\begin{align*}
\EOM_\lambda(\Sigma \times S^1_B) &\iso \left(\ul{\mr{Map}} \left(\Sigma_{\mr{Dol}} \times (S^1_B)_{\lambda \text{-dR}}, BG \right)^\wedge_{\ul{\mr{Map}} \left(\Sigma\times (S^1_B)_{\lambda \text{-dR}}, BG \right)}\right)_{\mr{dR}} \\ 
&\iso \left(\ul{\mr{Map}} (T[1]\Sigma, \Flat_G(S^1))^\wedge_{\ul{\mr{Map}} \left(\Sigma\times (S^1_B)_{\lambda \text{-dR}}, BG \right)}\right)_{\mr{dR}} \\
&\iso (T^*_{\mr{form}}[-1]\ul{\mr{Map}}(\Sigma, \Flat_G(S^1)))_{\mr{dR}}.
\end{align*}
This falls into a family of equivalences, by replacing the de Rham prestack with the Hodge prestack, whose central fiber is given by the formal completion
\begin{align*}
T_{\mr{form}}[1]\ul{\mr{Map}} \left(\Sigma_{\mr{Dol}} \times (S^1_B)_{\lambda \text{-dR}}, BG \right) &\iso T_{\mr{form}}[1]T_{\mr{form}}^*[-1]\ul{\mr{Map}}(\Sigma, \Flat_G(S^1)) \\
&\iso T_{\mr{form}}^*T_{\mr{form}}[1]\ul{\mr{Map}}(\Sigma, \Flat_G(S^1)),
\end{align*}
by Lemma \ref{mapping_stack_claims} part 3.  To conclude the proof we observe that the degree 1 symmetry of the tangent complex generating the de Rham deformation via Example \ref{de_rham_stack_twist_example} corresponds -- under the equivalence -- to the symmetry on the right-hand side deforming $T_{\mr{form}}^*T_{\mr{form}}[1]\ul{\mr{Map}}(\Sigma, \Flat_G(S^1))$ to $T_{\mr{form}}^*(\ul{\mr{Map}}(\Sigma, \Flat_G(S^1))_{\mr{dR}})$ by Remark \ref{de_rham_deformation_under_equivalence_remark} part 2.
\end{proof}

Given that the A-twist is computed by an identical procedure to the more general $\lambda$-twist, one might ask what the point is of considering the $\lambda$ family of twists at all. The claim, which we hope to return to in future work, is that in order to see more refined structures in the geometric Langlands program, it is necessary to consider such twists. The following remark provides a hint of this structure.

\begin{remark} \label{bubble_EOM_remark}
The curve $\Sigma = \bb{CP}^1$ deserves a little more attention; we'll describe an infinitesimal version of the above calculation, explicitly using the family of theories obtained by varying $\lambda$.  Instead of describing the solutions to the equations of motion on the derived stack $\bb{CP}^1 \times S^1_B$, we'll instead consider a different complex structure on a complex neighborhood of $S^2 \times S^1$.  The following construction should be thought of as informal and motivational, since we'll use complex analytic constructions that don't make sense in derived algebraic geometry.  Consider the complex manifold
\[(\CC \times \CC^\times) \bs (\{0\} \times S^1).\]
Note that there are diffeomorphisms $\CC \times \CC^\times \iso \CC \times (0,\infty) \times S^1 \simeq B^3 \times S^1$ for an open three-ball $B^3$ around $0$. Removing $\{0\} \times S^1$ from $\CC \times \CC^\times$ corresponds to removing $\{0\} \times S^1$ from $B^3 \times S^1$ on the right-hand side, yielding a diffeomorphism $(B^3 \bs \{0\}) \times S^1 \simeq (S^2 \times (-1,1) ) \times S^1$.  Thus we can think of $(\CC \times \CC^\times) \bs (\{0\} \times S^1)$ as a complex manifold thickening $S^2 \times S^1$.

\vspace{-5pt}
From Proposition \ref{HT_twist_first_description}, the space of solutions to the equations of motion is obtained by applying the de Rham space construction to the moduli space of $G$-bundles on $ (\CC \times \CC^\times) \bs ( \{0\} \times S^1)$ with a Higgs field on $\CC$ and a flat $\lambda$-connection on $\CC^\times$. Let us denote the two connected components of $\CC^\times \bs S^1$ by $A_{\mr{in}}$ and $A_{\mr{out}}$. Note that a $G$-bundle $P$ on $(\CC \times \CC^\times) \setminus ( \{0\} \times S^1  )$ is equivalent to the data of a triple $(P', \phi_{\mr{in} } , \phi_{\mr{out}} )$, where $P'$ is the restriction of $P$ to $\CC^\times  \times \CC^\times$, $\phi_{\mr{in}}$ is the extension of $P'|_{\CC^\times \times A_{\mr{in}} }$ to $\CC \times A_{\mr{in}}$, and $\phi_{\mr{out}}$ is the extension of $P'|_{\CC^\times \times A_{\mr{out}} }$ to $\CC \times A_{\mr{out}}$. Note that ignoring the annular factor we would obtain a $G$-bundle on a ``bubbled'' plane $B := \CC \amalg_{\CC^\times} \CC$ made by gluing the two planes along $\CC^\times$.

\vspace{-5pt}
Then we can describe the moduli space of solutions to the equations of motion on $(\CC \times \CC^\times) \bs (\{0\} \times S^1)$ as a datum $(P',\phi_{\mr{in}},\phi_{\mr{out} })$ of this form, together with a Higgs field and a flat $\lambda$-connection in the two complex directions. Since we have a flat $\lambda$-connection in the $\CC^\times$-direction throughout, we can understand the space of germs of solutions to the equations of motion near $S^2 \times S^1$ as the de Rham stack of $\underline{\mr{Map}}(S^1_B, \higgs^{\mr{bos}}_G(B))$. It is essential here to have $\lambda \neq 0$: otherwise we cannot simply describe the moduli spaces in a way that depend only on the topology of $\CC^\times$, and not its algebraic structure. 

\vspace{-5pt}
Finally, we can replace $\CC$ by the formal disk $\bb{D}$.  One then obtains as the space of solutions to the equations of motion 
\[\EOM(\bb B \times S^1_B) \iso T_{\mr{form}}^*(\underline{\mr{Map}} (S^1_B, \bun_G(\bb{B}) )_{\mr{dR}} )\]
where $\bb B$ is the ``formal bubble'' $\bb{B}:= \bb{D} \amalg_{\bb {D^\times} } \bb{D}$.  The space of $G$-bundles on the formal bubble $\bb B$ is a familiar space in geometric representation theory: the quotient of the \emph{affine Grassmannian} $\mc G \mr r_G$ by the arc group $G(\CC[[t]])$.  We'll investigate the action of a quantization of this moduli space $\EOM(\bb B \times S^1_B)$ on a quantization of $\EOM(\Sigma \times S^1)$ for general surfaces $\Sigma$, inherited from the geometric structure of the bases of these cotangent spaces in future work.
\end{remark}

To conclude this section, we'd also like to understand germs of solutions to the equations of motion near manifolds of codimension 2.
\begin{prop}
$\EOM_\lambda(\Sigma \times \CC) \cong T_{\mr{form}}^*[1](\bun_G(\Sigma)_{\mr{dR}})$.
The equivalence arises as the fiber at 1 of an equivalence of deformations, whose fiber over 0 is
\[\ul{\mr{Map}}(\Sigma_{\mr{Dol}} \times \pt_{\lambda\text{-dR}}, BG)_{\mr{Dol}} \iso T_{\mr{form}}^*[1]T_{\mr{form}}[1]\bun_G(\Sigma)^\wedge_{T_{\mr{form}}^*\bun_G(\Sigma)}.\]
\end{prop}

\begin{proof}
Lemma \ref{mapping_stack_claims} provides an equivalence
\begin{align*}
\EOM_\lambda(\Sigma) &\iso \left(\ul{\mr{Map}}(\Sigma_{\mr{Dol}} \times \pt_{\lambda\text{-dR}}, BG)^\wedge_{\ul{\mr{Map}}(\Sigma \times \pt_{\lambda\text{-dR}},BG)}\right)_{\mr{dR}} \\
&\iso \left(\ul{\mr{Map}}(T[1]\Sigma, BG)^\wedge_{\ul{\mr{Map}}(\Sigma \times \pt_{\lambda\text{-dR}},BG)}\right)_{\mr{dR}} \\
&\iso (T_{\mr{form}}^*\ul{\mr{Map}}(\Sigma, BG))_{\mr{dR}}
\end{align*}
using the 2-shifted symplectic structure on $BG$.  As in the proof of Proposition \ref{A_twist_phase_space_prop}, this equivalence arises as the generic fiber of a natural deformation, whose fiber over zero is 
\begin{align*}
T_{\mr{form}}[1]\ul{\mr{Map}}(\Sigma_{\mr{Dol}} \times \pt_{\lambda\text{-dR}}, BG) &\iso T_{\mr{form}}[1]T_{\mr{form}}^*\ul{\mr{Map}}(\Sigma, BG)) \\
&\iso T_{\mr{form}}^*[1]T_{\mr{form}}[1]\ul{\mr{Map}}(\Sigma, BG))\\
&\iso T_{\mr{form}}^*[1]T_{\mr{form}}[1]\bun_G(\Sigma).
\end{align*}
Again, as in Lemma \ref{mapping_stack_claims} we observe by Remark \ref{de_rham_deformation_under_equivalence_remark} part 2 that the degree 1 symmetry of the tangent complex generating the de Rham deformation corresponds to the symmetry on the right-hand side deforming the Dolbeault stack to the Hodge prestack, thus providing an equivalence of deformations, as required. 
\end{proof}

\begin{remark}
As above, if $\Sigma=\bb{CP}^1$ we have the option to perform an infinitesimal construction.  By the same reasoning as in Remark \ref{bubble_EOM_remark} one can choose as a thickening $(\CC \times \CC) \bs (\{0\} \times I)$, where $I$ is the imaginary axis in the second factor: there is a diffeomorphism
\[S^2 \times (-1,1) \times I \simeq (B^3 \bs \{0\} ) \times I \simeq (\CC \times \CC) \bs (\{0\} \times I )\] 
which we again think of as a choice of complex thickening of $S^2$.  Running through the same calculation as in Remark \ref{bubble_EOM_remark} we end up with the moduli space of germs of solutions to the equations of motion $T_{\mr{form}}^*[1] (\bun_G(\bb{B}) _{\mr{dR}})$.  This will naturally appear in an interpretation of geometric Satake as arising from line operators.
\end{remark}

\appendix
\appendixpage
\addappheadtotoc
\section{Supersymmetry Algebras} \label{SUSY_appendix}
We'll begin by setting up some general language for describing supersymmetry algebras before describing the particular cases we're interested in (supersymmetry in 2, 4 and 10 dimensions).  The notion of \emph{twisting} supersymmetry algebras and supersymmetric field theories makes sense in any dimension and signature.  The material in this section is standard.  Proofs can be found for instance in \cite{Deligne} or \cite{Varadarajan}.

Let $p$ and $q$ be non-negative integers, and let $n = p+q$.  We'll describe supersymmetry algebras in pseudo-Riemannian signature $(p,q)$.  The main pieces of data that we'll need to specify are a spin representation and a spin-invariant vector-valued pairing on this representation.
\begin{definition}
A (real or complex) representation of the Lie algebra $\so(p,q)$ is \emph{spinorial} if it extends to a module for the even (real or complex) Clifford algebra $\mr {Cl}^+(p,q)$.
\end{definition}

There is a complete classification of spinorial $\so(p,q)$ representations.
\begin{prop} \label{spin_rep_classification_prop}
Over $\CC$, $\so(p,q)$ either has a \emph{unique} non-trivial irreducible representation $S$ of dimension $2^{\frac{n-1}2}$ if $p+q$ is odd, or has two distinct non-trivial irreducible representations $S_\pm$ each of dimension $2^{\frac n2 -1}$ if $p+q$ is even.  In the latter case we write $S$ for $S_+ \oplus S_-$.  We call $S$ the space of \emph{Dirac spinors} and $S_{\pm}$ the spaces of positive and negative helicity \emph{Weyl spinors}. 

Over $\RR$, the representation $S$ is the complexification of a real representation $S_\RR$ when $p-q \equiv 0, 1 \text{ or } 7 \!\mod 8$.  The representations $S_\pm$ are the complexifications of real representations $S_{\RR \pm}$ when $p-q \equiv 0 \!\mod 8$.  We call $S_\RR$ the space of \emph{Majorana spinors} and $S_{\RR \pm}$ spaces of \emph{Majorana--Weyl spinors}.  When instead $p-q \equiv 2$ or $6 \!\mod 8$ the representation $S_+ \oplus S_+^*$ \footnote{A real form for $S_- \oplus S_-^*$ would also work; the two agree up to complex conjugation.} is the complexification of a real representation, which we also denote by $S_\RR$ and refer to as the space of Majorana spinors.
\end{prop}

We write $V_\RR$ for the $n$-dimensional vector representation $\RR^{p,q}$ of $\so(p,q)$, and $V_\CC$ for its complexification.  The second component necessary to define supersymmetry algebras is the following.
\begin{definition}
A \emph{pairing} on a spin representation $\Sigma$ is a symmetric $\so(p,q)$-equivariant linear map
\[\Gamma \colon \Sigma \otimes \Sigma \to V_k\]
where $k = \RR$ or $\CC$.
\end{definition}

Again, we have a good control over the existence and uniqueness of such pairings.  We can construct them using the Clifford multiplication, and duality properties of the spinors.
\begin{prop}
Over $\CC$ there exist unique pairings (up to rescaling)
\begin{align*}
&\Gamma \colon S \otimes S \to V_\CC \ \ \text{ if } n \equiv 1, 3, 5, \text{ or } 7 \! \mod 8 \\
&\Gamma \colon S_\pm \otimes S_\pm \to V_\CC \ \text{ if } n \equiv 2 \text{ or } 6 \! \mod 8 \\
&\Gamma \colon S_\pm \otimes S_\mp \to V_\CC \ \text{ if } n \equiv 0 \text{ or } 4 \! \mod 8.
\end{align*}
These pairings descend to give unique $V_\RR$-valued pairings on the Majorana or Majorana--Weyl spinors whenever they exist.
\end{prop}

We can use this to describe pairings on more general spinorial representations.  There are pairings on the representation $S \otimes W$ -- where $W$ is a finite-dimensional vector space -- for each element of $\gl(W)$.  If we also require our pairings to be non-degenerate then there is a unique pairing up to $\so(p,q)$-equivariant isomorphism.

Now, we can define the supersymmetry algebra associated to this data.
\begin{definition}
The (real) \emph{supertranslation algebra} associated to a spinorial representation $\Sigma$ of $\so(p,q)$ is the super Lie algebra
\[T = V_\RR \oplus \Pi(\Sigma)\]
where the only bracket is the pairing $\Gamma \colon \Sigma \otimes \Sigma \to V_\RR$.  The (real) \emph{super Poincar\'e algebra} is the super Lie algebra
\[P = (\so(p,q) \ltimes V_\RR) \oplus \Pi(\Sigma)\]
where there are brackets given by $\Gamma$, by the internal bracket on the even piece and by the action of $\so(p,q)$ on $\Sigma$.  We define \emph{complex} supertranslation and super Poincar\'e algebras analogously, with $V_\RR$ replaced by $V_\CC$, and with $\Sigma$ a complex spinorial representation.
\end{definition}
To complete the definition, we need one more piece of data, namely a subalgebra of the R-symmetry algebra.
\begin{definition}
The \emph{R-symmetry algebra} associated to a supertranslation algebra is the algebra of outer automorphisms acting trivially on the bosonic piece.  Given a subalgebra $\gg_R$ of the R-symmetry algebra, the (real) \emph{supersymmetry algebra} is the super Lie algebra
\[\mc A = (\so(p,q) \ltimes V_\RR) \oplus \gg_R \oplus \Pi(\Sigma)\]
with brackets as before, plus the action of $\gg_R$ on $\Sigma$.  The complexified supersymmetry algebra is defined analogously.
\end{definition}

When $\Sigma = S^N$, we say there are \emph{$N$ supersymmetries}.  When $\Sigma = S_+^{N_1} \oplus S_-^{N_2}$ we say there are \emph{$(N_1, N_2)$ supersymmetries}.  If we impose the condition that the pairing $\Gamma$ is non-degenerate then we can only have $N_1 \ne N_2$ when $n \equiv 2$ or $6 \mod 8$ in the complex case, or when  $n \equiv 2$ or $6 \mod 8$ \emph{and} $p \equiv q \mod 8$ in the real case. 

\begin{definition}
A \emph{supersymmetric field theory} on $\RR^{p,q}$ is a field theory on $\RR^{p,q}$ equipped with an action of the complexified supersymmetry algebra extending the natural action of the complexified Poincar\'e algebra $\so(p,q) \ltimes V_\CC$.
\end{definition}

\begin{example}[Dimension 4]
The principal theories that we're interested in this paper are supersymmetric theories in dimension 4.  In this and the subsequent examples we'll be most interested in the complexified supersymmetry algebra, so the choice of signature won't be too important.  For specificity we'll work in Euclidean signature $(4,0)$.  Recall that we have an isomorphism of groups, $\mr{Spin}(4) \cong \SU(2)_+ \times \SU(2)_-$. Let $S_+$ and $S_-$ be the complex 2-dimensional defining representations of the two copies of $\SU(2)$, respectively. Let $V_\RR$ be the real 4-dimensional vector representation of $\mr{Spin}(4)$. If we define $V_\CC := V_\RR \otimes_\RR \CC$, then there is an isomorphism $\Gamma \colon  S_+ \otimes S_- \stackrel{\cong}{\longrightarrow} V_\CC $ as complex $\mr{Spin}(4)$-representations.

Let $W$ be a finite-dimensional complex vector space.  There is a natural non-degenerate pairing on the spinorial representation $(S_+ \otimes W) \oplus (S_- \otimes W^*)$, given by the isomorphism $\Gamma$ and the canonical pairing $W \otimes W^* \to \CC$.  The \emph{super-translation algebra} associated to $W$ is the super Lie algebra
\[T^W = V_\CC \oplus \Pi \left(S_+ \otimes W \oplus S_- \otimes W^* \right),\]
with Lie bracket given by this pairing.

One can compute that the R-symmetry algebra for this representation and pairing is the algebra $\gl(W)$ acting on $W$ and $W^*$ by the fundamental and anti-fundamental representations respectively.  Given a subalgebra $\gg_R \sub \gl(W)$, there is an associated \emph{supersymmetry algebra}
\[\mc A^W = (\so(4;\CC) \oplus \gg_R) \ltimes T^W.\]
If $\dim W = k$, we also denote this algebra by $\mc A^{N=k}$.  We'll be particularly interested in the case where $\dim W = 4$ and $\gg_R = \sl(4)$.  As we'll see, this is the supersymmetry algebra that will act on $N=4$ supersymmetric gauge theories.
\end{example}

\begin{example}[Dimension 2]
Two-dimensional theories will arise for us as dimensional reductions of 4d theories along a Riemann surface.  Again, since we're most interested in the complexified supersymmetry algebra we'll not be too concerned about the choice of signature, but it is worth remarking that the case of Lorentzian signature is special due to the existence of Majorana--Weyl spinors.  We have an isomorphism $\spin(2) \iso \mr U(1)$.  Let $S_\pm$ be the complex 1-dimensional representations of the circle of weight $\pm 1$.  The vector representation of $\Spin(2)$ corresponds to the weight two representation of $\mr U(1)$, so there are natural pairings $\Gamma \colon S_\pm \otimes S_\pm \to V_\CC$ (using a canonical isomorphism between $V_\CC$ and its dual).  

Let $W_+$ and $W_-$ be finite-dimensional complex vector spaces, and choose inner products $W_\pm \otimes W_\pm \to \CC$.  Combining this with the pairing above yields a pairing $\Gamma$ on the spinorial representation $(S_+ \otimes W_+) \oplus (S_- \otimes W_-)$, and thus a super Poincar\'e algebra
\[P_2^{(W_+, W_-)} = (\so(2;\CC) \ltimes V_\CC) \oplus \Pi\left((S_+ \otimes W_+) \oplus (S_- \otimes W_-)\right).\]
The R-symmetry algebra associated to this super Poincar\'e algebra is $\gl(W_+) \oplus \gl(W_-)$, and associated to a subalgebra $\gg_R$ of this algebra we produce a \emph{supersymmetry algebra}
\[\mc A_2^{(W_+, W_-)} = (\so(2;\CC) \ltimes V_\CC) \oplus \gg_R \oplus \Pi\left((S_+ \otimes W_+) \oplus (S_- \otimes W_-)\right).\]
If $\dim W_+ = N_1$ and $\dim W_- = N_2$, we say we have $(N_1, N_2)$ supersymmetries, and write $\mc A_2^{(N_1, N_2)}$.

Let's describe dimensional reduction from 4 to 2 dimensions (for the complexified algebra, though we could also investigate the real case in Riemannian or Lorentzian signature).  That is, take $\CC^2 \sub \CC^4$, and consider the subalgebra of the complex infinitesimal isometries $\so(4;\CC) \ltimes \CC^4$ mapping this subspace to itself, which has form $(\so(2;\CC) \ltimes \CC^2) \oplus \so(2;\CC)$.  Let $S_+$ and $S_-$ be the spaces of 4d Weyl spinors.  As modules for this subalgebra, the first $\so(2;\CC)$ acts with weights $(\pm 1, \mp 1)$ on $S_\pm$ respectively, and the second $\so(2;\CC)$ acts with weight $(\pm 1,\pm 1)$ on $S_\pm$.  Thus the $N=k$ super Poincar\'e algebra in dimension 4 naturally dimensionally reduces to the $N = (2k,2k)$ supersymmetry algebra in dimension 2, with R-symmetry group $\so(2;\CC) \iso \gl(1;\CC)$. 
\end{example}

\begin{example}[Dimension 10]
There is a supersymmetric gauge theory in dimension 10 which is ``universal'' in the sense that a range of supersymmetric gauge theories that are studied in lower dimensions arise from it by a combination of dimensional reduction and restriction of scalars \cite{ABDHN}.  We'll focus on the case of minimal supersymmetry, i.e. $N = (1,0)$, describe the Majorana--Weyl spinor representations in signature $(1,9)$, then describe the complexification.

Abstractly, the classification \ref{spin_rep_classification_prop} tells us to expect a pair of mutually dual irreducible spinorial representations of $\so(1,9)$ over the real numbers, each of dimension 16.  We can actually describe these representations very concretely; the details are described by Deligne in \cite[Chapter 6]{Deligne}. 

It suffices to construct a non-trivial 32-dimensional module for the algebra $\mr{Cl}(V,Q)$, where $V$ is 10-dimensional, and $Q$ is a quadratic form of signature (1,9).  Concretely, we'll set $V = \bb O \oplus H$ with $\bb O$ 8-dimensional and $H = \langle e,f\rangle$ 2-dimensional, and we set
\[Q(\omega + ae + bf) = \omega \cdot \ol \omega - ab\]
where $\omega \cdot \ol \omega$ is the octonion norm-squared.  Let $S^\RR_{10} = (\bb O^2) \oplus (\bb O^2)$ be a 32-dimensional real vector space.  We must describe a Clifford multiplication $\rho \colon V \otimes S^\RR_{10} \to S^\RR_{10}$ making $S^\RR_{10}$ into a module for $\mr{Cl}(V,Q)$.  This is concretely given by
\[\rho \colon \bb O \oplus H \to \eend(S^\RR_{10})\]
where
\begin{align*}
 \rho(\omega) &= \pmat{ \pmat{0&m_\omega\\m_\omega&0} & 0 \\ 0 & \pmat{0&m_\omega\\m_\omega&0}} \text{ for } \omega \in \bb O,\  m_\omega(\alpha) = \ol \omega \cdot \ol \alpha \\
 \rho(e) &= \pmat{\pmat{0&1\\0&0}&0\\0&\pmat{0&1\\0&0}} \quad
 \text{and } \ \ \rho(f) = \pmat{\pmat{0&0\\-1&0}&0\\0&\pmat{0&0\\-1&0}}. \\
\end{align*}
One can check that this gives a well-defined Clifford multiplication, and thus defines a 32-dimensional real spin representation which splits as a sum of two 16-dimensional representations of the even part of the Clifford algebra: call them $S^\RR_{10+}$, spanned by the first and third components of $\bb O^4$, and $S^\RR_{10-}$ spanned by the second and forth. There is also the induced pairing $\Gamma \colon S^\RR_{10\pm} \otimes S^\RR_{10\pm} \to V$, which one checks is given on $S^\RR_{10+}$ and $S^\RR_{10-}$ respectively by
\begin{align*}
\Gamma((\alpha_1,\alpha_2), (\beta_1,\beta_2)) &= \ol \alpha_1 \cdot \ol \beta_1 + \ol \alpha_2 \cdot \ol \beta_2 - \tr(\alpha_1 \cdot \ol \beta_1 + \alpha_2 \cdot \ol \beta_2)f \\
\text{and } \Gamma((\alpha_1,\alpha_2), (\beta_1,\beta_2)) &= \ol \alpha_1 \cdot \ol \beta_1 + \ol \alpha_2 \cdot \ol \beta_2 + \tr(\alpha_1 \cdot \ol \beta_1 + \alpha_2 \cdot \ol \beta_2)e
\end{align*}
where $\tr(\alpha) = \alpha + \ol \alpha$ is the octonionic reduced trace, and where the calculation is done using the identity $\langle \Gamma(s,t),v\rangle = (\rho(v)s, t)$ for spinors $s,t$ and vectors $v$. This now gives us a complete description of the supersymmetry algebra in 10-dimensions: it is given by
\[(\so(1,9) \ltimes \RR^{1,9}) \oplus \Pi(S^\RR_{10+})\]
with brackets given by the internal bracket on $\so(1,9)$, the action of $\so(1,9)$ on the translations, the action of $\so(1,9)$ on the supersymmetries, and the pairing $\Gamma \colon S^\RR_{10+} \otimes S^\RR_{10+} \to \RR^{1,9}$.

Finally, we can complexify the supersymmetry algebra to obtain a superalgebra of form
\[(\so(10;\CC) \ltimes \CC^{10}) \oplus \Pi(S_{10+}).\]
The complexification $S_{10+} = S^\RR_{10+} \otimes \CC$ is a 16-complex dimensional Weyl spinor representation of $\so(10;\CC)$.  Clifford theory says that the complexification $\so(10;\CC)$ embeds in the (even part of the) Clifford algebra $\mr{Cl}^+_{10} \iso \mat_{16}(\CC) \oplus \mat_{16}(\CC)$ as the elements of spinor norm one.  The Weyl spinors are the fundamental representation of the first matrix algebra factor.

More concretely, we write $S_{10+}$ as $\bb O^2 \oplus i\bb O^2$ where $\bb O$ is a 4-complex dimensional vector space.  We write $\CC^{10}$ as $\bb O \oplus i\bb O \oplus \CC\langle e,f \rangle$.  The Clifford multiplication is then given by 

\begin{align*}
\rho(\omega) &= \pmat{ \pmat{0 & m_\omega \\ m_\omega & 0} &0 \\ 0 & \pmat{0 & m_\omega \\ m_\omega & 0}}, \quad \rho(i\omega) = \pmat{0&\pmat{0 & m_\omega \\ m_\omega & 0}\\\pmat{0 & m_\omega \\ m_\omega & 0}&0} \ \text{ for } \omega \in \bb O\\
 \rho(e) &= \pmat{\pmat{0&1\\0&0}&0\\0&\pmat{0&1\\0&0}} \quad \text{and} \quad \rho(f) = \pmat{\pmat{0 & 0 \\ -1 & 0}&0 \\ 0&\pmat{0 & 0 \\-1 & 0} }.
\end{align*}

This complexified algebra dimensionally reduces to recover the $N=4$ supersymmetry algebra discussed above in four-dimensions.  We choose an embedding $\CC^4 \inj \CC^{10}$ and consider the subalgebra of the supersymmetry algebra fixing this subspace.  The bosonic piece has the form $\so(4;\CC) \ltimes V_\CC \oplus \sl(4;\CC)$, where the $\sl(4;\CC)$ fixes the subspace pointwise (and arises from complexification of $\so(6) \iso \su(4)$).  We must check that the action of $\so(6;\CC) \oplus \sl(4;\CC)$ on the 16-complex-dimensional space of spinors recovers the space $S_+ \otimes W \oplus S_- \otimes W^*$ that we expect.  We can do this by looking at the actions of the two summands separately, using that the action is still spinorial, and the fact that it arose as complexification of a representation for the (Lorentzian) real form.

Firstly, $\sl(4;\CC)$ has two Weyl spinor representations, the fundamental $W$ and the anti-fundamental $W^*$, and we must have equal numbers of each (since the complexification of the Majorana spin representation is their sum).  The modules $S_{10+}$ has no $\so(6;\CC)$-fixed points, so there are no trivial factors and $S_{10+} \otimes_\RR \CC \iso (W \oplus W^*) \otimes_\CC \CC^2$.  Secondly, $\so(4;\CC)$ has two Weyl spinor representations $S_+$ and $S_-$.  By the same argument we have equal numbers of each and there are no trivial summands, so $S_{10+} \otimes_\RR \CC \iso (S_+ \oplus S_-) \otimes_\CC \CC^4$.  Finally, to describe the relationship between these two actions we observe that the actions commute and complexify a real Lie algebra action.
\end{example}

\section{Lie Algebras and Deformation Theory} \label{linfty_appendix}
For motivation and reference, we've included the fundamental definitions and results on sheaves of Lie algebras and deformation theory. None of this material is original, and most of the results in the smooth category context be found in \cite{CostelloSH}, \cite{GG} and appendix A of \cite{CostelloGwilliam1}.  The derived deformation theoretic results we reference are due to Hinich \cite{Hinich} and Getzler \cite{Getzler}, or in a more homotopical setting to Lurie \cite{DAGX} and Hennion \cite{Hennionthesis}.

As we work in the setting of $\infty$-categories and the two operads Lie and $L_\infty$ are homotopy equivalent we are free to use the languages of Lie and $L_\infty$-algebras interchangeably, mainly choosing our terminology in order to be more compatible with the literature for the appropriate context.

\begin{definition}
A \emph{curved $L_\infty$ algebra} over a cdga $R$ with respect to an ideal $I$ is a locally free graded $R^\natural$ module $L$ equipped with a degree 1 differential 
\[d \colon \symc_{R^\natural}(L^\vee[-1]) \to \symc_{R^\natural}(L^\vee[-1])\]
making $\symc_{R^\natural}(L^\vee[-1])$ into a dg-module over $R$, such that $d$ vanishes on $\sym^0$ modulo the ideal $I$.  We denote $\symc_{R^\natural}(L^\vee[-1])$ by $C^\bullet(L)$ and call it the \emph{Chevalley--Eilenberg algebra} of $L$. 
\end{definition}

By taking the Taylor coefficients of the differential $d$ we obtain a sequence of degree 0 graded anti-symmetric operations $\ell_n \colon (\wedge^n L)[n-2] \to L$, dual to the composite
\[L^\vee[-1] \inj C^\bullet(L) \overset d \to C^\bullet(L) \surj \sym^n(L^\vee[-1])\]
which satisfy higher analogues of the Jacobi identities, recovering a more classical definition of a (curved) $L_\infty$ algebra.  One way of thinking about our definition is that Lie algebras are Koszul dual to commutative algebras, so defining the Lie algebra structure on $L$ is equivalent to defining a commutative dga structure on its Koszul dual $C^\bullet(L)$. 

We'll want to study versions of $L_\infty$ algebras varying over a topological space.  This will be useful for perturbative field theory, where an $L_\infty$ algebra describes the deformations of a particular solution to the equations of motion on an open set $U$ in spacetime, in order to describe the relationship between these solutions on different open sets.

\begin{definition} \label{local_L_infty_definition}
A \emph{local $L_\infty$ algebra} over a manifold $M$ is a cochain complex of vector bundles $L$ over $M$ such that the sheaf of sections is given the structure of a sheaf of $L_\infty$ algebras where the operations $\ell_n$ are polydifferential operators.

\vspace{-5pt}
If $G$ is an algebraic supergroup, a $G$-action on a local $L_\infty$ algebra $L$ is a $C^\bullet(G)$-module structure on $L(U)$ for each open set $U \sub X$ making $L$ into a sheaf of curved $L_\infty$ algebra over $C^\bullet(G)$ relative to the ideal $C^{>0}(G)$.  Here $C^\bullet(G)$ denotes the complex where $C^i(G) = \OO(G^i)$, with the usual differential using the group structure.  One similarly defines a $\gg$-action for a super Lie algebra $\gg$ to be a local module structure on each open set for the Chevalley--Eilenberg complex $C^\bullet(\gg)$.
\end{definition}

The perturbative definition of a classical field theory used by Costello in \cite{CostelloSH} builds on the following definition capturing local geometry of a given space. The idea is that in algebraic geometry, one is able to investigate formal neighborhoods of a point by only considering local Artinian algebras.

\begin{definition}
A \emph{formal derived moduli problem} is a functor $F$ from the category $\art_{\mr{dg}}^{\leq 0}$ of differential graded Artinian algebras cohomologically in degrees $\leq 0$ to the category $\sset$ of simplicial sets satisfying the following conditions:
\vspace{-10pt}
\begin{itemize}
\item the space $F(\CC)$ is contractible.
\item If $A \rightarrow B$ and $A'\rightarrow B$ are morphisms in $\art_{\mr{dg}}^{\leq 0}$ which are surjections on $H^0$, then the induced map $F(A\times_B A') \rightarrow F(A)\times_{F(B)}F(A')$ is a homotopy equivalence.
\end{itemize}
\vspace{-5pt}
\end{definition}

Note that the second condition ensures the ability to glue $\spec A$ and $\spec A'$ along $\spec B$ whenever we have closed embeddings at the classical level.

For example, given a point $p \in X =\spec R $ for $R \in \mr{cdga}^{\leq 0}$, or a maximal ideal $\mathfrak{m} \subset R$, the functor $X_p \colon \art_{\mr{dg}}^{\leq 0} \rightarrow \sset $ defined by \[(A,\mathfrak{m}_A) \mapsto ( \{ \phi \colon R \rightarrow A \otimes \Omega^\bullet(\Delta^n) \mid \phi(\mathfrak{m} ) =  \mathfrak{m}_A \otimes \Omega^\bullet(\Delta^n)  \})_{n \in \Delta } \] is a formal moduli problem. Geometrically $X_p(A)$ encodes the data of infinitesimal extension of $p$ via $A$.

The most important tool we are going to take advantage of in order to understand formal moduli problems is the Maurer--Cartan functor.

\begin{definition}
Let $L$ be an $L_\infty$ algebra. The \emph{Maurer--Cartan functor} $\mr{MC}_L \colon \art_{\mr{dg}}^{\leq 0} \rightarrow \sset$ is defined to be the functor given by $(R,\mathfrak{m}) \mapsto \mr{MC}_L(R)$, where the simplicial set $\mr{MC}_L(R)$ has as $n$-simplices elements $\alpha \in L \otimes \mathfrak{m} \otimes \Omega^\bullet(\Delta^n)$ of cohomological degree 1, which satisfy the Maurer--Cartan equation
\[\sum_{n \ge 0} \frac 1{n!}\ell_n(\alpha^{\otimes n}) = 0.\] 
\end{definition}

This is not manifestly a homotopy invariant notion, and thus not manifestly well-defined.  However, there is an equivalent rephrasing of the Maurer--Cartan functor that \emph{is} manifestly homotopy invariant.

\begin{prop}
$\mr{Hom}_{\cdga_*}(C^\bullet(L) , R) = \mr{MC}_L(R)$ for $R \in \art_{\mr{dg}}^{\leq 0}$.
\end{prop}

A proof of this fact appears in Section 2.3 of Lurie \cite{DAGX}; as we've phrased it it's implied by his Theorem 2.3.1.

\begin{theorem}[{\cite[2.0.2]{DAGX}}]
The Maurer--Cartan functor provides an equivalence of categories
\[\mr{MC} \colon \{L_\infty \text{ algebras}\} \to \{\text{formal derived moduli problems}\}\]
with quasi-inverse given by taking the $(-1)$-shifted tangent complex equipped with a canonical $L_\infty$ structure.
\end{theorem}

We sometimes write $BL$ for the formal moduli problem $\mr{MC}_L$. Then the theorem in particular says the following
\vspace{-10pt}
\begin{itemize}
\item There is an equivalence $\bb{T}_0[-1] BL \iso L$.
\item Every formal derived pointed moduli problem $X$ can be realized as $B L_X $ for some $L_\infty$ algebra $L_X$, in the sense that the formal derived moduli problem describing maps into $X$ is equivalent to the formal moduli problem $\mr{MC}_{L_X}$.
\end{itemize}
\vspace{-10pt}

The proposition allows one to think of $C^\bullet(L)$ as the structure sheaf of the formal moduli problem $B L$. Note that $C^\bullet(L)$ is in general not an object of the category $\mr{cdga}^{\leq 0}$, having stacky nature.

For our purpose, it is important to understand mapping stacks in terms of an $L_\infty$-algebra.

\begin{lemma}
Let $L$ be an $L_\infty$-algebra and $A$ be an object of $\art_{\mr{dg}}^{\leq 0}$. Then $L \otimes A$ is the $L_\infty$-algebra governing the deformations of the constant map $\spec A \rightarrow BL$.
\end{lemma}

We only sketch the proof for the 0-simplex to give an idea.

\begin{proof}[Proof sketch]
If $B$ is another Artinian algebra, then $\alpha \in MC_{L \otimes A}(B)[0]$ is an element $\alpha \in (L \otimes A \otimes \mathfrak{m}_B)^1$ satisfying Maurer--Cartan equation. Since the maximal ideal of $A \otimes B$ is $\mathfrak{m}_A \otimes B + A \otimes \mathfrak{m}_B$, from $MC_{L \otimes A}(B) \subset MC_L(A \otimes B)$, an element $\alpha \in MC_{L \otimes A }( B)$ can be characterized as an element of $MC_L(A \otimes B)$ which vanishes modulo $\mathfrak{m}_A$. Hence, geometrically, $MC_{L \otimes A}(B)$ represents families of maps $\spec A \rightarrow BL$ parametrized by $\spec B$ which are constant at the unique geometric point $\spec \mathbb{C} \in \spec A$.
\end{proof}

In other words, for the mapping stack $\underline{\mr{Map}}(X,Y)$, its formal derived moduli problem at $f$ is controlled by the $L_\infty$-algebra $\Gamma ( X,   f^* L_Y)$. 

The main construction we are using in the paper is in an algebraic setting.

\begin{theorem}\cite[4.2.0.1]{Hennionthesis}
If $\mc X$ is a derived Artin stack locally of finite presentation, then its shifted tangent complex $\mathbb{T}_{\mc X}[-1]$ is a Lie algebra object of $\qcoh(\mc X)$.
\end{theorem}


\bibliographystyle{alpha}
\bibliography{GL-KW}

\textsc{Institut des Hautes \'Etudes Scientifiques}\\
\textsc{35 Route de Chartres, Bures-sur-Yvette, 91440, France}\\
\texttt{celliott@ihes.fr}\\
\vspace{5pt}\\
\textsc{Department of Mathematics, Yale University}\\
\textsc{10 Hillhouse Ave., New Haven, CT 06511, USA} \\
\texttt{philsang.yoo@yale.edu}
\end{document}